\documentclass[11pt]{article}
\usepackage[utf8]{inputenc}
\usepackage[T1]{fontenc}
\usepackage[sc]{mathpazo}
\usepackage{microtype}
\usepackage{amsmath,xcolor}
\usepackage{amssymb}
\usepackage{amsthm}
\usepackage{bm,float}
\usepackage{dsfont}
\usepackage{authblk}
\usepackage{fullpage,caption,wrapfig}
\usepackage{comment}
\usepackage{mathtools}
\usepackage[shortlabels]{enumitem}
\usepackage{complexity}
\usepackage[backend=bibtex, style=alphabetic, backref=true,maxbibnames=99, url=false]{biblatex} 
\usepackage{bbm}

\usepackage{nag,tikz}
\usetikzlibrary{calc}
\usetikzlibrary{decorations.pathreplacing}
\usetikzlibrary{shapes, patterns, decorations, fit, intersections, arrows, automata, positioning}

\usepackage{algorithm,algpseudocode}
\usepackage{multicol}
\usepackage[scaled]{helvet} 
\usepackage{thmtools} 
\usepackage{hyperref}
\hypersetup{
    colorlinks=true,
    linkcolor=violet,
    filecolor=magenta,      
    urlcolor=cyan,
    citecolor=blue,
    pdffitwindow=true,
}\usepackage[capitalize, nameinlink]{cleveref}

\theoremstyle{plain}
\newtheorem{theorem}{Theorem}[section]
\newtheorem{lemma}[theorem]{Lemma}
\newtheorem{corollary}[theorem]{Corollary}
\newtheorem{proposition}[theorem]{Proposition}
\newtheorem{fact}[theorem]{Fact}

\newtheorem{observation}[theorem]{Observation}
\newtheorem{claim}[theorem]{Claim}

\newtheorem{definition}[theorem]{Definition}

\theoremstyle{remark}
\newtheorem{remark}[theorem]{Remark}
\newtheorem{example}[theorem]{Example}

\newcommand{\old}[1]{}

\newcommand{\st}{\text{s.t.}}

\renewcommand{\R}{\ensuremath{\mathbb R}}

\renewcommand{\P}[1]{{\mathbb{P}}\left[#1\right]}
\renewcommand{\PP}[2]{{\mathbb{P}}_{#1}\left[#2\right]}

\renewcommand{\E}[1]{{\mathbb{E}}\left[#1\right]}
\renewcommand{\EE}[2]{{\mathbb{E}}_{#1}\left[#2\right]}

\renewcommand{\path}[2]{{ S_{#1}, \ldots, S_{#2} }}

\def\decrease{\beta}

\def\b1{{\bf 1}}
\def\1{{\bf 1}}

\def\C{{\cal C}}

\def\eps{{\epsilon}}

\def\cI{{\cal I}}
\def\cH{{\cal H}}

\def\cA{{\cal A}}
\def\cN{{\cal N}}

\def\cost{c}

\def\setminus{-}
\def\R{\mathbb{R}}

\def\cC{{\cal C}}
\newcommand{\norm}[1]{\|#1\|}
\def\bbe{{\bf e}}

\newcommand{\declareperson}[1]{\expandafter\newcommand\csname#1\endcsname[1]{\textcolor{orange}{#1: ##1}}}

\declareperson{Anna}
\declareperson{Shayan}
\declareperson{Nathan}
\hypersetup{ colorlinks=true, linkcolor=blue, filecolor=magenta, urlcolor=blue, }

\title{A (Slightly) Improved Bound on the Integrality Gap of the Subtour LP for TSP}
\author{Anna R. Karlin\thanks{\href{mailto:karlin@cs.washington.edu}{karlin@cs.washington.edu}. Research supported by Air Force Office of Scientific Research grant FA9550-20-1-0212 and NSF grant CCF-1813135.}}
\author{Nathan Klein\thanks{\href{mailto:nwklein@cs.washington.edu}{nwklein@cs.washington.edu}. Research supported in part by NSF grants DGE-1762114, CCF-1813135, and CCF-1552097.}}
\author{Shayan Oveis Gharan\thanks{\href{mailto:shayan@cs.washington.edu}{shayan@cs.washington.edu}. Research supported by Air Force Office of Scientific Research grant FA9550-20-1-0212, NSF grants  CCF-1552097, CCF-1907845,  ONR YIP grant N00014-17-1-2429, and a Sloan fellowship.}} 
\affil{University of Washington}

\begin{document}
\maketitle 
\begin{abstract}
We show that for some $\epsilon > 10^{-36}$ and any metric TSP instance, the max entropy algorithm studied by \cite{KKO21} returns a solution of expected cost at most $\frac{3}{2}-\epsilon$ times the cost of the optimal solution to the subtour elimination LP. This implies that the integrality gap of the subtour LP is at most $\frac{3}{2}-\epsilon$. 

This analysis also shows that there is a randomized $\frac{3}{2}-\epsilon$ approximation for the 2-edge-connected multi-subgraph problem, improving upon Christofides' algorithm. 
\end{abstract}
\thispagestyle{empty} 

\newpage

\tableofcontents
\thispagestyle{empty} 
\newpage 

\setcounter{page}{1}
\section{Introduction}

One of the most fundamental problems in combinatorial optimization is the traveling salesperson problem (TSP), formalized as early as 1832 (c.f. \cite[Ch 1]{ABCC07}).
In an instance of  TSP we are given a set of $n$ cities $V$ along with their pairwise symmetric distances, $c:V\times V \to\R_{\geq 0}$. The goal is to find a Hamiltonian cycle of minimum cost. In the metric TSP problem, which we study here, the distances satisfy the triangle inequality. Therefore, the problem is equivalent to finding a closed Eulerian connected walk of minimum cost.

It is NP-hard to approximate TSP within a factor of $\frac{123}{122}$ \cite{KLS15}.  An algorithm of Christofides-Serdyukov~\cite{Chr76,Ser78} from four decades ago gives a $\frac32$-approximation for TSP.
Over the years there have been numerous attempts to improve the Christofides-Serdyukov algorithm and exciting progress has been made for various special cases of metric TSP, e.g., \cite{OSS11,MS11,Muc12,SV12,HNR21, KKO20, HN19, GLLM21}.
 Recently, ~\cite{KKO21} gave the first improvement for the general case by demonstrating that the so-called ``max entropy" algorithm of \cite{OSS11} gives a randomized $\frac{3}{2}-\epsilon$ approximation for some $\epsilon > 10^{-36}$.

	The method introduced in \cite{KKO21} exploits the optimum solution to the following linear programming relaxation of metric TSP studied by \cite{DFJ59,HK70,BG93}, also known as the subtour elimination LP:
\begin{equation}\label{eq:tsplp}
\begin{aligned}
	\min \quad& \sum_{u,v} x_{\{u,v\}} c(u,v)& \\
	\text{s.t.,} \quad &  \sum_{u} x_{\{u,v\}} = 2&\forall v\in V,\\
	& \sum_{u\in S, v\notin S} x_{\{u,v\}}\geq 2,&\forall S \subsetneq V, S\not= \emptyset\\
	& x_{\{u,v\}}\geq 0 &\forall u,v\in V.
\end{aligned}	
\end{equation} 
	
	 However, ~\cite{KKO21} did not show that the integrality gap of the subtour elimination polytope is bounded below $\frac{3}{2}$, and therefore did not make progress towards the ``4/3 conjecture" which posits that the integrality gap of LP \eqref{eq:tsplp} is $\frac{4}{3}$. In this work we remedy this discrepancy by proving the following theorem, improving upon the bound of $\frac{3}{2}$ from Wolsey~\cite{Wol80} in 1980:

\begin{theorem}\label{thm:main}
	Let $x$ be a solution to LP \eqref{eq:tsplp} for a TSP instance. For some absolute constant $\epsilon > 10^{-36}$, the \hyperlink{tar:alg}{max entropy algorithm} outputs a TSP tour with expected cost at most $\frac{3}{2}-\epsilon$ times the cost of $x$. Therefore the integrality gap of the subtour elimination LP is at most $\frac{3}{2} - \epsilon$. 
\end{theorem} 

To prove \cref{thm:main}, we amend Section 4 of \cite{KKO21} but keep the remainder of the analysis essentially the same. Unlike \cite{KKO21}, this argument now preserves the integrality gap by avoiding the use of the optimum solution in bounding the cost of the matching. See \cref{sec:overview} for a discussion of our new approach.

\subsection{Other Consequences}
\paragraph{Path TSP} In recent exciting work, Traub, Vygen, Zenklusen \cite{TVZ20} showed that an $\alpha$-approximation algorithm for metric TSP can be used as a black box to get a $\alpha(1+\eps)$ approximation algorithm for Path TSP. This together with \cite{KKO21} implies that there is a $3/2-\eps$ approximation algorithm for Path TSP (for $\eps>10^{-36}$). On the other hand, it is a folklore result that the integrality gap of the natural LP relaxation of Path TSP is at least $3/2$.  Therefore, a consequence of the above theorem is that although the best possible approximation factors of the two problem are the same (up to polynomial reductions), the natural LP relaxation of metric TSP has a strictly smaller integrality gap.

\paragraph{2-ECSM} In the 2-edge-connected multi-subgraph problem, or 2-ECSM for short, we are given a weighted graph $G$ and we want to find a minimum cost 2-edge-connected spanning subgraph, where an edge can be chosen multiple times.
The classical Christofides-Serdyukov algorithm gives a 3/2-approximation for 2-ECSM and despite significant attempts \cite{CR98,BFS16,SV14,BCCGISW20} improved algorithms were designed only for special cases of the problem.
Since in \cite{BG93} it is shown that LP \eqref{eq:tsplp} is a valid relaxation for 2-ECSM, we obtain:

\begin{corollary}	
For some absolute constant $\epsilon > 10^{-36}$ the \hyperlink{tar:alg}{max entropy algorithm} is a randomized $\frac{3}{2}-\epsilon$ approximation for the 2-edge-connected multi-subgraph problem.
\end{corollary}





\subsection{New techniques and contributions}\label{sub:newtechniques}

This paper can be seen as a case study on how to reason about and deal with {\em near} minimum cuts. One can deduce from the classical cactus representation of a graph $G$ \cite{DKL76} (i) the structure of {\em all} min cuts of $G$ and (ii) the structure of the edges of $G$ in the sense that every edge $\{u,v\}$ maps to a unique {\em path} in the cactus between the images of $u$ and $v$. Furthermore, such a path intersects every cycle of the cactus on at most one cactus edge. The theory has found many application from designing fast algorithms
\cite{Kar00,KP09} to the analysis of approximation algorithms for TSP \cite{KKO20} and connectivity augmentation \cite{BGJ20,CTZ21}.

Two decades later, the theory of min cuts was extended to near min cuts in works of Bencz\'ur and Goemans \cite{Ben95, BG08} where they introduced the polygon representation which represents all cuts of a graph with at most $\frac{6}{5}k$ edges, where $k$ is its edge connectivity. Although these works completely classify the structure of all near min cuts of a given graph $G$, they do not characterize the structure of the \textit{edges} of $G$ with respect to these cuts, which can be important in applications (for example, in many of the recent applications of min cuts,
 one also needs to exploit the structure of the edges in relation to the cactus).
The structure on the edges turns out to be highly relevant in this work as well, and as a byproduct of our analysis we make progress towards classifying the way in which the edges of $G$ relate to the structure of the polygon representation.
 
 


For motivation, consider a generic family of network design problems in which we want to construct a network such that every pair $u,v$ of vertices has connectivity at least $c_{u,v}$. A natural approach is to write an LP relaxation to find a (minimum cost) vector $x: E \to \R_{\ge 0}$ such that for every cut $S$ separating $u$ and $v$, $x(\delta(S))\geq c_{u,v}$. We can round this LP using independent rounding or a dependent rounding scheme such as sampling from max entropy distributions. Using classical concentration bounds one can show that if $x(\delta(S))\gg c_{u,v}$ then with high probability the rounded solution has at least $c_{u,v}$ edges across this cut. So the main challenge is to ``fix'' near tight cuts, i.e., cuts where $x(\delta(S))\approx c_{u,v}$.  For an explicit instantiation of this scheme see \cite{KKOZ22}. A better understanding of the global structure of the family of near tight cuts has the potential to significantly simplify or even improve the approximation factor of such rounding algorithms. A classical technique to design algorithms for such network design problems is to apply uncrossing to extreme point solutions of the LP. One can view our contribution as an approximate uncrossing technique that deals with all near tight cuts (instead of just tight cuts) as we explain next.

\paragraph{An Approximate Uncrossing Technique.} A fundamental technique in the field of approximation algorithms is the uncrossing technique\footnote{See e.g. \cite{LRS11} for a number of applications of this technique.} of Jain \cite{Jai01}. Given a graph $G=(V,E)$,  a weight vector $x:E\to\R_{\geq 0}$, and a  function $f:V\to\R$, suppose that $x(\delta(S))\geq f(S)$ for all $S\subseteq V$. Let $\cN$ be the family of sets $S$ such that $x(\delta(S)) = f(S)$, i.e., the family of {\em tight} sets with respect to $f$. The uncrossing technique says that if $f$ is (weakly) supermodular then we can refine $\cN$ to a laminar family of sets, $\cH$, such that if all sets of $\cH$ are tight, then all sets of $\cN$ are tight as well. For a concrete example, suppose $f$ is a constant function, say $f(S)=2$ for all $\emptyset\subsetneq S\subsetneq V$. Then, sets of $\cH$ can be constructed using the cactus representation \cite{DKL76} of cuts in $\cN$. The significance of this method is that if $x$ is a basic feasible solution to a LP with constraints $x(\delta(S))\geq f(S)$ for all $S$, one can use this machinery to argue that the support of $x$ has size $O(|V|)$.

Informally, we prove the following, which 
can be seen as  an {\em approximate uncrossing technique}: 
\begin{theorem}[Informal]\label{thm:uncrossing}Suppose we have a vector $x:E\to\R_{\geq 0}$ such that $x(\delta(S))\geq f(S)$ for all $S$; define $\cN$ to be sets $S$ where $x(\delta(S))\leq f(S)(1+\eps)$ for some fixed $\eps>0$. If $f(.)$ is constant, say $f(S)=2$ for all $S$, then there is a set $\cN^*\subseteq \cN$ and a collection of edge sets $F_1,\dots,F_m\subseteq E$ such that the following hold:
\begin{itemize}
	\item $|\cN^*|= O(|V|)$, $m= O(|V|)$.
	\item $x(F_i)\geq 1-\eps/2$ for all $1\leq i\leq m$.
	\item Every edge $e$ is in at most $O(1)$ of the $F_i$'s.
	\item For every set $S\in \cN\smallsetminus \cN^*$ there exists $1\leq i<j\leq m$ such that $F_i\cap F_j=\emptyset$ and $F_i\cup F_j\subseteq \delta(S)$ and for every $S\in \cN^*$, there exists $1\leq i\leq m$ such that $F_i\subseteq \delta(S)$. 
\end{itemize}
\end{theorem}
In words, although we cannot simply refine $\cN$ to a linear number of sets, we can refine the edges in cuts of $\cN$ to a linear number of sets $F_1,\dots, F_m$ such  that we can essentially capture the edges of $\delta(S)$ for any $S\in \cN\smallsetminus \cN^*$ by a pair of disjoint $F_i$'s. We give a slightly weaker condition for cuts in $\cN^*$; namely we only capture half of their edges by $F_i$'s.

\begin{example}For a simple example of the above theorem, suppose $\eps=0$, i.e. $\cN$ is the set of min cuts of a graph $G$. Furthermore, suppose that every proper  cut in $\cN$ is \hyperlink{tar:crossing}{crossed} (recall that $S$ is proper if $1<|S|<|V|-1$) and that $\cN$ has at least one proper cut. 
Then, one can use an uncrossing technique, namely that if $A,B\in \cN$ then $A\cap B\in \cN$, to prove that $G$ must be cycle, namely we can order vertices of $G$, $v_0,\dots,v_{n-1}$ such that $x_{\{v_i,v_{i+1\text{ mod n}}\}}=1$.
In such a case we let $\cN^*=\emptyset$ and $F_i=E(v_i,v_{i+1\text{ mod }n})$.
\label{eg:cycle}\end{example}

\begin{example}\label{eg:laminar}
For a second example, suppose again $\eps=0$ and $\cN$ is the set of mincuts of a graph $G$ where $\cN$ forms a laminar family (no two cuts cross). It turns out that we cannot decompose edges of cuts of $\cN$ into a linear sized collection of sets where every edge appears only a constant number of times. The main reason is that some edges may appear in an unbounded number of cuts. In this case we let $\cN^*=\cN$ and for every $A\in \cN$ (with immediate parent $B\in \cN$ in the laminar family) we add a set $F_A=\delta(A)\smallsetminus \delta(B)$  to our collection.  It is straightforward to show, using the structure of min cuts, that $x(F_A)\geq 1$; furthermore, since the size of a laminar family is linear in $V$, this gives a valid decomposition in the sense of above theorem.
\end{example}
Lastly, if $\eps=0$ and $\cN$ is the set of min cuts of an arbitrary graph, one can represent all min cuts of $\cN$ by a cactus \cite{DKL76} which can be seen as a tree of cycles. In such a case, one can use a construction similar to \cref{eg:cycle} for each cycle where instead of a vertex $v_i$ we have a set $a_i \subseteq V$ and one similar to \cref{eg:laminar} for the tree part of the cactus. For a concrete application of such a decomposition of min cuts see \cite{KKO20}.



One of the main challenges in dealing with near min cuts relative to min cuts is that if $x(\delta(A)),x(\delta(B))\leq 2+\eps$ then $x(\delta(A\cap B))\leq 2+2\eps$. Therefore, if $\eps=0$, then min cuts are closed under intersection, set difference and union, but this is no longer true when $\eps>0$. So, to employ the classical uncrossing machinery one should be very careful to "uncross" only a constant number of times (independent of $\eps$) to make sure that every cut remains within $2+O(\eps)$. This is the main reason that the polygon representation of near min cuts (see below) is more sophisticated, e.g., we can no longer argue $x(E(a_i, a_{i+1}))\approx 1$, see \cref{fig:nearmincutbadexample}.

Although we don't study it here, we believe it may be worthwhile to find generalizations of \cref{thm:uncrossing} which hold for any (weakly) supermodular function.

\begin{remark} 
 We do not explicitly prove \cref{thm:uncrossing} in this extended abstract, as it is not used to prove \cref{thm:main}. However it can be deduced from arguments in \cref{sec:twoside} and \cref{app:oneside}. 
\end{remark}

\paragraph{Extensions to the Polygon Representation} To obtain our uncrossing framework we prove new properties of the polygon representation.
Given a graph $G=(V,E)$, let $k$ be the edge-connectivity of $G$, i.e. the number of edges in a minimum cut of $G$. For $\eps>0$, consider the set of $(1+\eps)$-near minimum cuts of $G$: cuts $(S,\overline{S})$ where $|E(S,\overline{S})| < (1+\eps)k$.
Bencz\'ur \cite{Ben95} and Bencz\'ur, Goemans \cite{BG08} proved that if $\eps \le 1/5$ then the near minimum cuts of $G$ admit a {\em polygon representation}. Namely, every connected component $\cC$ of \hyperlink{tar:crossing}{crossing} $(1+\eps)$ near min cuts can be represented by the diagonals of a convex polygon. In this polygon, the vertices of $G$ are partitioned into sets called \textit{atoms}, and every atom is mapped to a cell of this polygon defined by the diagonals and the boundary of the polygon itself (see \cref{sec:polyrep} for more details). 

The polygon representation can be seen as a generalization of the well-known cactus representation \cite{DKL76} of minimum cuts where a cycle of the cactus is replaced by a convex polygon. Unlike a cycle, some vertices/atoms map to the interior of the polygon, which are called ``inside'' atoms. The inside atoms at first look like a mystery and one can ask many questions about them such as how many can exist and what structures they can exhibit.

 Here, we explain two lemmas we proved which might find further applications beyond TSP in the future. 
 First, we give a necessary condition for a cell of a polygon to contain an inside atom:
\begin{lemma}[Informal, see \cref{thm:halfplanes}]
	Consider a polygon $P$ for a connected component $\C$ of a family of $1+\eps$ near min cuts for $\eps \le 1/5$ (where representing diagonals correspond to cuts in $\C$). Any cell of $P$ that has an inside atom must have at least $\Omega(1/\eps)$ many sides. 
\end{lemma}
This can be seen as a generalization of \cite[Lem 22]{BG08} to the case in which the cell is allowed to be adjacent to vertices of the polygon $P$.

Now, we explain our second extension: it follows from the cactus representation of minimum cuts that for a graph $G$ and a min cut $S$ one can partition the set of all min cuts that cross $S$ into two groups ${\cal A}=\{A_1,\dots,A_k\}$ and ${\cal B}=\{B_1,\dots,B_l\}$ for some $k,l\geq 0$ such that $S\cap A_1\subseteq S\cap A_2 \subseteq \dots S\cap A_k$ and, similarly, $S\cap B_1\subseteq \dots\subseteq S\cap B_l$. We prove a generalization of this fact for near min cuts:
\begin{lemma}[Informal, see \cref{lem:crosschain}]
Consider the set of $1+\eps$ near min cuts of a graph $G$ for $\eps\leq 1/10$; for any such near min cut $S$, one can partition the $1+\eps$ near min cuts crossing $S$ into two groups ${\cal A}=\{A_1,\dots,A_k\}$ and ${\cal B}=\{B_1,\dots,B_l\}$ such that $S\cap A_1 \subseteq S\cap A_2\subseteq \dots \subseteq S\cap A_k$ and similarly for cuts in ${\cal B}$.
\end{lemma}

\subsection{Outline of rest of paper} After reviewing preliminaries in \cref{sec:prelims}, we give a high-level overview of our proof technique in \cref{sec:overview}. The main new technical contributions of this paper are in \cref{sec:polyrep} and  \cref{sec:twoside}. The remaining content of the paper essentially follows from ~\cite{KKO21}. 


\section{Preliminaries}
\label{sec:prelims}
In the interest of getting quickly to the overview in \cref{sec:overview}, on a first pass reading of this paper, the reader may wish to skip over the (short) proofs later in this section.

\subsection{Algorithm}

Let $x^0$ be an optimum solution of LP \eqref{eq:tsplp}. 
Without loss of generality we assume $x^0$ has an edge $e_0=\{u_0,v_0\}$ with $x^0_{e_0}=1, c(e_0)=0$.
(To justify this, consider the following process: given $x^0$, pick an arbitrary node, $u$, split it into two nodes $u_0,v_0$ and set $x_{\{u_0,v_0\}}=1, c(e_0)=0$ and assign half of every edge incident to $u$ to $u_0$ and the other half to $v_0$.) 

Let $E_0=E\cup\{e_0\}$ be the support of $x^0$ and let $x$ be $x^0$ restricted to $E$ and $G=(V,E)$. Note $x^0$ restricted to $E$ is in the spanning tree polytope  \eqref{eq:spanningtreelp} of $G$.

For a vector $\lambda:E\to\R_{\geq 0}$, a $\lambda$-uniform distribution $\mu_\lambda$ over spanning trees of $G=(V,E)$ is a distribution where for every spanning tree $T\subseteq E$, $\PP{\mu_\lambda}{T}=\frac{\prod_{e\in T} \lambda_e}{\sum_{T'} \prod_{e\in T'} \lambda_e}$.
The second step of the algorithm is to find a vector $\lambda$ such that for every edge $e\in E$, $\PP{T \sim \mu_\lambda}{e\in T}=x_e(1\pm\eps)$, for some $\eps<2^{-n}$. Such a vector $\lambda$ can be found using the multiplicative weight update algorithm \cite{AGMOS10} (see \cref{thm:maxentropycomp}) or by applying interior point methods \cite{SV12} or the ellipsoid method \cite{AGMOS10}. (We note that the multiplicative weight update method can only guarantee $\eps<1/\text{poly}(n)$ in polynomial time.)

Finally, similar to Christofides' algorithm, we sample a tree $T\sim\mu_{\lambda}$ and then add the minimum cost matching on the odd degree vertices of $T$.

\hypertarget{tar:alg}{\begin{algorithm}[h]
\begin{algorithmic}
	\State Find an optimum solution $x^0$ of  \cref{eq:tsplp}, and let $e_0=\{u_0,v_0\}$ be an edge with $x^0_{e_0}=1,c(e_0)=0$.
	\State Let $E_0=E\cup \{e_0\}$ be the support of $x^0$ and $x$ be $x^0$ restricted to $E$ and $G=(V,E)$.
	\State Find a vector $\lambda:E\to\R_{\geq 0}$ such that for any $e\in E$, $\PP{T \sim \mu_\lambda}{e \in T}=x_e(1\pm 2^{-n})$.
	\State Sample a tree $T\sim\mu_\lambda$.
	\State Let $M$ be the minimum cost matching on odd degree vertices of $T$.
	\State Output $T \cup M$.
\end{algorithmic}
\caption{Max Entropy Algorithm for TSP}\label{alg:tsp}
\end{algorithm}}
The above algorithm from \cite{KKO21} is a slight modification of the algorithm proposed in  \cite{OSS11}. While the proof of \cref{thm:main} heavily utilizes properties of max entropy trees, we note that \cref{thm:cutsbothsideswithinside} (the main contribution of this paper) only uses the fact that the spanning tree distribution respects the marginals of $x$.
\begin{theorem}[\cite{AGMOS10}]
\label{thm:maxentropycomp}
Let $z$ be a point in the spanning tree polytope (see \eqref{eq:spanningtreelp}) of a graph $ G=(V, E)$.
For any $\eps>0$, a vector $\lambda:E\to\R_{\geq 0}$ can be found such that the corresponding $\lambda$-uniform spanning tree distribution, $\mu_\lambda$, satisfies
$$
\sum_{T\in {\cal T}: T \ni e} \PP{\mu_\lambda}{T}  \leq (1+\varepsilon)z_e,\hspace{3ex}\forall e\in E,$$
i.e., the marginals are approximately preserved.  In the above ${\cal T}$ is the set of all spanning trees of $(V,E)$. The running
time is polynomial in $n=|V|$, $- \log \min_{e\in E} z_e$ and $\log(1/\eps)$.
\end{theorem}

%
%
%

\subsection{Notation}
We write $[n]:=\{1,\dots,n\}$ to denote the set of integers from $1$ to $n$.
For a set of edges $A\subseteq E$ and (a tree) $T\subseteq E$, we write\footnote{We put this notation in a box because it is so important and ubiquitous in this paper.} 
$$\boxed{\hypertarget{tar:AT}{A_T = |A \cap T|}.}$$
For a set $S\subseteq V$, we write 
$$E(S)=\{\{u,v\}\in E: u,v\in S\}$$ to denote the set of edges in $S$ and we write 
$$\delta(S)=\{\{u,v\}\in E: |\{u,v\}\cap S|=1\}$$ 
to denote the  set of edges that leave $S$. 
For two {\em disjoint} sets of vertices $A,B\subseteq V$, we write
$$ E(A,B)=\{\{u,v\}\in E: u\in A, v\in B\}.$$
For a set $A\subseteq E$ and a function $x:E\to\R$ we write
$$ x(A):=\sum_{e\in A} x_e.$$
\hypertarget{tar:crossing}{For two sets $A,B\subseteq V$, we say $A$ {\em crosses} $B$ if all of the following sets are non-empty:
$$ A\cap B, A\smallsetminus B, B\smallsetminus A, \overline{A\cup B}.$$}
\hypertarget{tar:G=(V,E,x)}{We write $G=(V,E,x)$ to denote an (undirected) graph $G$ together with special vertices $u_0,v_0$ and a weight function $x:E\to\R_{\geq 0}$. Similarly, let $G_0 = (V,E_0,x^0)$ and let $G_{/ e_0} = G_0/\{e_0\}$, i.e. $G_{/ e_0}$ is the graph $G_0$ with the edge $e_0$ contracted.}

\subsection{Polyhedral background}
For any graph $G=(V,E)$,
Edmonds \cite{Edm70} gave the following description for the convex hull of spanning trees of a graph $G=(V,E)$, known as the {\em spanning tree polytope}.
\begin{equation}
\begin{aligned}
& z(E) = |V|-1 & \\
& z(E(S)) \leq |S|-1 &  \forall S\subseteq V\\
& z_e \geq 0 & \hspace{6ex} \forall e\in E.
\end{aligned}
\label{eq:spanningtreelp}
\end{equation}
Edmonds \cite{Edm70} proved that the extreme point solutions of this polytope are the characteristic vectors of the spanning trees of $G$.

\begin{fact} \label{fact:sptreepolytope}
Let $x^0$ be a feasible solution of \eqref{eq:tsplp} such that $x^0_{e_0}=1$ with support $E_0=E\cup \{e_0\}$. 
Let $x$ be $x^0$ restricted to $E$; then $x$ is in the spanning tree polytope of $G=(V,E)$. 
\end{fact}
\begin{proof}
For any set $S\subseteq V$ such that $u_0,v_0\notin S$, $x(E(S))=\frac{2|S|-x^0(\delta(S))}{2}\leq |S|-1$.
If $u_0\in S, v_0\notin S$, then
$x(E(S)) = \frac{2|S|-1 - (x^0(\delta(S)) -1 )}{2}\leq |S|-1$.
Finally, if $u_0,v_0\in S$, then 
$x(E(S)) = \frac{2|S|-2 - x^0(\delta(S))}{2} \leq |S|-2$.
The claim follows because $x(E)=x^0(E_0)-1=n-1$.
 \end{proof}

Since $c(e_0)=0$, the following fact is immediate.
\begin{fact} \label{fact:expcostT}Let $G=(V,E,x)$ where  $x$ is in the spanning tree polytope. If $\mu$ is any distribution of spanning trees with marginals $x$ then $\EE{T\sim\mu}{c(T \cup e_{0})}=c(x)$.
 \end{fact}
 
 To bound the cost of the min-cost matching on the set $O$ of odd degree vertices of the tree $T$, we use the following characterization of the $O$-join polyhedron\footnote{The standard name for this is the $T$-join polyhedron. Because we reserve $T$ to represent our tree, we call this the $O$-join polyhedron, where $O$ represents the set of odd vertices in the tree.} due to Edmonds and Johnson \cite{EJ73}.
\begin{proposition}
\label{prop:tjoin}
For any graph $G=(V,E)$, cost function $c: E \to \R_+$, and a set $O\subseteq V$ with an even number of vertices,  the minimum weight of an $O$-join equals the optimum value of the following integral linear program.
\begin{equation}
\begin{aligned}
\min \hspace{4ex} & \cost(y) \\
\st \hspace{3ex} & y(\delta(S)) \geq 1 & \forall S \subseteq V, |S\cap  O| \text{ odd}\\
& y_e \geq 0 & \forall e\in E
\end{aligned}
\label{eq:tjoinlp}
\end{equation}
\end{proposition}

\begin{definition}[Satisfied cuts]\label{def:satisfiedcuts}
\hypertarget{tar:satisfy}{For a set $S\subseteq V$ such that $u_0,v_0\notin S$ and a spanning tree $T\subseteq E$ we say a vector $y:E\to\R_{\geq 0}$ 	satisfies $S$ if one of the following holds:
\begin{itemize}
\item $\delta(S)_T$ is even, or
\item $y(\delta(S))\geq 1$.	
\end{itemize}}
\end{definition}
To analyze this class of algorithms, the main challenge is to construct a (random) vector $y$ that satisfies all cuts (with probability 1) and for which $\E{c(y)}\leq (1/2-\eps)c(x)$.  
 


\subsection{Near Min Cuts}
\begin{definition}\hypertarget{tar:nearmincut}{For $G=(V,E,x)$, we say a cut $S\subseteq V$ is an {\em $\eta$-near min cut} if $x(\delta(S))< 2+\eta$.\footnote{Note this differs slightly from the notation in \cite{Ben95, BG08} and \cref{sub:newtechniques} in which an $\eta$ near min cut is said to be within a $1+\eta$ factor of the edge connectivity of the graph.}}
\end{definition}


The following lemma is a standard fact about crossing near min cuts:
\begin{lemma}\label{lem:cutdecrement}
For $G=(V,E,x)$, let $A,B\subsetneq V$ be two crossing $\eps_A, \eps_B$ near min cuts respectively. Then,
$A\cap B, A\cup B, A\smallsetminus B, B\smallsetminus A$ are $\eps_A+\eps_B$ near min cuts.
\end{lemma}
\begin{proof}
We prove the lemma only for $A\cap B$; the rest of the cases can be proved similarly.
By submodularity,
$$ x(\delta(A\cap B)) + x(\delta(A\cup B)) \leq x(\delta(A)) + x(\delta(B)) \leq 4+\eps_A+\eps_B.$$
Since $x(\delta(A\cup B))\geq 2$, we have $x(\delta(A\cap B))\leq 2+\epsilon_A+\eps_B$, as desired.
\end{proof}

\begin{lemma}\label{lem:sub-NMC-shared}
If $A,B\subsetneq V$ are disjoint and $C=A \cup B$ is an $\epsilon$ near min cut then $x(E(A,B)) \ge 1 - \frac{\epsilon}{2}$.
\end{lemma}
\begin{proof}
$$2+\epsilon \ge x(\delta(C)) = x(\delta(A)) + x(\delta(B)) - 2 \cdot x(E(A,B)) \ge 4 - 2 \cdot x(E(A,B))$$ 
And the claim follows.
\end{proof}

The following lemma is proved in \cite{Ben97}:
\begin{lemma}[{\cite[Lem 5.3.5]{Ben97}}]
\label{lem:nmcuts_largeedges}
For $G=(V,E,x)$, let $A,B\subsetneq V$ be two crossing $\epsilon$-near minimum cuts. 
Then, $$x(E(A\cap B, A\setminus B)),x(E(A\cap B, B\setminus A)), x(E(\overline{A\cup B}, A\setminus B)), x(E( \overline{A\cup B}, B\setminus A)) \geq (1-\epsilon/2).$$
\end{lemma}

\begin{lemma}
\label{lem:shared-edges}
For $G=(V,E,x)$, let $A,B\subsetneq V$ be two $\eps$ near min cuts  such that $A \subsetneq B$. Then 
$$x(\delta(A) \cap \delta(B)) = x(E(A,\overline{B}))\le 1 + \eps, \text{ and }$$
$$x(E(A,B \smallsetminus A))\geq 1-\eps/2. $$
\end{lemma}
\begin{proof}
Notice
\begin{align*}&2+\epsilon \ge x(\delta(A)) = x(E(A,B \smallsetminus A)) + x(E(A,\overline{B}))\\
&2+\epsilon \ge x(\delta(B)) = x(E(B \smallsetminus A,\overline{B})) + x(E(A,\overline{B}))
\end{align*}
Summing these up, we get
$$2x(E(A,\overline{B})) + x(E(A,B \smallsetminus A)) + x(E(B \smallsetminus A, \overline{B})) = 2x(E(A,\overline{B}))+x(\delta(B\smallsetminus A)) \le 4+2\eps.$$
Since $B \smallsetminus A$ is non-empty,
	$x(\delta(B\smallsetminus A)) \ge 2$,
which implies the first inequality.
To see the second one, let $C=B\smallsetminus A$ and note
$$ 4\leq x(\delta(A))+x(\delta(C)) = 2 x(E(A,C)) + x(\delta(B))\leq 2 x(E(A,C))+ 2+\eps$$
which implies $x(E(A,C))\geq 1-\eps/2$.
\end{proof}

\subsection{Random spanning trees}
The following simple lemmas appear in e.g. \cite{KKO21}:

\begin{lemma}\label{lem:treeconditioning}
Let $G=(V,E,x)$, and let $\mu$ be any distribution over spanning trees with marginals $x$. 
For any  $\eps$-near min cut $S\subseteq V$ 
(such that none of the endpoints of $e_0=(u_0,v_0)$ are in $S$), we have
$$\PP{T\sim\mu}{S \text{ is a subtree of $T$}} = \PP{T \sim \mu}{|T \cap E(S)| = |S|-1} \ge 1- \eps/2.$$ 
\end{lemma}
\begin{proof}
First, observe that
$$\E{E(S)_T}= x(E(S)) \geq  \frac{2|S|-x(\delta(S))}{2} \geq |S|-1 -\eps/2,$$ 
where we used that since $u_0,v_0\notin S$, for any $v\in S$ we have $\E{\delta(v)_T)} = x(\delta(v))=2$.

Let $p_S=\PP{T \sim \mu}{S\text{ is a subtree of $T$}}$.
Then, we must have
$$ |S|-1 - (1-p_S) = p_S(|S|-1) + (1-p_S)(|S|-2)\ge \E{E(S)_T} \geq |S|-1 - \eps/2.$$
Therefore, $p_S\geq 1-\eps/2$.
\end{proof}
\begin{corollary}\label{lem:treeoneedge}
Let  $A,B\subseteq V$ be disjoint sets such that  $A,B,A\cup B$ are $\eps_A,\eps_B,\eps_{A\cup B}$-near minimum cuts w.r.t., $x$ respectively, where none of them  contain endpoints of $e_0$.  Then for any distribution $\mu$ of spanning trees on $E$ with marginals $x$,
$$\PP{T\sim \mu}{E(A,B)_T=1}\geq 1-(\eps_A+\eps_B+\eps_{A\cup B})/2.$$	
\end{corollary}
\begin{proof}
By the union bound, with probability at least $1-(\eps_A+\eps_B+\eps_{A\cup B})/2$, $A,B,$ and $A\cup B$ are trees. 
But this implies that we must have exactly one edge between $A,B$.
\end{proof}

The following simple fact also holds by the union bound.
\begin{fact}\label{fact:0edgerandomspanningtree}
Let $G=(V,E,x)$ and let $\mu$ be a distribution over spanning trees with marginals $x$. For any set $A\subseteq E$	, we have
$$ \PP{T\sim\mu}{T\cap A=\varnothing} \geq 1-x(A).$$
\end{fact}

\section{Proof Overview}\label{sec:overview}

\cref{alg:tsp} consists of two steps: sampling a tree whose marginals match $x$ (and hence has expected cost equal to $c(x)$), and then augmenting this with a minimum cost matching on the odd degree vertices of the tree. The goal of the current paper is to show that
the expected cost of the minimum cost matching on the odd degree vertices of the sampled tree is at most $(1/2-\eps)c(x)$. This is done by showing the existence of a cheap feasible $O$-join solution to \eqref{eq:tjoinlp}. Note that we merely need to prove the existence of a cheap $O$-join solution. The actual optimal $O$-join solution can be found in polynomial time.

First, note that if we only wanted to get an $O$-join solution of value at most $c(x)/2$, to \hyperlink{tar:satisfy}{satisfy} all cuts, it is enough to set $y_e := 0.5 x_e$ for each edge\footnote{This is because $x$ satisfies $x(\delta(S))\ge 2$ for all $S$, whereas $y$ must satisfy $y(\delta(S))\ge 1$ just for those cuts that have odd intersection with the tree $T$.} \cite{Wol80}. Now notice that if all of the near min cuts of $x$ containing $e$ are even, then we can reduce $y_e$ strictly below $0.5 x_e$.
The difficulty in implementing this approach comes from the fact that a high cost edge can be on many near min cuts and it may be exceedingly unlikely that {\em all} of these cuts will be even simultaneously. The idea in \cite{KKO21} is to initialize $y_e:= 0.5 x_e$ and then modify it by adding to it a random\footnote{where the randomness comes from the random sampling of the tree}  {\em slack vector} $s:E\to\R$: For each edge $e$, when  certain special (few) $\eta$-near-mincuts that $e$ is on are even in the tree, $s_e$ is set to  $-x_e \decrease$ where $\decrease\approx \eta/4$  chosen  in the proof of \cref{thm:main}; for other cuts that contain $e$, whenever they are odd, the slack of {\em other} edges on that cut is increased to satisfy them (i.e., maintain feasibility of $y$ for that cut). The bulk of the effort was to show that this can be done while guaranteeing that $\E{s_e}<-\eps x_e$ for some $\eps>0$, and therefore
$\E{y_e} = 0.5 x_e + \E{s_e} < (0.5 - \epsilon) x_e$.

To help the reader understand both the big picture as well as the ideas and contribution of the current paper, it is useful to first review in a bit more detail the approach taken in \cite{KKO21}. Let $\cN_\eta$ be the set of all $\eta$-near min cuts of $x$.
 A key idea there was to partition $\cN_\eta$ into three types: a set of near min cuts $\cH$ that form a {\em hierarchy} (which is a laminar family of  cuts), a set of cuts $\cN_{\eta,1}$ that are "crossed on one side" and a set of  cuts  $\cN_{\eta,2}$ that are "crossed on both sides"\footnote{ We will explain the terms in quotes shortly. Also, see \cref{defn:crossedonetwo}.}. \cite{KKO21} showed that if we \text{only} need to satisfy the $O$-Join constraints coming from $\cH$, then we can find such a vector $s$.
 
However, this vector $s$ (which is negative in expectation) might "break" O-join constraints on cuts that are {\em not} in the hierarchy (i.e., cuts in $\cN_{\eta,1}$ and  $\cN_{\eta,2}$). To resolve this, \cite{KKO21} showed how a negligible {\em increase} in the slack of certain edges (a  slack component they called $s^*$)  can be used to restore the feasibility of the O-join solution on \textit{all} cuts, including those that are not in the hierarchy.  See \cref{sec:hierarchy} for more on this.

Concretely,
because the cuts in $\cN_{\eta,2}$ have a rather complex structure,  to simplify their handling, \cite{KKO21} changed the plan: Instead of starting with $y_e = 0.5x_e$, they started with $y_e = (x_e + OPT_e)/2$, where $OPT_e$ is an indicator for edge $e$ being in the optimal integral TSP solution. They then constructed slack vectors relative to the near min cuts of $(OPT + x)/2$. The advantage of doing so is that it guarantees that all near mincuts correspond to intervals of vertices along the optimal cycle, {\em greatly} simplifying the structure of the family of near min cuts under consideration. Slack on the edges in the optimal cycle was then used to handle the cuts in $\cN_{\eta,1}$ and $\cN_{\eta,2}$. 

Unfortunately, this meant that the bound on the expected cost of the minimum cost matching from \cite{KKO21} is
at most $(1/2-\epsilon)((c(x) + c(OPT))/2$, which is insufficient to prove that the integrality gap of the LP is strictly below 3/2.

In the present paper, we return to the plan of initializing $y_e := 0.5 x_e$ and then construct a slack vector for each edge with the desired properties.  Our starting point is the {\em polygon decomposition}  $\mathcal D$ of the $\eta$-near min cuts of $x$~\cite{BG08}\footnote{See \cref{sec:polyrep} for a formal introduction to polygons. In particular, a reader unfamiliar with polygons will likely need to read \cref{sec:polybasics} to understand this section, though we provide a very brief overview now and again in \cref{sec:noinside}.}. 
As stated previously, a polygon\footnote{\ One difference between a cycle on $m$ nodes and a polygon with $m$  "outside atoms" is that  in a cycle all of the $m \choose 2$ simple cuts are min-cuts, whereas in a polygon only some of the $m \choose 2 $ simple cuts are $\eta$-near min cuts. Indeed a cycle is the "simplest" kind of polygon. Another major difference is that polygons may also have "inside" atoms. See \cref{sec:polyrep}.} is a connected component of crossing $2 + \eta$ near minimum cuts, where two cuts are connected if they cross each other.  It turns out that the way the polygon representation  $\mathcal D$ is constructed, each cut in $\cN_{\eta,2}$ is  in exactly one polygon, and each edge on such a cut will have its slack increase in at most one polygon. 
Thus, cuts  in $\cN_{\eta,2}$ can be handled independently for each polygon. 
%

The main result of this paper is to show how to handle the cuts in $\cN_{\eta,2}$ (polygon by polygon) without resorting to the use of the OPT vector. Specifically, we prove the following: 
 \begin{theorem}[Informal main theorem]\label{thm:informalmain}
For any connected component $\cC$ of $\cN_\eta$ (i.e. a polygon),  let $\cC_2$ be the cuts in $\cC$ that are crossed on both sides. For any $\alpha>0$, there is a vector $s^*: E \rightarrow  R$ depending on $T$  s.t. 
\begin{enumerate}
\item[(i)] $\forall e \in E$, $s_e^*  \ge 0$;
\item[(ii)] $\E{s^*_e} = O(\eta \alpha x_e)$, where the expectation is over the choice of tree $T$.	
\item[(iii)] If $S\in \cC_2$ is a cut such that $\delta(S)_T$ is odd, then $s^*_T(\delta(S)) := \sum_{e \in \delta(S)} s^*_e \ge \alpha(1-\eta)$.
\end{enumerate}
 
\end{theorem}

Once cuts in $\cN_{\eta,2}$ are handled, the remaining cut structure becomes significantly simpler in that the polygons start to look very much like cycles: they contain only "outside atoms"\footnote{See \cref{sec:polybasics}.} and the fractional mass $x(a_i, a_{i+1})$ between adjacent atoms is $1\pm \Theta (\eta)$ ~\cite{KKO21}. This enables us, with minor modification to the way in which cuts crossed on one side are handled (see \cref{thm:crossed-one-side}), to adapt one of the main results in \cite{KKO21}.
\begin{restatable}[Informal Theorem adapted from \cref{thm:crossed-one-side} and \cref{thm:payment-main}]{theorem}{informalhierarchy}
\label{thm:hierarchyKKO}
Given a family $\cN_{\eta,\leq 1}$ of near-min cuts containing {\bf  no} cuts crossed on both sides, for any $\decrease>0$, there is a vector $s: E \rightarrow \R$
 depending on $T$ such that	\begin{enumerate}
\item [(i)] $\forall e \in E$, $s_e  \ge -\decrease  x_e$ 
\item [(ii)] $\E{s_e} < -\epsilon \decrease x_e$ for some absolute constant $\epsilon > 0$, independent of $\eta$, 
where the expectations are over the choice of $T$.
\item [(iii)] If $S \in \cN_{\eta,\leq 1}$ is a cut such that $\delta(S)_T$ is odd, then $s_T(\delta(S)) = \sum_{e \in \delta(S)} s_e \ge 0$.
\end{enumerate}
\end{restatable}
Note that \cref{thm:payment-main} (used to prove the above theorem) crucially relies on the fact that the tree is sampled from a max-entropy distribution, whereas \cref{thm:informalmain} does not.

Before we explain the ideas underlying the proof of \cref{thm:informalmain}, we quickly show how by setting
$$y_e(T):= 0.5 x_e + s^*_e+ s_e\quad\forall e,$$
these two theorems together imply  the main result of this paper.

First, we show that $\E{c(y) }\le c(x) (0.5-\epsilon)$ 
To see this, observe that  \cref{thm:informalmain}(ii)  together with 
\cref{thm:hierarchyKKO}(ii) imply  that for every edge $e \in E$,
$$\E{y_e} =  0.5x_e + \E{s_e} + \E{s^*_e}\\
 \le x_e \left(0.5  + O(\eta\alpha) - \epsilon \decrease\right)  \le x_e \left(0.5 - \eta\eps'\right),$$
for $\alpha,\beta$ as chosen below and $\eta$ sufficiently smaller than $\epsilon$. 
Summing over all edges, this gives
$$\E{c(y)} \le \left(\frac{1}{2} - \epsilon'\right)c(x).$$
Note that since $s_e^*$ is always nonnegative, it does not help us in our quest to reduce $y_e$ strictly below $0.5 x_e$. That reduction comes only from $s_e$ being negative. Indeed, the raison d'etre  of the slack vector $s^*$ is to repair the feasibility of cuts which are odd in the tree but which have $s_e$ negative on some edges in $\delta(S)$. This is why it is crucial that $\E{s_e}$ is much smaller than $-\E{s_e^*}$. 

Next, we show that $y(T)$ is feasible for every tree $T$. For this, we need to consider three types of cuts: 

\paragraph{Case 1:} $\delta(S)_T$ is odd and $x(\delta(S)) > 2 + \eta$. Since $s^*_T(\delta(S))\ge 0$ and $s_T(\delta(S)) \ge -\decrease x(\delta(S)$, we have
$$ y_T(\delta(S)) = 0.5x(\delta(S))+ s^*_T(\delta(S)) + s_T(\delta(S)) \ge (0.5 - \decrease)x(\delta(S))\ge (0.5 - \decrease)(2 + \eta) \ge 1,$$
for $\decrease \approx \eta/4$.

\paragraph{Case 2:}$\delta(S)_T$ is odd, $S\in \cN_{\eta,\leq 1}$. In this case $s_T(\delta(S)),s^*_T(\delta(S))\ge 0$ so $y_T(\delta(S) \ge 0.5 x(\delta(S)) \ge 1.$
\paragraph{Case 3:}$\delta(S)_T$ is odd, $x(\delta(S)) \le  2 + \eta$, and $S\in\cN_{\eta,2}$. In this case, 
 $s^*_T(\delta(S))\ge  \alpha(1-\eta)$ and $s_T(\delta(S)) \ge -\decrease x(\delta(S))$, so we have
$$ y_T(\delta(S)) = 0.5 x(\delta(S))+ s^*_T(\delta(S) + s_T(\delta(S)) \ge (0.5 - \decrease)x(\delta(S)) + \alpha(1-\eta) \ge 
1,$$
for $\alpha\approx 2\decrease$ using $x(\delta(S))\ge 2$.

\vspace{0.2in}
\begin{figure}[htb!]\centering
	\begin{tikzpicture}
	\pattern[pattern=north east lines, pattern color=blue!35] (228:3) -- +(1,0) -- (180:3) arc (180:228:3); 
	\pattern [pattern=north east lines, pattern color=purple!35] (312:3) -- +(-1,0) -- (360:3) arc (360:276:3);
		\fill [color=purple!15] (276:3) -- (294:2.4) -- (312:3)  arc (312:276:3); 
		\fill [color=blue!15] (228:3) -- (246:2.4) -- (264:3) arc (264:228:3);
		\foreach \i/\l in {0/0, 180/15, 192/16, 204/17, 216/18, 228/19, 240/20, 252/21, 264/22, 276/23, 288/24, 300/25, 312/26, 324/27, 336/28, 348/29}
			\node [fill, circle, inner sep=0.5mm] at (\i:3) (p_\l) {};
		\foreach \i/\j in { 15/16, 16/17, 17/18, 18/19, 19/20, 20/21, 21/22, 22/23, 23/24, 24/25, 25/26, 26/27, 27/28, 28/29, 29/0}
			\path (p_\i) edge (p_\j);	
		\foreach \i/\l in {19/l, 26/r}
			\node at (\i*12:3.4) () {\footnotesize $p_\l$};	
		\path [color=green,line width=1.2pt] (p_19) edge node [above] {$S$} (p_26);
		\node [color=blue] at (184:2.3) (){$S_L$};
		\node [color=red] at (356:2.3) () {$S_R$};
		\path [color=blue,line width=1.2pt] (p_22) edge  (p_15);
		\path [color=purple,line width=1.2pt] (p_23) edge  (p_0);
		\node [color=orange] at (204:3.7) () {\small $S_L\smallsetminus S$};
		\node [color=blue] at (245:3.4) (){\small $S^{\cap L}$};
		\node [color=red] at (295:3.4) (){\small $S^{\cap R}$};
		\node [color=green] at (335:3.7) (){\small $S_R\smallsetminus S$};
		\foreach \i/\c in {186/orange, 198/orange, 210/orange, 222/orange, 234/blue, 246/blue, 258/blue, 282/red, 294/red, 306/red, 318/green, 330/green, 342/green, 356/green}
			\node [circle,fill,inner sep=0.8mm,color=\c] at (\i:3) () {};  
		\node at (295:2.9) (b1) {}; \node at (340:2.9) (b2) {};
		\node at (245:2.9) (a1) {}; \node at (200:2.9) (a2) {};
		\path [color=red,->,line width=1.3pt] (b1) edge [bend left=90] (b2) ;
		\path [color=blue,->,line width=1.3pt](a1) edge [bend right=90] (a2);
	\end{tikzpicture}
\caption{\footnotesize (Note that to simplify the pictures, we usually draw a polygon as a circle.)  The figure shows a cut $S$ that is crossed on both sides. The cut $S$ consists of all atoms below the green line. The cut $S_R$ is the cut crossing $S$ on the right which minimizes the number of outside atoms in the intersection, i.e., it minimizes the number of red atoms. Similarly, $S_L$ crosses $S$ on the left and minimizes the number of blue atoms. While not shown in the picture, it's possible for the red and blue atoms to overlap. ($S_R$ might cross $S_L$.) Edges between red atoms and green atoms are in $E^\rightarrow (S)$ and edges between blue atoms and orange atoms are in $E^\leftarrow(S)$. Edges in $E^\circ(S)$ are all remaining edges in $\delta(S)$. \cref{claim:C2evenwhp} shows that with probability $1 - O(\eta)$, in the randomly sampled tree $T$, there is exactly one (red,green) edge (i.e., $B^\rightarrow(S)$ does not occur) and exactly one (blue, orange) edge (i.e., $B^\leftarrow(S)$ does not occur)  and those are the only edges in $\delta(S) \cap T$ (i.e., also  $B^\circ(S)$ does not occur).}
	\label{fig:S-SR-SL}
\begin{tikzpicture}[scale=0.8]
		\fill [color=blue!12] (264:3) arc (264:156:3);
		\node [inner sep=3, fill, circle, color=purple, label={[yshift=0]\small $a_0$}] at (90:3) (){};
		\foreach \i/\l in {0/0, 12/1, 24/2, 36/3, 48/4, 60/5, 72/6, 84/7, 96/8, 108/9, 120/10, 132/11, 144/12, 156/13, 168/14, 180/15, 192/16, 204/17, 216/18, 228/19, 240/20, 252/21, 264/22, 276/23, 288/24, 300/25, 312/26, 324/27, 336/28, 348/29}
			\node [fill, circle, inner sep=0.5mm] at (\i:3) (p_\l) {};
			\path [line width=1.2pt,color=blue] (p_22) edge (p_13);
			\draw [line width=1.8,color=blue,->] (-1.63,-0.6) -- +(-0.5,-0.35);
		\foreach \i/\j in {0/1, 1/2, 2/3, 3/4, 4/5, 5/6, 6/7, 7/8, 8/9, 9/10, 10/11, 11/12, 12/13, 13/14, 14/15, 15/16, 16/17, 17/18, 18/19, 19/20, 20/21, 21/22, 22/23, 23/24, 24/25, 25/26, 26/27, 27/28, 28/29, 29/0}
			\path (p_\i) edge (p_\j);
		\foreach \i/\l in {258, 270, 282,  162, 174, 186}
			\node [fill,red,circle,inner sep=0.9mm] at (\i:3) (){}; 
		\foreach \i/\l in {258/i-1, 270/i, 282/i+1, 162/k+1}
			\node [color=red] at (\i:2.5) () {\small $a_{\l}$}	;
		
		\foreach \i/\l in {252/i-1, 264/i, 276/i+1, 156/k}
			\node  at (\i:3.3) () {\footnotesize $p_{\l}$};

		\begin{scope}[xshift=8cm]
			\fill [color=green!12] (264:3) arc (264:396:3);
			\node [inner sep=3, fill, circle, color=purple, label={[yshift=0]\small $a_0$}] at (90:3) (){};
		\foreach \i/\l in {0/0, 12/1, 24/2, 36/3, 48/4, 60/5, 72/6, 84/7, 96/8, 108/9, 120/10, 132/11, 144/12, 156/13, 168/14, 180/15, 192/16, 204/17, 216/18, 228/19, 240/20, 252/21, 264/22, 276/23, 288/24, 300/25, 312/26, 324/27, 336/28, 348/29}
			\node [fill, circle, inner sep=0.5mm] at (\i:3) (p_\l) {};
			\path [line width=1.2pt,color=green] (p_22) edge (p_3);
			\draw [line width=1.8,color=green,->] (1.09,-0.5) -- +(0.5,-0.35);
		\foreach \i/\j in {0/1, 1/2, 2/3, 3/4, 4/5, 5/6, 6/7, 7/8, 8/9, 9/10, 10/11, 11/12, 12/13, 13/14, 14/15, 15/16, 16/17, 17/18, 18/19, 19/20, 20/21, 21/22, 22/23, 23/24, 24/25, 25/26, 26/27, 27/28, 28/29, 29/0}
			\path (p_\i) edge (p_\j);
		\foreach \i/\l in {258, 270, 282, 30, 18, 6}
			\node [fill,red,circle,inner sep=0.9mm] at (\i:3) (){}; 
		\foreach \i/\l in {258/i-1, 270/i, 282/i+1, 30/j}
			\node [color=red] at (\i:2.65) () {\small $a_{\l}$}	;
		
		\foreach \i/\l in {252/i-1, 264/i, 276/i+1, 36/j}
			\node  at (\i:3.3) () {\footnotesize $p_{\l}$};
	
		\end{scope}

	\end{tikzpicture}
\caption{\small Of all cuts crossed on both sides, $L(p_i)$, the blue set, extends farthest to the left from $p_i$. Similarly, $R(p_i)$, the green set, is the one that extends farthest to the right from $p_i$.  (For reference when we later include inside atoms: if the root is not an outside atom, $L(p_i)$ can wrap around past $a_0$ and there may be atoms in the interior of the blue region, aka inside atoms. However, the outside atoms of $L(p_i)\cup R(p_i)$ form a contiguous interval around the cycle and don't include all outside atoms. Even in the case in which the root is not an outside atom, the shaded region is the side of the diagonal which does not contain the root.) 
}
	\label{fig:LpRp}
\end{figure}

\begin{figure}[htb]\centering
	\begin{tikzpicture}[scale=0.8]
		\fill [color=yellow] (330:3) -- ++(-1.07,-0.19) -- +(0.2,0.45);
		\fill [color=green!20] (294:3) -- (330:3)+(-1.07,-0.19) -- (330:3) arc (330:294:3);
		\fill [color=brown!20] (330:3) -- (330:3)+(-0.87,0.26) -- (390:3) arc (390:330:3);
		\draw (0,0) circle (3);
		\foreach \i/\c in {78/purple, 90/purple, 102/purple, 324/green, 312/green, 300/green, 336/brown, 348/brown, 0/brown, 12/brown, 24/brown}{
			\node [circle,inner sep=2,color=\c,fill] at (\i:3) () {};
			}
		\node at (312:3.1) (a){}; \node at (0:3.1) (b) {};
		\path [color=green,->,line width=1.3pt](a) edge [bend right=90] (b);
		\foreach \i/\c/\l in {78/purple/m-1, 90/purple/0, 102/purple/1, 324/green/i, 338/brown/i+1}
			\node [color=\c] at (\i:3.4) () {\footnotesize $a_{\l}$};
		\node [circle,inner sep=4,color=purple,fill] at (90:3) (){};
		\draw [color=black,line width=1.1pt] (330:3) -- node [above] {\small $L(p)$} (180:3);
		\node [circle,fill,inner sep=2] at (180:3) () {};
		\node at (180:3.2) () {\footnotesize $l$};
		\draw [color=green, line width=1.1pt] (294:3) --  (30:3);
		\node [green] at (17:1.7) () {\small $L(p)_R$};
		\draw [color=red,line width=1.1pt] (330:3) -- node [above] {\small $S$} (228:3);
		\foreach \i/\l in {228/q,330/p}{
			\node [circle,fill,inner sep=1.5] at (\i:3) () {};
			\node at (\i:3.3) () {\footnotesize$\l$};
		}
		\draw [color=black,->,line width=1.3pt] (-1,-0.5) -- +(-.14,-0.4);
		\draw[color=red,->,line width=1.3pt] (-0.5,-2) -- +(0.1,-.4);
		\draw[color=green,line width=1.3pt,->] (2.15,0.1) -- +(0.4,-0.1);
	\end{tikzpicture}
	\caption{\footnotesize  Since $S$ is crossed on the right and any cut that crosses $S$ on the right also crosses $L(p)$ on the right,  the cut $S_R$, which contains the fewest number of atoms in $S$ (green atoms), is the same as $L(p)_R$. The edges in  $E^\rightarrow(L(p))= E^\rightarrow(S)$ are those that go between green atoms and brown atoms. Note also that any edge in $\delta(S)$ with one endpoint to the right of $p$  that is {\em not} in $E^\rightarrow(L(p))$ is in  $E^\circ(L(p))$. (It can't be in $E^\leftarrow(L(p))$ since those edges have one endpoint to the left of $l$.) 
	Note also that since the green + yellow region as well as the brown region are each the difference of two crossing $\eta$ near min cuts, each is a $2\eta$ near min cut. So by \autoref{lem:sub-NMC-shared}, the fraction of edges with one endpoint in each of these regions is $1-O(\eta)$. (To extend this to the case where polygon $P$ may have inside atoms we show that there are no atoms in the yellow region see \cref{lem:samerightEPsameedgeset}, or \cref{thm:halfplanes} for a more general statement.)
	}
	\label{fig:ErL=ErS}
\end{figure}
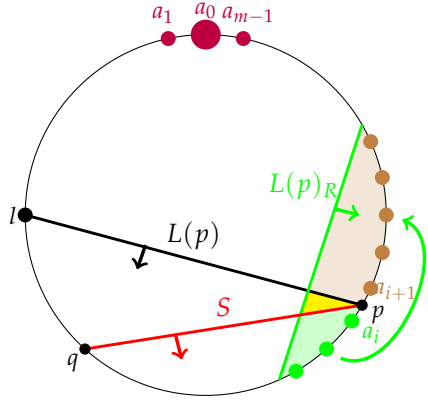

\begin{figure}[htb]\centering
	\begin{tikzpicture}[scale=0.8]
		\draw (0,0) circle (3);
		\node [circle,fill,purple,inner sep=3,label={[yshift=0cm]\small $a_0$}] at(90:3)() {};
		\draw [color=blue,line width=1.2pt] (324:3) -- (144:3)
		(300:3) -- (168:3) (276:3) -- (192:3);
		\draw[color=blue,line width=1.3pt,dotted] (235:2.4)-- +(-0.3,-0.4);
		\draw[color=red,line width=1.2pt] (252:3) -- (24:3);
		\foreach \i/\l in {324/1, 300/2, 276/3}{
			\node [fill,circle,inner sep=2] at (\i:3) (){ };
			\node at (\i:3.3) () {\footnotesize $p_{\l}$};
		}
		\node [circle,color=red,fill,inner sep=2] at (264:3) (a) {};
		\node [circle,color=red,fill,inner sep=2] at (342:3) (b) {} edge [bend right=15] node [above] {$e$} (a);
		\node at (264:3.3) () {\footnotesize $a_j$};
		\node at (342:3.3) () {\footnotesize $a_k$};		
		\draw [color=blue,line width=1.4pt,->] (144:2.5) -- +(-0.3,-0.4);
		\draw [color=blue,line width=1.4pt,->] (175:2.3) -- +(-0.3,-0.4);
		\draw [color=blue,line width=1.4pt,->] (210:2.4) -- +(-0.3,-0.4);
		\draw[color=red,line width=1.4pt,->] (0:1.62) -- +(0.3,-0.3);
		\node [color=red]at (30:3.1) () {\footnotesize $L(p_i)_R$ for all $i$};
		\foreach \i/\l in {135/1, 163/2, 195/3}
			\node at (\i:2.5) () {\footnotesize $L(p_{\l})$};
	\end{tikzpicture}
	\caption{\small Edge $e= \{a_j,a_k\}$ is in $E^\rightarrow(L(p_i))$ for all $i$.}
	\label{fig:einmanyL(p)}
\end{figure}

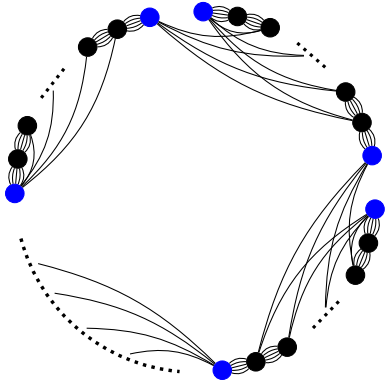
\begin{figure}[htb]
\centering
\begin{tikzpicture}[scale=0.8]
\tikzstyle{every node} = [inner sep=0,minimum size=7,draw,fill=none,circle];
\foreach \a in {0, 1, 2}{
\begin{scope}[rotate=93*\a]
\foreach \i/\l/\c in {0.5/0/blue, 1.5/1/black, 2.5/2/black,  5.4/4/black,6.4/5/black,7.4/6/blue}{
	\node at (-\i*11:3) [color=\c,fill] (\a_\l) {};
}
\node at (-4*11:3) [draw=none] (\a_3) {};
\foreach \i/\j in {0/1, 1/2, 4/5, 5/6}{
	\path (\a_\i) edge  (\a_\j) (\a_\i) edge [bend left=15] (\a_\j);\path (\a_\i) edge [bend right=15] (\a_\j);
	\path (\a_\i) edge [bend left=30] (\a_\j);\path (\a_\i) edge [bend right=30] (\a_\j);
}
\path (\a_0) edge [bend right=30](\a_2);
\foreach \i in {3,...,5}{
	\path (\a_0) edge [bend right=20](\a_\i);}
\draw [dotted,line width=1.2] (-3.4*11:3) -- (-4.6*11:3);
\end{scope}
}
\begin{scope}[rotate=91*3]
\foreach \i/\l/\c in {0.5/0, 1.5/1, 2.5/2,  4/3, 5.4/4, 6.4/5,7.4/6}{
	\node at (-\i*11:3) [draw=none] (3_\l) {};
}
\end{scope}
\foreach \a/\b in {1/0,2/1, 0/3}{
	\path (\a_6) edge [bend right=20](\b_2)  (\a_6) edge [bend right=15] (\b_3)
	(\a_6) edge [bend right=15](\b_4)
	(\a_6) edge [bend right=15] (\b_5);
}
\draw [dotted, line width=1.3] (195:3) arc (195:265:3);
\end{tikzpicture}
\caption{\small Suppose that there are exactly $1/\eta$ (black) vertices between any two vertices in the above figure, where each edge has fractional value $\eta$.  Also, any two consecutive vertices (if both of them are not blue) have exactly $1/\eta$ parallel edges between them. Then, it is easy to check that the above graph is fractionally $2$ edge connected.
Furthermore, the set of $2+O(\eta)$-near minimum cuts comprises a single connected component and every vertex will become an (outside) atom of the corresponding polygon. This is because every diagonal which separates two blue vertices on both sides is a near min cut. In addition, every interval with $O(1)$ many consecutive vertices where at most one of vertex is blue is also a near mincut. In such a case, for every pair of adjacent blue vertices, we have $E(a_i, a_{i+1}) = \emptyset$. }
\label{fig:nearmincutbadexample}
\end{figure}

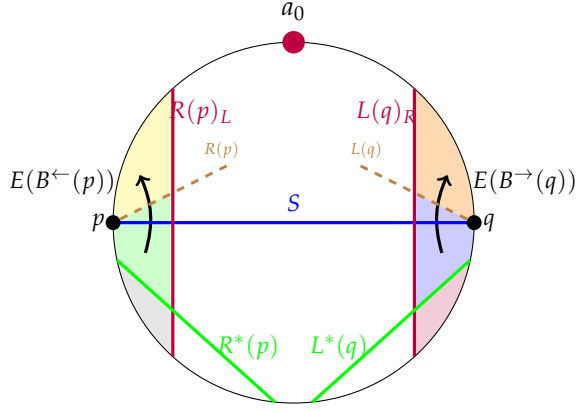
\begin{figure}[htb]\centering
	\begin{tikzpicture}[scale=0.8]
		\fill [color=gray!20](228:3) -- +(0,0.75)-- (192:3) arc (192:228:3);
		\fill [color=purple!20] (312:3) -- +(0,0.75) -- (-12:3) arc (348:312:3);
		\fill [color=yellow!30] (180:3) -- +(1,0.5) -- (132:3) arc (132:180:3);
		\fill [color=orange!30] (0:3) -- +(-1,0.5) -- (48:3) arc (48:0:3);
		\fill [color=green!20] (180:3) -- ++(1,0.5) -- +(0,-2) -- (192:3) arc (192:180:3);
		\fill [color=blue!20] (0:3) -- ++(-1,0.5) -- +(0,-2) -- (-12:3) arc (-12:0:3);
		\node at (195:2.6) (a1){}; \node at (160:2.8) (a2) {};
		\node at (-15:2.6) (b1) {};\node at (20:2.8) (b2){};
		\path [line width=1.2pt,->] (a1) edge [bend right=20] (a2);
		\path [line width=1.2pt,->] (b1) edge [bend left=20] (b2);
		\node at (10:3.9) () {\footnotesize $E(B^\rightarrow(q))$};
		\node at (170:3.9) () {\footnotesize $E(B^\leftarrow(p))$};
		\draw (0,0) circle (3);
		\node [circle,fill,purple,inner sep=3,label={[yshift=0cm]\small $a_0$}] at(90:3)() {};
		\draw [color=blue,line width=1.2pt] (0:3) -- node [above] {\footnotesize $S$} (180:3);
		\foreach \i/\c/\l in {180/black/p, 0/black/q}{
			\node [circle,inner sep=2,color=\c,fill] at (\i:3) (\l) {};
			\node at (\i:3.25) () {\footnotesize $\l$};
		}
		\draw [color=purple,line width=1.2pt] (132:3) --(228:3) (48:3) --(-48:3);
		\node[color=purple] at (130:2.4) () {\footnotesize$R(p)_L$};
		\node [color=purple] at (50:2.4) () {\footnotesize $L(q)_R$};
		\draw [color=brown,dashed,line width=1.1pt] (p) -- +(2,1) (q) -- +(-2,1);
		\node [color=brown]at (45:1.7) () {\tiny $L(q)$};
		\node [color=brown] at (135:1.7) () {\tiny $R(p)$};
		\draw [color=green,line width=1.2pt] (192:3) -- (264:3) (276:3) -- (-12:3);
		\node [color=green] at (250:2.2) () {\footnotesize $R^*(p)$};
		\node [color=green] at (290:2.2) () {\footnotesize $L^*(q)$};
	\end{tikzpicture}
	\caption{\small This figure illustrates the definitions in (\ref{defn:slackincrease}). The near min cut S contains all atoms below the blue line. $L(q)_R$ crosses $L(q)$ on the right and $R(p)_L$ crosses $R(p)$ on the left. $L(q)^{\cap R}$ is the set of atoms in the blue + pink region on the right. $L^*(q)$ is the near min cut crossing $L(q)^{\cap R}$ on the left that maximizes the number of outside atoms in the pink region. Similarly $R^*(p)$ is the near min cut crossing $R(p)^{\cap L}$ on the right that maximizes the number of outside atoms in the grey region. The edges in $E(B^\rightarrow (q))$ are those edges from the blue region to the orange region and the edges in $E(B^\leftarrow (p))$ are those edges from the green region to the yellow region. Note that the figure is misleading in that sets that are shown as disjoint here may not in fact be disjoint.}
	\label{fig:L*R*}
\end{figure}

\begin{figure}[htb]\centering
	\begin{tikzpicture}[scale=0.8]
		\fill [color=gray!20] (240:3) -- (0.4,0) -- (0:3) arc (360:240:3);

		\draw [color=purple!30,line width=6pt] (90:3.1) arc  (90:290:3.1);
		\draw [color=brown!50,line width=6pt] (20:3.1) arc (20:89:3.1);
		\draw [color=green,line width=1.2pt] (240:3) -- (48:3);
		\node at (0.85,-1) () {\footnotesize $L(q)^{\cap R}$};
		\node [color=green] at(60:1.6) () {\small $L(q)_R$};
		\draw [line width=1.4pt,->,color=green] (0,-0.5) -- +(0.3,-0.3);
		\node at (45:4) () {\small \text{$r'$ in brown}};
		\node at (210:4) () {\small $l'$ in purple};
		\draw (0,0) circle (3);
		\node [circle,fill,purple,inner sep=3,label={[yshift=0cm]\small $a_0$}] at(90:3)() {};
		\draw [color=red,line width=1.2pt] (0:3) -- node [above left] {\small $L(q)$} (180:3);
		\draw [color=red,->,line width=1.4pt] (-1,0) -- +(0,-0.4);
		\foreach \i/\c/\l in {0/black/q, 312/black/p, 18/red/b, 294/red/a, 240/black/l, 48/black/r}{
			\node [circle,inner sep=2,fill,color=\c] at (\i:3) (\l) {};
			\node at (\i:3.25) () {\footnotesize $\l$};		
		}
		\path [color=red] (a) edge (b); 
	\end{tikzpicture}
	\caption{\small  Setup for proof of \cref{claim:atmost2sets}: Let $L(q)_R = (l,r)$ and $L(p)_R=(l',r')$ (where $l,r,l',r'$ are polygon points). The grey region is $L(q)^{\cap R}:= L(q) \cap L(q)_R$. Note that neither $L(p)$ or $L(p)_R$ are shown in this figure, since our proof in fact will need to argue about how these cuts are situated relative to those shown.  WLOG (as shown in the figure) $p$ is to the left of $q$. Now, for contradiction, suppose that $e=\{a,b\}\in E(B^\rightarrow(p))\cap E(B^\rightarrow(q))$. Then $a,b\in L(p)_R\cap L(q)_R$. So, in the above, since no cut contains $a_0$, it must be that $l'$ is to the left of $a$ and $r'$ is to the right of $b$.}
	\label{fig:locl'r'circ}
\end{figure}
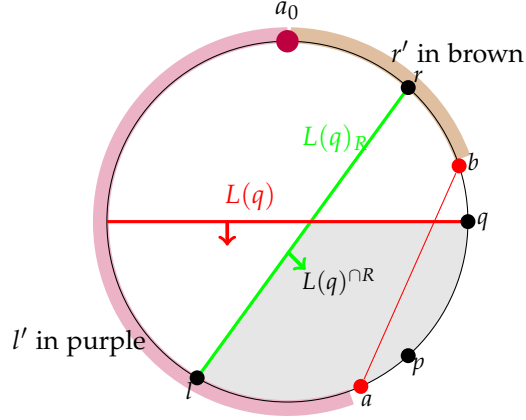

\def\scale{0.7}
\begin{figure}[htb!]\centering
\begin{tikzpicture}[scale=\scale]
	\draw [decorate,line width=1.3pt, decoration = {brace,mirror,amplitude=8pt}] (9,0.5) --node [above=0.3cm] {\small $L(q)^{\cap R}$}  (4,0.5);
		\draw [color=red!20,line width=9pt]	(9,0.15) -- (2,0.15);
		\draw [color=green!20,line width=9pt] (12,-0.15) -- (4,-0.15);
		\draw [color=purple!20,line width=9pt] (5.7,-1) -- node [below] {\small range of $l'$}(0.5,-1);
		\draw [color=brown!20,line width=9pt] (10.3,-1) -- node [below] {\small range of $r'$}(13.5,-1);
		\node [color=red] at (2.5,0.5) () {\small $L(q)$};
		\node [color=green] at (11.5,0.5) () {\small $L(q)_R$};
		\draw [color=black] (0,0) -- (14,0);
		\foreach \i in {0, 14}
			\node [color=purple,inner sep=4,label={[yshift=-0.9cm]\small root},circle,fill] at (\i,0) () {};
		\foreach \i/\c/\l/\is in {4/black/l/2, 6/blue/a/4, 7/red/p/2, 9/red/q/2, 10/blue/b/4, 12/black/r/2}
			\node [fill,circle,inner sep=\is,label={[yshift=-0.9cm,color=\c]\small $\l$},color=\c] at (\i,0) () {};
	\end{tikzpicture}
	\caption{\small  Proof of \cref{claim:atmost2sets} continued: Because we are dealing with outside atoms only and no cuts contain the root $a_0$, we may as well visualize the polygon as a line (with wraparound) and each cut as an interval along the line.  The above figure repeats \cref{fig:locl'r'circ} when viewed as a line segment. }
	\label{fig:locl'r'}

\begin{tikzpicture}[scale=\scale]
			\draw [color=red!20,line width=9pt]	(9,0.15) -- (2,0.15);
		\draw [color=green!20,line width=9pt] (12,-0.15) -- (4,-0.15);
		\draw [color=purple!20,line width=6pt] (5,-1) -- node [below] {$L(p)_R$}(13,-1);
		\foreach \i/\l in {5/l',13/r'}
			\node [circle,color=purple,fill,inner sep=3, label={[yshift=-0.8cm]\small$\l$}] at (\i,-1) () {};
		\node [color=red] at (2.5,0.5) () {$L(q)$};
		\node [color=green] at (11.5,0.5) () {$L(q)_R$};
		\draw [color=black] (0,0) -- (14,0);
		\foreach \i in {0, 14}
			\node [color=purple,inner sep=4,label={[yshift=-0.9cm]\small root},circle,fill] at (\i,0) () {};
		\foreach \i/\c/\l/\is in {4/black/l/2, 6/blue/a/4, 7/red/p/2, 9/red/q/2, 10/blue/b/4, 12/black/r/2}
			\node [fill,circle,inner sep=\is,label={[yshift=-0.9cm,color=\c]\small $\l$},color=\c] at (\i,0) () {};
	\end{tikzpicture}
	\caption{\small  Claim: $l'$ can not be to the right of $l$. If it is, as in the figure above, then $L(p)_R$ crosses $L(q)$ on the right and has a smaller intersection with $L(q)$. Contradiction to choice of $L(q)_R$!}
	\label{fig:l'notR}
	
	\begin{tikzpicture}[scale=\scale]
			\draw [color=red!20,line width=9pt]	(9,0.15) -- (2,0.15);
		\draw [color=green!20,line width=9pt] (12,-0.15) -- (4,-0.15);
		\draw [color=purple!20,line width=6pt] (3,-1) -- node [below] {$L(p)_R$}(13,-1);
		\draw [color=gray,line width=6pt] (7,-1.8) -- node[below] {$L(p)$} (1.5,-1.8);
		\foreach \i/\l in {3/l',13/r'}
			\node [circle,color=purple,fill,inner sep=3, label={[yshift=-0.8cm]\small$\l$}] at (\i,-1) () {};
		\node [color=red] at (2.5,0.5) () {$L(q)$};
		\node [color=green] at (11.5,0.5) () {$L(q)_R$};
		\draw [color=black] (0,0) -- (14,0);
		\foreach \i in {0, 14}
			\node [color=purple,inner sep=4,label={[yshift=-0.9cm]\small root},circle,fill] at (\i,0) () {};
		\foreach \i/\c/\l/\is in {4/black/l/2, 6/blue/a/4, 7/red/p/2, 9/red/q/2, 10/blue/b/4, 12/black/r/2}
			\node [fill,circle,inner sep=\is,label={[yshift=-0.9cm,color=\c]\small $\l$},color=\c] at (\i,0) () {};
	\end{tikzpicture}
	\caption{\small Claim: $l'$ can not be to the left of $l$. If so, as in the figure above, then $L(q)_R$ crosses $L(p)$ on the right and has a smaller intersection with $L(p)$. Contradiction to choice of $L(p)_R$! Therefore $l'=l$.}
	\label{fig:l'notL}
	\begin{tikzpicture}[scale=\scale]
	\draw [decorate,line width=1.3pt, decoration = {brace,mirror,amplitude=8pt}] (9,0.5) --node [above=0.3cm] {$L(q)^{\cap R}$}  (4,0.5);
		\draw [color=red!20,line width=9pt]	(9,0.15) -- (2,0.15);
		\draw [color=green!20,line width=9pt] (12,-0.15) -- (4,-0.15);
		\draw [color=gray,line width=9pt] (7,-1) -- node[below] {$L(p)$} (2.5,-1);
		\node [color=red] at (2.5,0.5) () {$L(q)$};
		\node [color=green] at (11.5,0.5) () {$L(q)_R$};
		\draw [color=black] (0,0) -- (14,0);
		\foreach \i in {0, 14}
			\node [color=purple,inner sep=4,label={[yshift=-0.9cm]\small root},circle,fill] at (\i,0) () {};
		\foreach \i/\c/\l/\is in {4/black/l/2, 6/blue/a/4, 7/red/p/2, 9/red/q/2, 10/blue/b/4, 12/black/r/2}
			\node [fill,circle,inner sep=\is,label={[yshift=-0.9cm,color=\c]\small $\l$},color=\c] at (\i,0) () {};
	\end{tikzpicture}
	\caption{\small  Since the left endpoint of $L(p)_R$ is  $l$, the left endpoint of $L(p)$ is to the left of $l$. Therefore, $L(p)$ crosses $L(q)^{\cap R}$ on left and it is a candidate for $L^*(q)$. Therefore, $L(p)\cap L(q)^{\cap R}\subseteq L^*(q)\cap L(q)^{\cap R}$ and we have $a\in L^*(q)\cap L(q)^{\cap R}$. But then, $e\notin E^(B^\rightarrow(q))$. Contradiction!.}
	\label{fig:atmost2finalcont}
\end{figure}

\subsection{Overview of proof of \cref{thm:informalmain} -- no inside atoms}
\label{sec:noinside}

Given a connected component $\C$ of cuts in $\cN_\eta$, we can partition vertices of $G$ into sets $a_0,\dots,a_{m-1}$ (called atoms); this is the coarsest partition such that for each $a_i$, and each $(S,\overline{S})\in {\cal C}$, we have $a_i\subseteq S$ or $a_i\subseteq \overline{S}$.  One of these atoms,  $a_0$ is the atom that contains $u_0,v_0$. We call $a_0$ the {\em root}.  \textit{In the following, we will often identify an atom with the set of vertices that it represents}\footnote{For example, it will be convenient to write cuts as subsets of atoms. In this case the cut is the union of the vertices in those atoms.}.

If $\eta=0$, then  \cite{DKL76}) shows that the structure of cuts in $\C$ can be represented by a cycle; namely we can arrange these atoms around a cycle such that, perhaps after renaming, for any $0\leq i\leq m-1$, $x(E(a_i,a_{i+1 \text{ mod } m}))=1$ and cuts of $\C$ are just the mincuts of this cycle.

As mentioned, \cite{Ben95,BG08} studied the case when $0 <\eta\leq 2/5$ and introduced the notion of {\em polygon representation}, in which case atoms can be placed on the sides of an equilateral polygon $P$ and some atoms  placed inside  the polygon, such that every cut in ${\cal C}$ can be represented by a diagonal of this polygon. See \autoref{fig:polygonrepresentation}. 

{\em In the rest of this section, we fix $\C$ and we outline the ideas behind the proof of \cref{thm:informalmain} in the special case that the polygon $P$ representing the connected component of cuts $\C$ contains  no inside atoms}. This latter assumption simplifies the argument but still illustrates many of the main ideas. 

We assume that the atoms of $P$ are labelled counterclockwise from $a_0$ to $a_{m-1}$. We associate to each diagonal (defining a cut) the side which does not contain $a_0$.\footnote{The reason we do this is that it is crucial for subsequent arguments to be able to condition on near min cuts being trees using \cref{lem:treeconditioning}, i.e., that for $S$ a near min cut, $E(S) \cap T$ is very likely to be a tree. However, this lemma can only be used on sets which do not contain $u_0,v_0$.} Thus, we will refer to a cut by the set of outside atoms it contains, say $[a_i,a_j]$, $i < j$. (This denotes the side of the diagonal containing the atoms $a_i, a_{i+1}, \ldots, a_j$.) We equivalently refer to this cut by giving the left and right polygon points of its diagonal $[p_{i-1},p_j]$.

As mentioned above, the raison d'etre for the slack vector $s^*$  that we construct here is to restore the feasibility of cuts $S$ in $\cC_2$ which are odd in the tree but which have $s_e$ negative on some edges in $\delta(S)$.  The high level approach in the proof is the following. Initialize $s_e^*:= 0$ for all $e$. Now define a set of {\em bad events} whose occurrence signifies that some  of these near min cuts are potentially in need of such a repair. These bad events should satisfy the follow desiderata:
\begin{itemize}
\item[(a)] Each bad event occurs with probability $O(\eta)$, where the probability is taken over the choice of tree $T$.
\item [(b)]  The occurrence of a bad event $B$ in a tree $T$ triggers a {\em slack increase} on an associated set of edges $E(B)$. Specifically, when $B$ occurs, each edge $e \in E(B)$  has its slack $s_e^*$ increased by $\alpha x_e$. 
\item [(c)] Each edge $e$  is in $E(B)$ for a {\em constant number of bad events} $B$. Combining (a) and (b), this implies that $\E{s_e^*} = O(\eta\alpha x_e)$ (condition 2 of \cref{thm:informalmain}).
\item[(d)] Each $\eta$-near-min cut $S$ is associated with a constant number of bad events ${\mathcal B}(S)$, such that when $\delta(S)_T$ is odd, at least one of the bad events $B \in {\mathcal B}(S)$ occurs. We will ensure that the edges in $E(B)$ (on which slack increases are triggered) are a subset of $\delta(S)$ of fractional value at least   $\Omega(1)$. Therefore, if $S$ is odd in the tree, $s^*_T(\delta(S)) \ge \alpha x(E(B)) \ge \Omega(\alpha)$ implying condition (iii) of \cref{thm:informalmain} (once the constant  are chosen appropriately).

\end{itemize}

\subsection{Satisfying the above desiderata}
Consider any near min cut $S$ in $P$, which is crossed on the left and on the right (see \cref{def:SLR} for the definition of being crossed on left/right) .  Let $S_L$ and $S_R$ be the cuts crossing $S$ on the left and right {\em with minimum sized intersection} with $S$. See \autoref{fig:S-SR-SL}.

One of the very nice things about cuts crossed on both sides is the following:

\begin{claim}\label{claim:C2evenwhp}
For any near min cut $S \in \cC_2$, $\P{\delta(S)_T= 2} \ge 1-O(\eta)$.
\end{claim}
\begin{proof}[Proof sketch]
To see this, for a set $S$ crossed on both sides, let
$$E^{\leftarrow} (S) := E(S\cap S_L, S_L \smallsetminus S), \quad E^{\rightarrow} (S) := E(S\cap S_R, S_R\smallsetminus S),\quad E^\circ(S) := \delta(S) \smallsetminus E^{\leftarrow} (S) \smallsetminus E^{\rightarrow} (S)$$
and consider the bad events
\begin{equation}\label{eqn:badfirst} B^\leftarrow(S) := \mathbb{1}\{E^{\leftarrow} (S)_T\ne 1\}\quad B^\rightarrow(S)=\mathbb{1}\{E^{\rightarrow} (S)_T\ne 1\}\quad B^\circ(S):= \mathbb{1}\{E^\circ (S)_T\ne 0\}.	
\end{equation}
 See \autoref{fig:S-SR-SL}. 

Clearly if none of these bad events occur, then $S$ is even in the tree (i.e., $\delta(S)_T = 2$).
Now, note that $S, S_R, S_L$ are all $\eta$-near min cuts and so by \cref{lem:nmcuts_largeedges} and \cref{lem:treeoneedge}, we have $x(E^\leftarrow (S)) \ge  1-\eta/2$, $x(E^\rightarrow (S)) \ge  1-\eta/2$, $x(E^\circ(S)) = x(\delta (S)\smallsetminus E^\leftarrow (S)\smallsetminus E^\rightarrow (S)) = O(\eta)$ and $\P{B^\leftarrow(S)}, \P{B^\rightarrow(S)}, \P{B^\circ(S)} = O(\eta).$
\end{proof}

The next step in our plan is to decide what slack increases are triggered by these bad events. The first thing one might think of is to have the above bad events  (\ref{eqn:badfirst}) trigger a slack increase on $E^\leftarrow (S) \cup E^\rightarrow (S)$. Namely, for each set $S$ crossed on both sides:
$$\forall e \in  E^\leftarrow (S) \cup E^\rightarrow (S)\quad \text{ set }s^*_e :=
 	\alpha x_e \cdot \mathbb{1}\{\text{at least one of } B^\leftarrow(S), ~B^\rightarrow(S) \text{ or } B^\circ(S)\text{ occurs}\}.$$
 	In addition to desiderata (a) and (b), this approach satisfies (d) since $x(E^\leftarrow (S) \cup E^\rightarrow (S))\ge 2-\eta$.
 
 Unfortunately though, this  does not satisfy desiderata (c), since if $e \in E(a_i, a_j)$, it could be that $e \in  E^\leftarrow (S) \cup E^\rightarrow (S)$ for {\em many} near min cuts $S$ 
     in which case  $$\E{s^*_e} = \alpha x_e \cdot\P{\exists S\text{ odd in }T\text{ s.t. }e \in E^\leftarrow (S) \cup E^\rightarrow (S))}$$ could be way too large (say, around $\alpha x_e$).

So, rather than defining a bad event for every cut $S$ crossed on both sides individually (i.e., up to $O(m^2)$ events), we instead define a constant number of bad events for each polygon point $p$, hence at most $O(m)$ events.

\subsubsection{Defining bad events for each polygon point}

For a fixed polygon point $p$, 
let $L := L(p)$ be the set crossed on both sides that extends farthest clockwise from $p$ and as above, let $L_R$ be the cut that crosses it on the right with the minimum number of outside atoms in the intersection.  Analogously define $R:= R(p)$ and $R_L$. See \autoref{fig:LpRp}.

Now we consider two bad events: 
\begin{equation}
\label{defn:badevents}
\begin{aligned}
	B^\rightarrow (p) &= \mathbb{1}\{E^{\rightarrow} (L(p))_T\ne 1\text{ or }E^\circ (L(p))_T\ne 0\}
	\\
B^\leftarrow (p) &= \mathbb{1}\{E^{\leftarrow} (R(p))_T\ne 1\text{ or }E^\circ (R(p))_T\ne 0\}. 
\end{aligned}
\end{equation}

For these events, we have the following two claims:
\begin{claim}\label{claim:1}
	 For any near min cut $S=[p,q]$, $E^{\rightarrow} (L(q)) = E^\rightarrow (S)$ and $E^{\leftarrow} (R(p)) = E^\leftarrow (S)$. Moreover, $E^\circ (S) \subset E^\circ (L(p))\cup E^0 (R(p)$. See \autoref{fig:ErL=ErS}. Therefore, if neither $B^\rightarrow (q)$ or $B^\leftarrow (p)$ occur, then $\delta(S)_T$ is even.  
\end{claim}

In addition we have
\begin{claim}\label{claim:2}
	 For any polygon point $p$, $\P{B^\rightarrow (p)}, \P{B^\leftarrow (p)} = O(\eta)$.
\end{claim}
This follows arguments similar to those used in \cref{claim:C2evenwhp}, using that $x(E^\rightarrow (L(p)))\ge 1- \eta/2$, and $x(E^\circ (L(p)) = O(\eta)$ (and similarly for $R(p)$).


 
 These bad events satisfy the desiderata (a) and (d) (assuming we define $E(B)$ such that $x(E(B)) \in  \Omega(1)$). 
 
\subsubsection{Defining the slack increase sets for bad events} 

It remains to determine the sets $E(B^\rightarrow(p)),E(B^\leftarrow(p))$  for which slack increases are triggered when the bad events occur. In particular, we will let $E(B^\rightarrow(p)) \subseteq E^\rightarrow(L(p))$ and $E(B^\leftarrow(p)) \subseteq E^\leftarrow(R(p))$ such that:
\begin{enumerate}
\item[(*)] $x(E(B^\rightarrow(p))) \ge \Omega(1)$ and $x(E(B^\leftarrow(p))) \ge \Omega(1)$ (to guarantee (d)),
\item[(**)] All edges $e$ are in at most a constant number of sets $E(B)$ (to guarantee (c)). 
\end{enumerate}
Assuming we can satisfy (*) and (**), we can set $s^*_e = \alpha x_e$ for all $e \in E(B)$ when $B$ occurs to satisfy all four desiderata.

\paragraph{First try:} The most natural choice is to simply let $E(B^\rightarrow(p)) = E^\rightarrow (L(p))$. Here, (*) obviously holds but unfortunately (**) fails. 
 	Indeed, there are examples (see \autoref{fig:einmanyL(p)}) for which there exist edges $e\in E(a_i,a_j)$ with $|j-i|=\Omega(m)$ that belong to 
 	$E^\rightarrow(L(p_k))$ for $\Omega(m)$ many values of $i\leq k\leq j$.

\paragraph{Second try:} 
Let $a_i$ be the atom immediately to the left of $p$ and $a_{i+1}$ the atom immediately to the right of $p$ (i.e. $p=p_i$). Note that all edges with one endpoint in $a_i$ and one in $a_{i+1}$ are in $E^\rightarrow (L(p))$. Now if it was always the case that $x(E(a_i, a_{i+1}))\geq \gamma$ for some universal constant $\gamma > 0$, then, when one of these bad events occurs, say $B^\rightarrow(p)$, we could simply increase the slack of every edge $e$ in $E(a_i, a_{i+1})$ by $\alpha x_e$. 
This approach is analogous to the method employed in \cite{KKO21} where slack was increased on OPT edges. One might have some hope that this is true since it holds with $\gamma = 1$ for the cactus representation of min cuts (i.e. when $\eta = 0$). 

Unfortunately, as observed in \cite{OSS11} there is a family of near minimum cuts such that the polygon representation has no inside atoms, yet $E(a_i,a_{i+1})=\emptyset$ for some consecutive pairs of (outside) atoms (see \autoref{fig:nearmincutbadexample}) (even though there are cuts whose diagonals end between those atoms). So, this method is doomed even if inside atoms are not present.

\paragraph{Our method:} 
The first try works if there are no ``long'' edges. So, to rectify that attempt we essentially ``ignore'' long edges (edges between distant atoms) in our charging argument and argue that they only contribute minimally to $E^\rightarrow(L(p))$ and $E^\leftarrow(R(p))$. 

To this end, define $L(p)^{\cap R} := L(p) \cap L(p)_R$, and let $L^*(p)$ be the cut crossing $L(p)^{\cap R}$ on the left that {\em maximizes} the number of outside atoms in the intersection of $L^*(p)$ and $L(p)^{\cap R}$ (and similarly $R^*(p)$ to maximize the intersection with $R(p)^{\cap L}$ on the right). If $L^*(p)$ does not exist, i.e. no cut crosses $L(p)^{\cap R}$ on the left, set $L^*(p) = \emptyset$, and similarly for $R^*(p)$. See  \autoref{fig:L*R*}. We let:
\begin{equation}
\begin{aligned}
	E(B^\rightarrow (p)) &:= E( L(p)^{\cap R} \smallsetminus L^*(p), L(p)_R \smallsetminus  L(p)^{\cap R})  \\
	E(B^\leftarrow (p)) &:= E( R(p)^{\cap L} \smallsetminus R^*(p), R(p)_L \smallsetminus  R(p)^{\cap L}).
\end{aligned}
\label{defn:slackincrease}
\end{equation}


The following claim establishes (*) for \cref{defn:slackincrease}. It can be proved using methods similar to \cref{claim:C2evenwhp}; see \autoref{fig:ErL=ErS}.
\begin{claim}\label{claim:xvaluenoinside}
	For all polygon points $p$, $x(E(B^\rightarrow (p))), x(E(B^\rightarrow (p)))\ge 1-O(\eta)$.
\end{claim}

And finally, the following claim establishes (**):

\begin{claim}\label{claim:atmost2sets}
	For any edge $e$, we have $e\in E(B^\rightarrow (p))$ for at most one polygon point $p$ and similarly  $e\in E(B^\leftarrow (q))$ for at most one polygon point $q$.
\end{claim}

The proof of \cref{claim:atmost2sets} is more involved. For an outline of the arguments used, see Figures \ref{fig:locl'r'circ}, \ref{fig:locl'r'}, \ref{fig:l'notR}, \ref{fig:l'notL}, \ref{fig:atmost2finalcont}.

	


	

	

%

\subsection{Extending to polygons with inside atoms}

For the general case, we still follow the same proof outline satisfying the four desiderata using the same definition of bad events given in (\ref{defn:badevents}) and the same definition of edges on which slack is increased in response to bad events given in (\ref{defn:slackincrease}).

We need to address the following challenges when generalizing the above proof:
\begin{enumerate}[(1)]
\item Edges may have endpoints adjacent to inside atoms. The proof outline above crucially used that both endpoints of every edge were outside atoms. 
\item Regions of the polygon we could previously assume were empty may now contain inside atoms.
\item While the sets are still defined such that they do not contain the root, the root (i.e., the atom containing $\{u_0, v_0\}$) is no longer necessarily an outside atom. 
Therefore we can have a sequence of cuts not containing the root wrapping around the polygon such that each cut crosses the one before it. In this case, the notion of ``left" and ``right" becomes unclear. One can still define ``left" and ``right" synonymously with clockwise and counterclockwise, but we can no longer say that an outside atom is to the left (and not to the right of) another outside atom, nor can we collapse the diagonals of the polygon to intervals of a line. 
\end{enumerate}
To handle these complexities, we introduce additional structural properties of polygons with inside atoms. These are presented in \autoref{sec:polygonsnewproperties}.

\section{Polygon Representation}
\label{sec:polyrep}

\subsection{Basics}
\label{sec:polybasics}
In this section we state and prove several facts for $\eta$-near minimum cuts of a fractionally 2-edge-connected graph. All of the statements in this section are generalizable to the $(1+\eta)\alpha$ near minimum cuts of an $\alpha$-edge-connected graph (by rescaling) for $\eta \le 6/5$.  
\begin{definition}[Connected Component of Crossing Cuts]\label{def:conn-component-cuts}
	Given the set of $\eta$-near min cuts of a graph $G=(V,E)$, construct a graph where two cuts are connected by an edge if they cross. Partition this graph into maximal connected components. In the following, we will consider maximal connected components $\cC$ of crossing cuts and simply call them {\em connected components}.
	We say a connected component is a {\em singleton} if it has exactly one cut and a {\em non-singleton} otherwise. 
	
	For a connected component $\C$, let $\{a_i\}_{i \ge 0}$ be the coarsest partition of vertices $V$ such that for any $C\in \C$, either $a_i\subseteq C$ or $a_i\subseteq \overline{C}$. Each set $a_i$ is called an {\em atom} of $\C$ and we write  ${\cA}({\cal C})$ to denote the set of all atoms.
	
	Note for any atom $a_i \in \cA(\cC)$ which is an $\eta$-near min cut, $(a_i,\overline{a_i})$ is a singleton component, and is not crossed by any $\eta$-near min cut. Therefore $(a_i,\overline{a_i}) \not\in \cC$. 
	
	We can now represent any cut in $S \in \cC$ either by the set of vertices it contains or as a subset of $\cA(\cC)$. \textbf{In the following, we will often identify an atom with the set of vertices that it represents\footnote{For example, it will be convenient to write cuts as subsets of atoms. In this case the cut is the union of the vertices in those atoms.}.}
\end{definition}

 To study these systems, we will utilize the polygon representation of near minimum cuts defined in \cite{Ben97} and then extended in \cite{BG08}. Their work implies that any connected component $\mathcal{C}$ of crossing $\eta$-near minimum cuts has a polygon representation with the following properties, so long as $\eta \le \frac{2}{5}$: 

\begin{figure}
\begin{center}
\begin{tikzpicture}[inner sep=1.7pt,scale=.7,pre/.style={<-,shorten <=2pt,>=stealth,thick}, post/.style={->,shorten >=1pt,>=stealth,thick}]
\tikzstyle{every node} = [draw, circle,color=black];
\begin{scope}[shift={(-5,0)}]
\foreach \i in {2,...,9}{
\path (\i*45:3) node  (a_\i) {\i};
}
\foreach \i in {0,...,7}{
\draw [color=blue,dashed,rotate around={22.5+45*\i:(22.5+\i*45:2.8)},line width=1.2] (22.5+\i*45:2.8) ellipse (.5 and 1.7);
}
\path (0,0) node (c) {1};
\foreach \a/\b in {2/3, 3/4, 4/5, 5/6, 6/7, 7/8, 8/9, 9/2}{
\path (a_\a) edge  (a_\b) edge [bend left=15] (a_\b) edge [bend right=15] (a_\b) edge (c);
}
\end{scope}
\begin{scope}[shift={(5,0)}]
\foreach \i in {2,...,9}{
\draw (22.5+\i*45:4) -- (22.5+45+\i*45:4);
\path (\i*45:3.4) node  (a_\i) {\i};
\draw [color=blue] (22.5+\i*45:4) -- (22.5+90+\i*45:4);
}
\path (0,0) node (c) {1};
\end{scope}
\end{tikzpicture}
\end{center}
\caption[A Family of near Minimum Cuts with an Inside Atom]{
Consider the graph on the left and suppose that every edge $e$ has fractional value $x_e =1/7$. This graph then has min cut value 2, with cuts of fractional value at most $2 + 1/7$ circled in blue. Note that this is a connected family $\C$ of near-min cuts, since every adjacent pair of blue cuts cross each other. The right image shows the polygon representation of $\C$. The blue lines in the right image are the representing diagonals. This representation has 8 outside atoms and $\{1\}$ is the only inside atom. The system of near minimum cuts corresponding to sets $ (\{2,3\}, \overline{\{2,3\}}), \ldots, (\{8,9\},\overline{\{8,9\}}), (\{9,2\}, \overline{\{9,2\}})$ shows an $8$-cycle for the inside atom $\{1\}$. See \autoref{def:kcycle}.}
\label{fig:polygonrepresentation}
\end{figure}
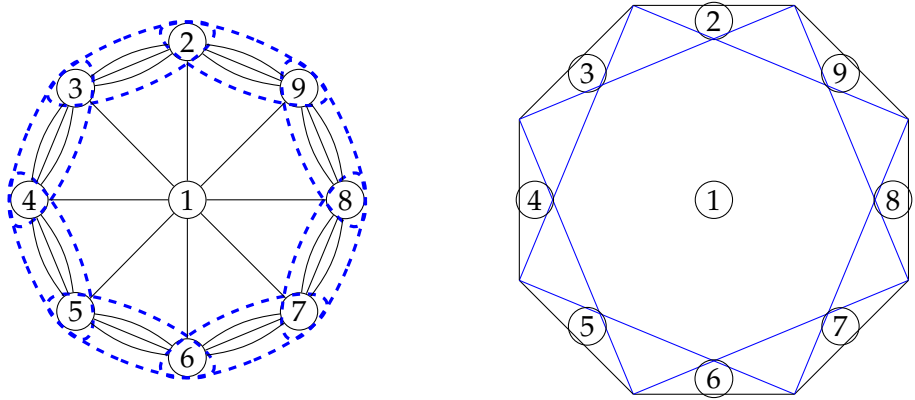

\begin{enumerate}
	\item A polygon representation is a convex regular polygon with a collection of \textit{representing diagonals}. All polygon edges and diagonals are drawn using straight lines in the plane. The diagonals partition the polygon into \textit{cells}. 
	\item Each atom of $\C$ is mapped to a cell of the polygon.  If one of these cells is bounded by some portion of the polygon boundary it is {\em non-empty} and we call its atom an \textit{outside atom}. We call the atoms of all other non-empty cells \textit{inside atoms}. Note that some cells may not contain any atom. WLOG label the outside atoms $a_0,\dots,a_{m-1}$ in counterclockwise order, and label the inside atoms arbitrarily. We also label points of the polygon $p_0,\dots,p_{m-1}$ such that outside atom $a_i$ is on the side $(p_i,p_{i+1})$ and $a_0$ is on the side $(p_{m-1},p_0)$. (In future sections we will refer to the special atom called the root, and if it is an outside atom WLOG we will label $a_0$ as the root.)
	\item No cell has more than one incident outer polygon edge.
	\item Each representing diagonal defines a cut such that each side of the cut is given by the union of the atoms on each side. Furthermore, the collection of cuts given by these diagonals is exactly $\mathcal{C}$. 
\end{enumerate}

The following fact follows immediately from the above discussion:
\begin{fact}
\label{fact:atLeast2Outside}
Any cut $S\in \cC$ (represented by a diagonal of $P$) must have at least two outside atoms.
\end{fact}

\begin{definition}[Outside atoms]
	For a polygon $P$ and a set $S$ of atoms of $P$, we write $O_P(S)$ to denote the set of outside atoms of $P$ in $S$; we drop the subscript when $P$ is clear from context. We also write $O(P)$ (or $O(\cA(\cC))$ where $\cC$ is the connected component of $P$) to denote the set of all outside atoms of $P$.
	
	Note that, given $S \in \cC$, since $S$ may be identified with a set of atoms, $O(S)$ is also well defined. 
\end{definition}

The following observation follows from  the fact that cuts correspond to straight diagonals in the plane and the polygon $P$ is regular:
\begin{observation}[{\cite[Prop 19]{BG08}}] \label{obs:Ocross}
	If $S,S'\in {\cal C}$ cross then $O(S)$ and $O(S')$ cross, and $O(S \cup S') \not= O(P)$.
\end{observation}

\cref{lem:mincutBshowsup} is used in the proof of \cref{lem:extendedprop20}, and depends on the following:

\begin{theorem}[{\cite[Lem 4.1.7]{Ben97}}]\label{thm:BenCC'}
	Let $\C,\C'$ be two (distinct) connected components of crossing cuts for a family of cuts of $G=(V,E)$. Then, there exists an atom $a\in \cA(\C)$ and $a'\in \cA(\C')$ such that $a\cup a'=V$. 
\end{theorem}

\begin{lemma}\label{lem:mincutBshowsup}
Consider the set of $\eta$-near minimum cuts (NMCs) of $G$ and let $\C$ be a connected component. Let $B\subset\cA(\C)$ be an $\eta$-NMC such that $1<|B|<|\cA(\C)|-1$. Then, $B\in \C$. 
\end{lemma}
\begin{proof}
For the sake of contradiction, suppose $B\notin \C$. Since $B$ is an $\eta$-NMC, it is in some connected component of cuts, say $B\in \C'$, where $\C\ne \C'$. Then, by \cref{thm:BenCC'} there exist atoms $a\in\cA(\C),a'\in\cA(\C')$ such that $a\cup a'=V$.
Observe that by the definition of atoms, each of $a,a'$ is contained in either $B$ or $\overline{B}$. Without loss of generality assume $a\subseteq B$ and $a'\subseteq \overline{B}$. But then, $a\neq B$ since $|B|>1$. So, $a\cup a'\neq V$ which is a contradiction. 
\end{proof}

\subsection{Our polygon notation and the root}\label{sec:polynotation}

All statements in the previous section do not depend on which side of each diagonal we consider. The ambiguity on the sides of the cut considered makes it difficult to define a consistent orientation on the polygon, for example to say whether a cut $A,\overline{A}$ crosses $B,\overline{B}$ ``on the left" or ``on the right." Motivated by this, we identify every cut with the side that does not contain $u_0,v_0$. This has the added benefit of allowing us to apply \cref{lem:treeconditioning} to every cut considered. For every polygon $P$, we call the unique atom $r$ containing $u_0,v_0$ \textit{the root}.  

Recall that, for $\eta>0$, we write $\cN_\eta  \subseteq 2^{V \smallsetminus \{u_0,v_0\}}$ to denote the family of all $\eta$-NMCs of $G_{/e_0}$. 
If $\eta$ is clear in context, we drop the subscript of $\cN_\eta$.
Throughout the paper, we will need to show that various sets $A,B \subseteq V\smallsetminus \{u_0,v_0\}$ cross. Since $u_0, v_0 \in \overline{A \cup B}$ to verify that a pair of sets $A$ and $B$ cross, it suffices to check the three conditions in the following fact.
\begin{fact}\label{fact:nou0v0cross}
For $A,B\subseteq V\smallsetminus\{u_0,v_0\}$, $A,B$ cross iff
$$A\cap B, A\smallsetminus B, B\smallsetminus A\neq\emptyset.$$	
\end{fact}

As above, unless otherwise specified, we let $\cC$ be a connected component  of cuts in $\cN_\eta$ with corresponding  polygon $P$ for $\eta \leq 1/10$. Again, call the {\em outside} atoms (of $P$) $a_0,\dots,a_{m-1}$, ordered counter-clockwise though these are not necessarily all the atoms in $P$. Note that the root $r$ is not necessarily an outside atom, but if it is, it is the atom labelled $a_0$. 

We use the existence of the root to prove the following two facts. The first fact is a consequence of \cref{obs:Ocross}:

\begin{fact}\label{fact:union-not-everything}
Let $S,S' \in \cC$. Then, $O(S \cup S') \not= O(\cA(\cC))$. 
\end{fact}
\begin{proof}[Proof sketch]
If $S,S'$ are (the non-root) sides of two diagonals of a polygon and there is an atom $r$ which is in neither of them, then there must be a polygon point which is in neither of them. 
\end{proof}

The following lemma is the main reason why we can treat each polygon separately in constructing the slack vector mentioned in \cref{sec:overview}. In \cref{sec:twoside}, we are careful to only define positive slack on edges of a polygon that do not have an endpoint in its root.

\begin{fact}\label{fact:edges-internal-to-one-poly}
For any edge $e=\{u,v\}$, there is at most one polygon in which the endpoints of $e$ lie in two different atoms which are not the root of their respective polygons.
\end{fact}
\begin{proof}
	Suppose not, and let $P,P'$ be two polygons in which $u$ and $v$ lie in different atoms that do not contain $r$. By \cref{thm:BenCC'}, there exists an atom $a \in P,a' \in P'$ such that $a \cup a' = V$. Since $u,v$ lie in different atoms (and the atoms of a polygon partition $V$) in both $P,P'$ it must be that (WLOG) $u \in a, v \in a'$. However, since $a \cup a' = V$, $u_0,v_0$ lies in $a$ or $a'$, so either $a$ is the root of $P$ or $a'$ is the root of $P'$, which is a contradiction. 
\end{proof}


We will use {\em``left" synonymously with ``clockwise"} and {\em``right" synonymously with ``counter-clockwise."}

\begin{definition}[Near min cut notation]
	We will interchangeably refer to a set $S\in\cN_\eta $  by specifying the extreme outside atoms it contains or by specifying the polygon points defining its diagonal. For the former, if $S=(a_l,a_r)$, then  $a_l$ is the leftmost outside atom in $S$ and $a_r$ is the rightmost outside atom in $S$. For the latter, if $S=(p_l,p_r)$, then $p_l$ is the polygon point immediately to the left of $a_l$ and $p_r$ is the polygon point immediately to the right of $a_r$. 
\end{definition}



\begin{definition}[$L(p)$, $R(p)$]
For a polygon point $p_i$, let $L(p_i)$  be the largest cut in $\cN_\eta$ containing $a_{i}$ and not $a_{i+1}$ which is crossed on both sides. Let $R(p_i)$ be the largest cut in $\cN_\eta$ containing $a_{i+1}$ and not $a_{i}$ which is crossed on both sides. (Note that $L(p_i),R(p_i)$ do not necessarily exist). See \autoref	{fig:LpRp}.
\end{definition}


%
%
%
%
%
%
%
%
%
%

The following definitions make formal the notion of ``crossed on one side" and ``crossed on both sides" (introduced in \cref{sec:overview}) for polygons with inside atoms.

\begin{definition}[Left, Right Crossing]
Let $S,S'\in \cC$ such that $S'$ crosses $S$.
For such a pair, we say \textit{$S'$ crosses $S$ on the left} if the leftmost (clockwise-most) outside atom of $O(S' \cup S)$ is in $S'$. 
 Otherwise, we say that \textit{$S'$ crosses $S$ on the right}. Note that by \cref{obs:Ocross}, $O(S),O(S')$ cross. 
\end{definition}
\begin{definition}[Crossed on one, both sides]
\label{defn:crossedonetwo}
We say a cut $S$ is \textit{crossed on both sides} if it is crossed by a cut (in $\cC$) on the left and a cut (in $\cC$) on the right and we say $S$ is {\em crossed on one side} if it is crossed only on the left or only on the right. 
\end{definition}

\begin{definition}[$\cN_{\eta,\le 1},\cN_{\eta,1},\cN_{\eta,2}$]
	Let $\cN_{\eta,2} \subseteq \cN_\eta$ be the set of cuts which are crossed on both sides in their respective polygons. Let $\cN_{\eta,1} \subseteq \cN_\eta$ be the set of cuts that are crossed on one side in their respective polygons and finally let $\cN_{\eta,\le 1} = \cN_\eta \smallsetminus \cN_{\eta,2}$ (i.e. the set of cuts which are crossed on one side or not crossed at all).
\end{definition}

Here we give an alternate set-theoretic characterization of $\cN_{\eta,2}$.

\begin{lemma}\label{lem:characterize-N2}
Let $C \in \cN_\eta$. Then, $C \in \cN_{\eta,2}$ if and only if there exist two cuts $A,B \in \cN_\eta$ which cross $C$ such that $(A \smallsetminus C) \cap (B \smallsetminus C) = \emptyset$. 	
\end{lemma}
\begin{proof}
The only if follows from \cref{lem:nointersectionLR}. So, assume there exist two cuts $A,B \in \cN_\eta$ which cross $C$ such that $(A \smallsetminus C) \cap (B \smallsetminus C) = \emptyset$. We will show $C \in \cN_{\eta,2}$. 

Since $A$ and $B$ both cross $C$, it must be that $C,A,B$ are in the same connected component of cuts $\cC \subseteq \cN_\eta$ (i.e. including all cuts in $\cN_{\eta,2}$). Let $P$ be the corresponding polygon.
	
	Suppose by way of contradiction that $A,B$ both crossed $C$ on the left (if they both cross on the right the argument is similar). Then $A,B$ must both contain the outside atom immediately to the left of the leftmost atom of $C$, which contradicts $(A \smallsetminus C) \cap (B \smallsetminus C) = \emptyset$.
\end{proof}

Consequently, we can give an alternate characterization of $\cN_{\eta,1}$ as the sets which are crossed but are not in $\cN_{\eta,2}$. This will be relevant in \cref{app:oneside}. 

\begin{definition}[$\cC_{2}$]
	For a connected component of cuts $\cC \subseteq \cN_\eta$, let $\cC_{2} = \cC \cap \cN_2$. 
\end{definition}

The following definition is quite important throughout our paper as it is used to specify the set of bad events we use to construct our slack vector (see \cref{sec:overview} for a gentle introduction):

 \begin{definition}[$S_L$, $S_R$]\label{def:SLR}
 For $S \in \cC_2$ let $S_L$ be the near minimum cut crossing $S$ on the left which minimizes $|O(S \cap S_L)|$. If there are multiple sets crossing $S$ on the left with the same minimum intersection, choose the smallest one to be $S_L$. Similarly, let $S_R$ be the near min cut crossing $S$ on the right which minimizes $|O(S \cap S_R)|$, and again choose the smallest set to break ties. See \autoref{fig:S-SR-SL}.
 


 \end{definition}


\subsection{Properties of inside atoms}

Before proving some new properties of the polygon representation we recall some basic properties of inside atoms from \cite{BG08}:

\begin{definition}[{\cite[Definition 3]{BG08}}]\label{def:kcycle}
A family of sets $C_1,\dots,C_k\subseteq V$, for some $k \ge 3$, forms a $k$-cycle if
\begin{itemize}
\item $C_i$ crosses both $C_{i-1}$ and $C_{i+1}$ (we treat $C_{k+1}$ as $C_1$ and $C_{0}$ as $C_k$);
\item $C_i\cap C_j=\emptyset$ for $j\neq i-1, i$ or $i+1$; and
\item $\bigcup_{1\leq i\leq k} C_i \neq V$.
\item If $k=3$, we have the additional condition $(C_i \cap C_{i+1}) \not\subseteq C_{i-1}$ for $i \in \{1,2,3\}$.
\end{itemize}
\end{definition}

\begin{lemma}[{\cite[Lemma 22]{BG08}}]\label{lem:no-kcycle}
Any $k$-cycle formed by cuts in a connected component ${\cal C}$ of $\eta$-near min cuts satisfies $k\geq 2/\eta$. (Note if $\eta = 0$ then no $k$-cycle exists.)
\end{lemma}

Therefore, for all $2/5$-near min cuts, there are no cycles with length less than 5.
\begin{lemma}[{\cite[Def 4]{BG08}}]\label{def:inside}
	An atom $a\in \cA(\cC)$ is an inside atom (in the representation defined above) if and only if there is a $k$-cycle $C_1,\dots,C_k\in {\cal C}$, such that $a \cap C_i=\emptyset$ for all $1 \le i \le k$. 
\end{lemma}
See \autoref{fig:polygonrepresentation} for an example of an inside atom and a cycle.

\begin{fact}\label{lem:adjoutsidekcycle}
Let $C_1,\dots,C_k$ be a $k$-cycle for a connected component $\C$ with polygon representation $P$ and $k \ge 5$. For any adjacent pair of outside atoms $a,b\in O(\cA(\cC))$, there is a $1\leq j\le k$ such that $a,b\in C_j$.	
\end{fact}
\begin{proof}
Since $a$ is an outside atom, there is a cut $C_i$ for some $1\leq i\leq k$ such that $a\in C_i$ (otherwise $a$ would be an inside atom). If $b\in C_i$ we are done. Otherwise, $a$ is a rightmost or leftmost outside atom in $C_i$. But then, since $C_i$ is crossed by $C_{i-1}$ and $C_{i+1}$ and $C_{i-1} \cap C_{i+1} = \emptyset$, it follows from \cref{obs:Ocross} (and the fact that both $C_{i-1},C_{i+1}$ contain at least two outside atoms) that either $a,b \in C_{i-1}$ or $a,b \in C_{i+1}$.
\end{proof}

\subsection{New properties of polygon representations}
\label{sec:polygonsnewproperties}

The following lemmas build on \cite{Ben95,BG08}. 
\cref{prop20} is a key property of polygons which \cref{lem:extendedprop20} extends: 

\begin{proposition}[{\cite[Proposition 20]{BG08}}]\label{prop20}
For any connected component ${\cal C}$ of $\eta$-near min cuts with $\eta \leq 2/5$ with polygon representation $P$, and any $S_1, S_2 \in \mathcal{C}$ with $S_1 \not= S_2$ we have $O_P(S_1) \neq O_P(S_2)$.
\end{proposition}

\begin{lemma}\label{lem:extendedprop20}
Let $P$ be the polygon representation  for a connected component $|{\cal C}|>1$ of $\eta$-NMCs of a (fractionally) 2-edge connected graph $G$ with atom set $\cA(\cC)$. If $A,B\subsetneq \cA(\cC)$ are two  $2/5$-NMCs with $O_P(A) = O_P(B)\ne \emptyset$ and there is an atom $r\in \cA(\C)$ such that $r\notin A,B$, then $A=B$ \footnote{As indicated earlier, $r$ is called the root of the polygon $P$}.
\end{lemma}
\begin{proof}
Take the graph $G$ and contract each atom of $\cA(\cC)$ to a single node. Call the resulting graph $G'$.
Clearly, $G'$ is still $2$-edge connected  since $G$ is $2$-edge-connected and all cuts in $\cC$ are represented in $G'$.
Now, consider the set of non-singleton $2/5$-near-min-cuts of $G'$. This set has a unique connected component of crossing cuts because any new (non-singleton) cut $S\notin \cC$ is crossed by a cut in $\cC$. (Suppose not: then, no cut on the atoms of $S$ crosses a cut on the atoms of $\overline{S}$, which contradicts that $\cC$ forms a connected component.) Call this component of cuts $\cC'$ and the corresponding polygon $P'$. It follows that $\cA(\cC)=\cA(\cC')$ (more precisely, a set  $S \subseteq V$ is an atom in $\cA(\cC)$ if and only if it is an atom in $\cA(\cC')$): no two atoms  from $P$ can be merged in $P'$ because we have not deleted any cuts, and no atoms in $P$ can be split in $P'$ because we have contracted them.
 While some outside atoms of $P$ may become inside atoms in $P'$, it follows by \cref{def:inside} that any inside atom of $P$ remains an inside atom in $P'$ (as any $k$-cycle of $\cC$ is also a $k$-cycle of $\cC'$). Therefore, 
 $$O_{P'}(A) = O_{P'}(B).$$
 
 Therefore, if $A,B \in \cC'$, by \cref{prop20}, $A=B$. So it remains to show that $A,B \in \cC'$. First, assume $2 \le |A| \le |\cA(\cC')|-2$ and $2 \le |B| \le |\cA(\cC')|-2$. Then, by \cref{lem:mincutBshowsup}, $A,B\in\cC'$ and we are done.
   

Now, we claim that $2 \le |A| \le |\cA(\cC')|-2$ and $2 \le |B| \le |\cA(\cC')|-2$, which by the above would complete the proof. For contradiction, assume $A\neq B$ and $|A|=1$ or $|A|=|\cA(\cC')|-1$. First assume $|A|=1$. Since $B\neq A$, $O_{P}(A)=O_P(B)\neq\emptyset$, and every polygon has at least three outside atoms, $2 \le |B| \le \cA(\cC) - 2$. By \cref{lem:mincutBshowsup}, $B\in \cC'$. Yet this implies that $B$ has at least two outside atoms in $P'$ (and therefore in $P$), which contradicts $O_{P'}(A) = O_{P'}(B)$. Otherwise, $|A| = |\cA(\cC')|-1$. Using that $A \not= B$ and $r \not\in A,B$, it follows that $B$ has at most $|\cA(\cC')|-2$ atoms. Again using that every polygon has at least three outside atoms, this implies $|B| \ge 2$. So similarly to above, we have $B \in \cC'$. Therefore, $|O_{P'}(B)| \le |O(P')|-2$ which contradicts that $|A|=|\cA(\cC')|-1$ using $O_{P'}(A)=O_{P'}(B)$.
\end{proof}

We generalize \cref{obs:Ocross} in the next lemma.
\begin{lemma}\label{lem:OABcross}
	Let $P$ be a polygon representation of a connected component ${\cal C}$ of $\eta$ NMCs of a (fractionally) 2-edge-connected graph $G$ for some $\eta \leq 2/5$. For any $2/5$ NMCs $A,B\subseteq \cA(\cC)$ with $O(A),O(B)\neq \emptyset$, if $A,B$ cross, then $O(A), O(B)$ cross and $O(A \cup B) \not= O(\cA(\cC))$.
\end{lemma}
\begin{proof}
Similar to the previous lemma, consider the graph $G'$ arising from contracting all atoms of $\cA(\cC)$ and let $\cC'$ be the (unique) connected component of non-singleton $2/5$-near-min-cuts of $G'$ with corresponding polygon $P'$. 
As before, $\cA(\cC')=\cA(\cC)$ and $O(\cA(\cC')) \subseteq O(\cA(\cC))$.

Notice that since $A,B$ cross (in $P$), each of them contains at least two atoms of $P$.
Therefore, since $A,B$ are $2/5$ NMCs and $A,B$ are not singletons, we must have $A,B\in \cC'$.  
Since $A,B$ cross in $P$ and $\cA(\cC')=\cA(\cC)$, they also cross in $P'$. By \cref{obs:Ocross}, it follows that $O_{P'}(A)$ and $O_{P'}(B)$ cross. Recall that outside atoms of $\cA(\cC)$ may become inside atoms of $\cA(\cC')$, but inside atoms of $\cA(\cC)$ remain inside atoms in $\cA(\cC')$. So, $O_P(A)$ and $O_P(B)$ cross as well (in particular it also follows that $O(A \cup B) \not= O(\cA(\cC))$).
\end{proof}

\subsubsection{Almost diagonal cuts and the chain lemma}

In some cases, we will need to refer to cuts which are generated by intersections of diagonals in $\cC$. Such cuts are a subset of the following class:

\begin{definition}[Almost Diagonal Cuts]
Let $\cC$ be a connected component of cuts in $\cN_\eta$. We say a set of atoms $S \subseteq \cA(\cC)\smallsetminus \{r\}$ is an {\em almost diagonal cut} if:
\begin{enumerate}
\item $S$ is a $2\eta$-near min cut,
\item $\emptyset \neq O(S) \subsetneq O(\cA(\cC))$,
\item $O(S)$ forms a contiguous interval in the polygon.
\end{enumerate}
Notice that by definition any cut in $\cC$ is  an almost diagonal cut. 
\end{definition}

When we reference almost diagonal cuts in the rest of the paper, we will always assume $\eta \le 1/5$. Notice that given any two $A,B \in \cN_\eta$, $A \cap B, A \smallsetminus B, B \smallsetminus A$ are almost diagonal cuts.

The following fact implies that one can naturally define left/right crossing analogous to \cref{def:SLR} for almost diagonal cuts. The following is a consequence of \cref{lem:OABcross}:
\begin{fact}\label{fact:contiguous}
Let $A,B$ be two crossing almost diagonal cuts. Then, $O(A)$ and $O(B)$ cross. In addition, neither $O(A \cup B)$ nor $O(A \cap B)$ contain all outside atoms and each of $O(A \cup B)$ and $O(A \cap B)$ form a contiguous interval of  outside atoms.
\end{fact} 

\begin{lemma}\label{lem:nointersectionLR}
For an almost diagonal cut $S$ with leftmost atom $a_l$ and rightmost atom $a_r$, let $L \in \cC$ cross $S$ on the left and $R \in \cC$ cross $S$ on the right. Then, if $\eta \le 1/5$, $(L \smallsetminus S) \cap (R \smallsetminus S) = \emptyset$. 	
\end{lemma}
\begin{proof}
	Suppose $(L \smallsetminus S) \cap (R \smallsetminus S) \not= \emptyset$. We will show that $L,R,S$ form a 3-cycle (\cref{def:kcycle}) which cannot exist. To show this (using that the root atom $r \not\in S \cup L \cup R$), it is enough to prove that all pairs cross and none of the three sets is a superset of the intersection of the two others. First, by assumption $L$ and $R$ cross $S$. $L$ and $R$ cross because $a_r \in R \smallsetminus L, a_l \in L \smallsetminus R$, so neither is a subset of the other, and because we assumed $(L \smallsetminus S) \cap (R \smallsetminus S) \not= \emptyset$ their intersection is nonempty. 
	
	In addition, we have $S \cap L \not\subseteq R$ because $a_l \in S \cap L$ but not in $R$. Similarly $S \cap R \not\subseteq L$. Finally, $L \cap R \not\subseteq S$ by the assumption $(L \smallsetminus S) \cap (R \smallsetminus S) \not= \emptyset$. Therefore $S,L,R$ form a 3-cycle which is a contradiction as $S,L,R$ are all $2/5$-near min cuts.
\end{proof}

A fundamental property of the cactus representation is that  the set of min cuts $A_1,\dots,A_k$ crossing a min cut $S$ form two laminar families inside $S$. In other words, perhaps after renaming we may assume $A_1 \cap S \subseteq A_2 \cap S \dots \subseteq A_j \cap S$ and $A_{j+1} \cap S \subseteq \dots \subseteq A_k \cap S$.

It is not immediately obvious that such a property extends to polygons because of the existence of inside atoms. Nonetheless, the following lemma demonstrates that this property is also true of near min cuts provided that $\eta$ is small enough. It is an immediate corollary of \cref{lem:crosschain}. 
\begin{lemma}[Chain Lemma]\label{lem:crosschain} 
Let  $S$ be an almost diagonal cut where $O(S)$ contains all outside atoms from $a$ to $b$. In addition, let $A_1,\dots,A_k\in\cC$ be the collection of $\eta$-near-min cuts crossing $S$ on the left. Then there is a permutation, $\pi:[k]\to [k]$ such that 
$$S\cap S_L=S\cap A_{\pi(1)} \subseteq \dots  \subseteq S\cap A_{\pi(k)}$$ 
i.e., their intersections with $S$ form a chain. The same statements also hold for cuts crossing $S$ on the right.	
\end{lemma}

\begin{lemma}\label{lem:crosschain}
Let $S$ be a near diagonal cut with leftmost outside atom $a$ and rightmost outside atom $b$. Furthermore, let $A=[a_1,a_2],B=[b_1,b_2]\in \cC$ be cuts which cross $S$ on the right. If $\eta \leq 1/10$ and in the interval of the outside atoms of $O(S)$, $a_1$ is to the left of $b_1$, then $B \cap S \subseteq A \cap S$. 
In the special case that $a_1=b_1$, we have $S\cap A=S\cap B$.  
\end{lemma}
\begin{proof}
First assume $b_1\neq a_1$. Since $b_1$ is to the right of $a_1$, and $A,B$ cross $S$ on the right, $O(A\cap S\cap B)=O(S\cap B)\neq\emptyset$ (as both sets have all outside atoms between $b_1$ and $b$). Note that $B$ crosses $A \cap S$. This is because $B$ has an atom outside of $S$ (as it crosses $S$ itself), $a_1 \in A \cap S \smallsetminus B$, and $b_1 \in A \cap S \cap B$. So by \cref{lem:cutdecrement} $A \cap S \cap B$ is a $4\eta$ near min cut. In addition, $S \cap B$ is a $3\eta$ near min cut since $B$ crosses $S$. Therefore, since $4\eta\leq 2/5$, by \cref{lem:extendedprop20} $A\cap S\cap B=S\cap B$.

In the special case that $a_1=b_1$, $O(A\cap S)=O(B\cap S)\neq\emptyset$ has all outside atoms between $a_1=b_1$ and $b$.
Since $A\cap S,B\cap S$ are $3\eta$ near min-cuts, by \cref{lem:extendedprop20}, we must have $A\cap S=B\cap S$ as desired.
\end{proof}

\subsection{Another structural property of inside atoms}

The following lemma is not explicitly used in the proof of the main theorem, but the statement may be useful to guide the reader's intuition (or in other settings where near min cuts arise). For example, it implies that the yellow region is empty in \cref{fig:ErL=ErS}.

\begin{lemma}\label{thm:halfplanes}
Let $ \eta \le  2/5$ and let $P$ be the polygon representation for a connected component $|{\cal C}|>1$ of $\eta$-NMCs of a (fractionally) 2-edge-connected graph. Suppose $H$ is the intersection of half-planes\footnote{Technically, half-polygons.} $H_0,\dots,H_{\ell-1}$ corresponding to diagonals $D_0,\dots,D_{\ell-1}$ of $P$ that has a positive area. If $H$ does not contain any side of $P$ (equivalently, it does not any contain any outside atom) and $\ell<1/\eta$ then $H$ does not have any inside atoms.
\end{lemma}
\begin{proof}
Define $C_0,\dots,C_{\ell-1} \subseteq V$ such that $C_i = \cA(\C) \cap \overline{H_i}$, i.e. all atoms which are not in the halfplane $H_i$. Without loss of generality assume there are no two sets $C_i,C_j$ such that $C_i \subseteq C_j$\footnote{This is because if $C_i \subseteq C_j$ then $H_j \cap P \subseteq H_i \cap P$ which means $H_i$ is redundant in defining $H$.}. 

First observe that $H$ is a 2-dimensional polytope and therefore without loss of generality we can assume each $D_i$ defines a side of $H$ and $D_0,\dots,D_{\ell -1}$ are ordered such that the vertices of the polytope $v_0,\dots,v_{\ell-1}$ are arranged cyclically counterclockwise where $v_i$ is the intersection of $D_i$ and $D_{i+1}$ (where for the rest of the proof we take all indices mod $\ell$). We will call a vertex \textit{external} if it is a polygon point of $P$ and \textit{internal} otherwise. 

We prove the claim by induction over $r$ that if $H$ has $r$ external vertices and $\ell+r < 2/\eta$, then $H$ is empty. Note that if $\ell < 1/\eta$, since $r\leq \ell$, we have $\ell+r<2/\eta$ which proves the theorem.

First assume $r=0$. Then, $C_{i}$ crosses $C_{i+1}$ for all $i$. By way of contradiction suppose $H$ contains an (inside) atom $a$. Without loss of generality, assume that $H_0,\dots,H_{k-1} \subseteq H_0,\dots,H_{\ell-1}$ (perhaps after renaming) is the minimal set of half-planes that contain $H$ and have no external vertices. We claim that there are no indices $i,j$ and $j \not= i-1,i,$ or $i+1$ such that $C_i \cap C_j \not= \emptyset$. Since $C_i \not\subseteq C_j$,  $C_j \not\subseteq C_i$, $a \not\in C_i,C_j$, it follows that $C_i,C_j$ cross and therefore the diagonals $D_i,D_j$ intersect in the interior of $P$. Now $D_0,\dots,D_i,D_j,\dots,D_{k-1}$ contains $H$ and has no external vertices and thus contradicts the minimality of $H_0,\dots,H_{k-1}$.  

Therefore by minimality, $C_i \cap C_j = \emptyset$ if $j \not= i-1,i$ or $i+1$. So, $C_0,\dots,C_{k-1}$ is a $k$-cycle for $a$: since there are no external vertices, $C_i$ crosses $C_{i+1}$ for all $i$, and $\bigcup_{i=0}^{\ell-1} C_i \not= V$ since $a \not\in C_i$ for all $i$. By \cref{lem:no-kcycle}, $k \ge  2/\eta$, which is a contradiction, because $k \le \ell < 1/\eta$. 

Now, suppose the claim is true when the number of external vertices is at most $r$; we will prove it holds when the number is $r+1$.  

Again, by way of contradiction suppose there is an inside atom $a$ in $H$. Then there is a $k$-cycle for $a$, $L_1,\dots,L_k \in \cC$. Now pick an arbitrary external vertex $v_i=p_j$ of $H$ for some $j$. Let $a_{j-1},a_{j} \in O(\cA(\cC))$ be the adjacent outside atoms immediately to the clockwise and counterclockwise of $p_j$. Therefore, by \cref{lem:adjoutsidekcycle}, there exists some $j$ such that $a_{j-1},a_j \in L_j$. Let $H_{\ell}$ be the halfspace corresponding to the side of $L_j$ containing $a$. Now consider the set of halfspaces $H_0,\dots,H_{\ell-1},H_{\ell}$. Note that $H'=\bigcap_{i=0}^{\ell} H_i \subsetneq H$ because $v_i=p_j \not\in H_\ell$. Therefore, $H'$ has at least one fewer external vertex and no sides of the polygon. Since $\ell+1+(r-1) \le 2/\eta$, by the induction hypothesis, $H'$ has no inside atoms which is a contradiction with the existence of $a$. 
\end{proof}

%

\section{Proof of the main theorem}\label{sec:twoside}

\subsection{Notation and a preliminary lemma}

We will use the same set of definitions for bad events and increase sets that we did in \cref{sec:overview} for polygons without inside atoms. For the benefit of the reader we repeat them here. Recalling the definitions from \cref{sec:polynotation}, partition each set $\delta(S)$ into three sets $E^\leftarrow(S), E^\rightarrow(S)$ and $E^\circ(S)$ such that
 \begin{align*}
 	E^\leftarrow(S) &= E(S \cap S_L, S_L \smallsetminus S) \\
 	 E^\rightarrow(S) &= E(S \cap S_R, S_R \smallsetminus S) \\
 	 E^\circ(S) &= \delta(S) \smallsetminus (E^\leftarrow(S) \cup E^\rightarrow(S))
 \end{align*}
In addition we define the left and right bad events:
\begin{align}
	B^\rightarrow (p) &= \mathbb{1}\{|E^{\rightarrow} (L(p)) \cap T|\ne 1\text{ or }|E^\circ (L(p)) \cap T|\ne 0\}\notag\\
B^\leftarrow (p) &= \mathbb{1}\{|E^{\leftarrow} (R(p)) \cap T|\ne 1\text{ or }|E^\circ (R(p)) \cap T|\ne 0\}. \label{defn:badevents2}
\end{align}
If $L(p)$ does not exist, simply assume the left bad event never occurs, and similarly if $R(p)$ does not exist assume the right bad event never occurs.

Define $L(p)^{\cap R} := L(p) \cap L(p)_R$, and let $L^*(p) \in \cC$ be the cut crossing $L(p)^{\cap R}$ on the left that {\em maximizes}  $|O(L^*(p) \cap L(p)^{\cap R})|$ (and similarly $R^*(p)$ to maximize the intersection with $O(R(p)^{\cap L})$ on the right). If $L^*(p)$ does not exist, i.e. no cut crosses $L(p)^{\cap R}$ on the left, set $L^*(p) = \emptyset$, and similarly for $R^*(p)$. We let:

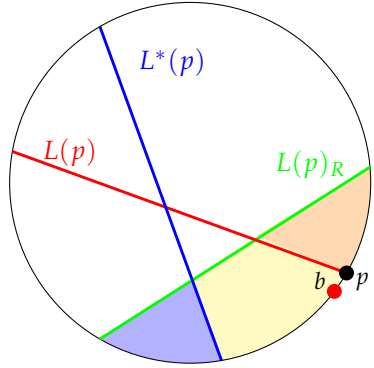
\begin{figure}[htb]\centering
	\begin{tikzpicture}[scale=0.8]
		\fill [color=blue!30] (240:3) -- (0.03,-1.61) -- (280:3)  arc (280:240:3); 
		\fill [color=yellow!30] (280:3) -- (0.03,-1.61) -- (1.1,-0.88) -- (330:3) arc (330:280:3);
		\fill [color=orange!30] (330:3) -- (1.1,-0.88) -- (5:3) arc (365:330:3);
		\draw (0,0) circle (3);
		\draw [color=green,line width=1.1pt] (240:3) -- (5:3);
		\node [color=green] at (2,0.3) () {\small $L(p)_R$};
		\draw [color=red, line width=1.1pt] (170:3) --   (330:3);
		\node [color=red] at (-2,0.5) () {\small $L(p)$};
		\draw [color=blue,line width=1.1pt] (120:3) -- (280:3);
		\node [color=blue] at (-0.3,2) {\small $L^*(p)$};
		\node [color=black,circle,fill=black,inner sep=2] at (330:3) () {};
		\node at (330:3.3) () {\footnotesize$p$};
		\node [circle,fill=red,inner sep=2] at (323:3) {};
		\node at (323:2.7) () {\footnotesize $b$};
	\end{tikzpicture}
	\caption{\small Recap of some basic definitions: $L(p)$ is the cut crossed on both sides with rightmost polygon point $p$ (and contains all atoms below the red diagonal. $L(p)_R$ is the cut crossing $L(p)$ on the right that minimizes the number of outside atoms in $L(p)^{\cap R} = L(p) \cap L(p)_R$, i.e., in yellow + blue.	Note that the cut $L(p)^{\cap R}$ contains all atoms in the yellow and blue regions (which may include inside atoms).  Since this region is the set difference of two $\eta$ near min cuts, it is a $2\eta$ near min cut. $L^*(p)$ is the cut crossing $L(p)^{\cap R}$ on the left that maximizes the number of outside atoms in the intersection, i.e., maximizes the number of outside atoms in the blue region. $E^\rightarrow (L(p))$ are the edges between atoms in the yellow region and atoms in the orange region. There is one edge in the tree that is in $E^\rightarrow (L(p))$ with probability $1-O(\eta)$ and when this event does not occur, $B^\rightarrow(p)$ occurs. ($B^\rightarrow(p)$ also occurs if $|E^\circ (L(p)) \cap T|\ne 0$.)}
\end{figure}

\begin{align}
	E(B^\rightarrow (p)) &:= E( L(p)^{\cap R} \smallsetminus L^*(p), L(p)_R \smallsetminus  L(p)^{\cap R})  \notag\\
	E(B^\leftarrow (p)) &:= E( R(p)^{\cap L} \smallsetminus R^*(p), R(p)_L \smallsetminus  R(p)^{\cap L})\label{defn:slackincrease2}.
\end{align}

The following important lemma is the generalization of \cref{claim:1} from \cref{sec:overview} to the case in which there may be inside atoms. It uses that by \cref{lem:extendedprop20} many regions of the polygon do not contain inside atoms.


\begin{lemma}\label{lem:samerightEPsameedgeset}
	Let $A, B \in \cC_{2}$ such that $A=(a_1,a_r)$ and $B=(a_2,a_r)$ share a rightmost polygon point $p_r$. Then, $A_R = B_R$ and $E^\rightarrow(A) = E^\rightarrow(B)$.
\end{lemma}
\begin{proof}
	WLOG assume $A \subseteq B$. First, we will prove that if a cut $R$ crosses $A$ on the right, it also crosses $B$ on the right. So, let $R$ be a set crossing $A$ on the right. Then, since $R$ contains $a_{r}$, $R \cap B \not= \emptyset$. Furthermore, $R$ contains atom $a_{r+1}$, so $R \not\subseteq B$. Finally, $B$ contains $a_1$ since $A \subseteq B$ yet $R$ does not. Therefore, $R$ crosses $B$ on the right. 
	
	Therefore, by \cref{def:SLR} $A_R = B_R$ since the set of cuts crossing $B$ on the right is a superset of cuts crossing $A$ on the right and any cut which crosses $B$ but not $A$ contains all atoms of $A$, so would have a larger intersection with $O(B)$. (If two sets have the same intersection with $O(A)$ we use the same tie-breaking rule for both $A_R$ and $B_R$.)   
	
	Now we prove that $E^\rightarrow(A) = E^\rightarrow(B)$. Let $R=A_R=B_R$. It suffices to show that $R \cap A = R \cap B$ and $R \smallsetminus A = R \smallsetminus B$ because any edge $e \in E^\rightarrow(A)$ has one endpoint in $R \cap A$ and one in $R \smallsetminus A$. To obtain $R \cap A = R \cap B$ notice that $O(R \cap A) = O(R \cap B) \not= \emptyset$ and by \cref{lem:cutdecrement} $R \cap A, R \cap B$ are $2\eta$ near minimum cuts (since $R$ crosses both $A$ and $B$), so by \cref{lem:extendedprop20}, $R \cap A = R \cap B$. Similarly to obtain $R \smallsetminus A = R \smallsetminus B$ notice that $O(R \smallsetminus A) = O(R \smallsetminus B) \not= \emptyset$ and $R \smallsetminus A, R \smallsetminus B$ are $2\eta$ near min cuts so by  \cref{lem:extendedprop20} we have $R \smallsetminus A = R \smallsetminus B$. 
	\end{proof}

\subsection{Main theorem}

The following is the main technical result of the paper:
\begin{restatable}[Main theorem]{theorem}{maintheorem}\label{thm:cutsbothsideswithinside}
Let $x^0$ be a feasible LP solution of \eqref{eq:tsplp} with support $E_0=E\cup \{e_0\}$ and let $x$ be $x^0$ restricted to $E$.  For {\em any} distribution $\mu$ of spanning trees with marginals $x$,  $0<\eta \leq 1/10$ and $\alpha > 0$, there is a random vector $s^*:E \to \R_{\geq 0}$ (the randomness in $s^*$ depends exclusively on $T\sim\mu$) such that
\begin{itemize}
\item For any $\eta$-near minimum cut $S$ which is crossed on both sides, if $\delta(S)_T$ is odd then $s^*(\delta(S)) \geq \alpha(1-\eta)$;
\item For any $e \in E$, $\E{s^*_{e}} \leq 18\alpha\eta x_e.$
\end{itemize}
\end{restatable}
%

Our slack vector for the above theorem will be exactly as in \cref{sec:overview}. In particular, for  every bad event $B$ which occurs among those defined in \cref{defn:badevents2}, we will set $s^*(e) = \alpha x_e$ for all $e \in E(B)$, where $E(B)$ is defined as in \cref{defn:slackincrease2}. In order to extend the argument from \cref{sec:overview} to prove this theorem in the case in which the polygon contains inside atoms, we prove the following two lemmas from which the theorem follows easily:

\begin{restatable}[All cuts are satisfied]{lemma}{allcutssatisfied}\label{lem:xed-both-sides-one-increases} Let $S=(p_l,p_r)$ be a cut which is crossed on both sides. Then, if $\delta(S)_T \not=2$, at least one of $B^\leftarrow(p_l),B^\rightarrow(p_r)$ occurs.
\end{restatable}

\begin{restatable}[Every edge is mapped to a constant number of bad events]{lemma}{constantnumberevents}\label{lem:constant-num-events}
Let $p,q$ be two polygon points such that $e=\{a,b\}$ and $a \in L(p) \cap L(q)$. 
Then, $e \not\in E(B^\rightarrow(p)) \cap E(B^\rightarrow(q))$.
\end{restatable}

Before proving these statements, we will show how they imply our main theorem. First we gives proofs for \cref{claim:2} and \cref{claim:xvaluenoinside} (as the formal proofs were omitted in the overview):

\begin{lemma}\label{lem:OPT-edge-increase-prob-xed-both-sides}
For any polygon point $p$, $\P{B^\rightarrow(p)}\leq 4.5\eta$ and $\P{B^\leftarrow(p)}\leq 4.5\eta$.
\end{lemma}
\begin{proof}
We will prove this for $B^\rightarrow(p)$, $B^\leftarrow(p)$ follows similarly. To simplify notation we abbreviate $L(p)$ to $L$.
Since $L$ is crossed on both sides, $L_L, L_R$ are well defined. Since by \cref{lem:cutdecrement} $L_R\cap L, L_R\smallsetminus L$ are $2\eta$-near min cuts and $L_R$ is an $\eta$-near mincut with respect to $x$, by \cref{lem:treeoneedge}, $\P{E^\rightarrow(L)_T=1}\geq 1-2.5\eta$. 

On the other hand, since $L, L_L, L_R$ are $\eta$-near min cuts, by \cref{lem:nmcuts_largeedges}, $x(E^\rightarrow(L)), x(E^\leftarrow(L)) \ge 1-\eta/2$. Therefore 
$$x(E^\circ(L)) \le 2+\eta - x(E^\leftarrow(L)) - x(E^\rightarrow(L)) \le 2\eta.$$ 
It follows that $\P{E^\circ(L)_T = 0}\geq  1-2\eta$. 
	Finally,	by the union bound, all events occur simultaneously with probability at least $1-4.5\eta$, which gives the lemma.
	\end{proof}
	
\begin{lemma}\label{lem:massinI}
	For any polygon point $p$, $x(E(B^\leftarrow(p))),x(E(B^\rightarrow(p)))\geq 1-\eta$
\end{lemma}
\begin{proof}
First, observe that $L^*(p)$ crosses $L(p)_R$.
Notice (where we use $\uplus$ to denote disjoint union):
$$ L(p)^{\cap R}\smallsetminus  L^*(p) \biguplus L(p)_R\smallsetminus L(p) = L(p)_R \smallsetminus L^*(p).$$
Therefore, by \autoref{lem:cutdecrement} $L(p)^{\cap R}\smallsetminus  L^*(p) \biguplus L(p)_R\smallsetminus L(p)$ is a $2\eta$ near mincut. 
So, by \autoref{lem:sub-NMC-shared}, $x(E(B^\leftarrow(p))\geq 1-\eta$. 

The proof for $x(E(B^\rightarrow(p)))$ is similar. 
 \end{proof}

\begin{proof}[Proof of \cref{thm:cutsbothsideswithinside}]
Our slack vector is defined as follows. Initialize $s^*(e) = 0$ for all edges $e$. Then for each polygon $P$, for each polygon point $p \in P$, whenever $B^\leftarrow(p)$  occurs, let $s^*_{e}=\alpha x_e $ for each $e \in E(B^\leftarrow(p)) $. Whenever $B^\rightarrow(p)$ occurs, let $s^*_{e}=\alpha x_e $ for each $e \in E(B^\rightarrow(p))$. 
	
Now we show the first condition of the theorem. Let $S=[p,q] \in \cC_2$ and suppose that $\delta(S)_T$ is odd. It appears in some polygon $P$. Then by \cref{lem:xed-both-sides-one-increases}, either $B^\leftarrow(p)$ or $B^\rightarrow(q)$ has occurred. Assume the former, the other case is similar. In this event, we set $s^*(e) = \alpha x_e$ for all $e \in E(B^\leftarrow(p))$. However, using \cref{lem:samerightEPsameedgeset}, we have
$$E(B^\leftarrow(p)) \subseteq E^\leftarrow(p) = E^\leftarrow(S)$$
Therefore, by \cref{lem:massinI} $s^*(S) \ge \alpha (1-\eta)$ as desired. 
	
Now we verify the second condition of the theorem. First note that by \cref{fact:edges-internal-to-one-poly}, for any edge $e$, there is at most one polygon $P$ such that $e$ does not have the root as one of its endpoints. Therefore, there is at most one polygon $P$ for which $s^*_e$ may be increased since if $e$ is adjacent to the root, $e \not\in E^\rightarrow(S),E^\leftarrow(S)$ for any set $S$ in its connected component. 

Now let $e=\{a,b\}$ for some polygon $P$ such that $a,b \not= r$. We show that there are at most two polygon points $p$ for which $e \in E(B^\rightarrow(p))$. Suppose otherwise and there are at least three such polygon points $p$. Since $E(B^\rightarrow(p)) \subseteq \delta(L(p))$ for each such point $p$, we have $e \in \delta(L(p))$, which implies that there are two polygon points $p,q$ such that one of $a$ or $b$, WLOG $a$, is in both $L(p)$ and $L(q)$ and $e \in E(B^\rightarrow(p)) \cap E(B^\rightarrow(q))$. However this contradicts \cref{lem:constant-num-events} since $e$ is in at most one such set. One can similarly show that there are at most two polygon points $p$ such that $e \in E(B^\leftarrow(p))$. Therefore, any edge $e$ is in $E(B)$ for at most four bad events $B$. 

By \cref{lem:OPT-edge-increase-prob-xed-both-sides}, each bad event occurs with probability at most $4.5\eta$. Therefore, by the union bound:
$$\E{s^*(e)} \le 4 \cdot 4.5\eta \alpha x_e = 18 \alpha \eta x_e,$$
which gives the second condition and completes the proof. 
\end{proof}

\subsection{All cuts are satisfied}

In this section we first prove the following from which \cref{lem:xed-both-sides-one-increases} will easily follow:

\begin{restatable}{lemma}{circsubset}\label{lem:circsubset} For $S=(p_l,p_r) \in \cC_2$,
$E^\circ(S) \subseteq E^\circ(L(p_r)) \cup E^\circ(R(p_l))$.	
\end{restatable}

Note that if $S=(p_l,p_r) \in \cC_2$ then $L(p_r)$ and $R(p_l)$ exist, because $S$ is a candidate for both.

Before giving the proof of this we show how it implies \cref{lem:xed-both-sides-one-increases}. 

\allcutssatisfied*
\begin{proof}
We prove by contradiction. Suppose none of $B^\leftarrow(p_l),B^\rightarrow(p_r)$ occur; we will show that this implies $\delta(S)_T = 2$. 

Let $R=R(p_l)$. By \cref{lem:samerightEPsameedgeset} we have $S_L = R_L$ and $E^\leftarrow(R) = E^\leftarrow(S).$
Similarly for $L=L(p_r)$ we have $E^\rightarrow(L)=E^\rightarrow(S)$.

Now, since $B^\leftarrow(p_l)$ has not occurred,
$$1 = E^\leftarrow(R)_T = E^\leftarrow(S)_T \text{ and } E^\circ(R)_T = 0$$
and since $B^\rightarrow(p_r)$ has not occurred,
$$1 = E^\rightarrow(L)_T =  E^\rightarrow(S)_T \text{ and } E^\circ(L)_T = 0$$
So, to get $\delta(S)_T=2$, it remains to show that  $T\cap E^\circ(S) = \emptyset$. By \cref{lem:circsubset}, we have $E^\circ(S) \subseteq E^\circ(L) \cup E^\circ(R)$, which gives the claim.
\end{proof}

To prove \cref{lem:circsubset} we first need the following:

\begin{corollary}\label{cor:leftrightemptyextended}
	For all sets $S \in \cC_2$, we have $E^\leftarrow(S) \cap E^\rightarrow(S) = \emptyset$. Similarly, for all sets $A,B \in \cC_2$ such that $B$ crosses $A$ on the right, $E^\leftarrow(A) \cap E^\rightarrow(B) = \emptyset$. 
\end{corollary}
\begin{proof}
To see the first claim, suppose $S_L$ crosses $S$ on the left and $S_R$ crosses $S$ on the right, by \cref{lem:nointersectionLR} $(S_L \smallsetminus S) \cap (S_R \smallsetminus S) = \emptyset$. Since every edge in $E^\leftarrow(S)$ has an endpoint in $S_L \smallsetminus S$ and every edge in $E^\rightarrow(S)$ has an endpoint in $S_R \smallsetminus S$, this proves the claim. 	

A similar argument using \cref{lem:nointersectionLR} proves the second claim.
\end{proof}

Now we can prove the main lemma: 

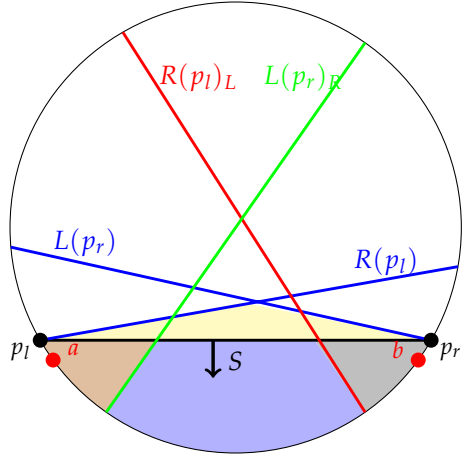
\begin{figure}[htb]\centering
	\begin{tikzpicture}
		\fill [color=brown!50] (210:3) -- (-1.05,-1.5) -- (235:3) arc (235:210:3);
		\fill [color=gray!50] (305:3) -- (1.05,-1.5) -- (330:3) arc (330:305:3);
		\fill [color=yellow!30] (210:3) -- (0.1,-1) -- (330:3);
		\fill [color=blue!30] (235:3) -- (-1.05,-1.5) -- (1.05,-1.5) -- (305:3) arc (305:235:3);
		\draw (0,0) circle (3);
		\draw [color=black,line width=1.1pt] (210:3) -- node [below] {\small$S$} (330:3);
		\draw [color=black,line width=1.3pt,->] (-0.3,-1.5) -- +(0,-0.5);
		\draw [color=blue,line width=1.1pt] (185:3) -- (330:3)
		(210:3) -- (350:3);
		\draw [color=red,line width=1.1pt] (120:3) -- (305:3);
		\draw [color=green,line width=1.1pt] (55:3) -- (235:3);
		\node [color=blue] at (-2,-0.2) () {\small$L(p_r)$};
		\node [color=blue] at (2,-0.45) () {\small$R(p_l)$};
		\node [color=red] at (-0.5,2) () {\small$R(p_l)_L$};
		\node [color=green] at (0.9,2) () {\small$L(p_r)_R$};
		\node [color=red,fill=red,circle,inner sep=2] at (216:3) () {};
		\node [color=red,fill=red,circle,inner sep=2] at (324:3) () {};
		\node at (210:3.3) () {\footnotesize $p_l$};
		\node [color=black,circle,fill,inner sep=2] at (210:3) () {};
		\node at (330:3.3) () {\footnotesize $p_r$};
		\node [color=black,circle,fill,inner sep=2] at (330:3) () {};
		\node [color=red] at (217:2.7) () {\footnotesize $a$};
		\node [color=red] at (323:2.7) () {\footnotesize $b$};
	\end{tikzpicture}
	\caption{Here we have a set $S$ which is crossed on both sides. We use in \cref{lem:circsubset} that $L(p_r) \cap R(p_l) = S$; in other words, the yellow region is empty.}
\end{figure}

\circsubset*
\begin{proof}
	For convenience, let $L=L(p_r)$ and $R=R(p_l)$. First suppose we had $L = S$ or $R = S$. In this case, by definition $E^\circ(S) = E^\circ(L)$ or $E^\circ(R)$, and we are done. So assume that $S \subsetneq L, R$. Therefore, $L$ and $R$ cross.  
	
	Now, notice that since $O(L \cap R) = O(S)$, by \cref{lem:extendedprop20}, we have $L \cap R = S$. This implies
	$\delta(S) \subseteq \delta(L) \cup \delta(R)$. To see this, let $e$ be an edge in $\delta(S)$. Then, it has an endpoint in both $L$ and $R$. However, its other endpoint is in $\overline{S} = \overline{L \cap R}$, and therefore cannot be in both $L$ and $R$ which implies it is in $\delta(L)$ or $\delta(R)$.  

	Now, by way of contradiction, suppose there exists an edge $e \in E^\circ(S)$ such that $e \not\in E^\circ(L) \cup E^\circ(R)$. Since $E^\leftarrow(S) = E^\leftarrow(R)$, $E^\rightarrow(S) = E^\rightarrow(L)$, and $e \in \delta(L) \cup \delta(R)$, it must be that $e \in E^\leftarrow(L) \cup E^\rightarrow(R)$. Since $L$ and $R$ cross, by \cref{cor:leftrightemptyextended} $e$ is in exactly one of $E^\leftarrow(L),E^\rightarrow(R)$; assume that $e \in E^\leftarrow(L)$ but not in $E^\rightarrow(R)$, the other case is similar. Therefore $e \not\in \delta(R)$, since it is not in $E^\leftarrow(R), E^\circ(R)$ or $E^\rightarrow(R)$. 
	
	However, $L_L$ crosses $L$ on the left and $R$ crosses $L$ on the right. Therefore, by \cref{lem:nointersectionLR}, we have $(L_L \smallsetminus L) \cap (R \smallsetminus L) = \emptyset$. However $e$ has one endpoint in $L \cap R = S$, one in $R \smallsetminus L$ (since $e \not\in \delta(R), e \in \delta(L)$), and one in  $L_L \smallsetminus L$, which is a contradiction since all three sets are disjoint. 
\end{proof}

\subsection{Every cut is mapped to a constant number of bad events}

In this section we prove \cref{lem:constant-num-events}. 

\constantnumberevents*
\begin{proof}
First note that by \cref{fact:union-not-everything}, $O(L(p) \cup L(q)) \not= O(\cA(\cC))$, so it forms a contiguous interval. WLOG assume that $q$ is the rightmost point in this interval.   

Suppose by way of contradiction that $e \in E(B^\rightarrow(p)) \cap E(B^\rightarrow(q))$	. In the below claim, we will show that $L(p)_R$ crosses $L(q)$ on the right. Now we will show that $L(p)$ crosses $L(q)^{\cap R}$ on the left, which would complete the proof. This is because $L(p)$ is a candidate for $L^*(q)$, and by \cref{lem:crosschain}, $L(p) \cap L(q)^{\cap R} \subseteq L^*(q) \cap L(q)^{\cap R}$, which implies $a \in L^*(q)$ and therefore we could not have $e \in E(B^\rightarrow(q))$. 

It remains to show that $L(p)$ crosses $L(q)^{\cap R}$ on the left. By assumption, $a \in L(p) \cap L(q)^{\cap R}\not= \emptyset$. $L(q)$'s rightmost atom is in $L(q)^{\cap R} \smallsetminus L(p)$. So it remains to show that $L(p) \not\subseteq L(q)^{\cap R}$. By way of contradiction suppose $L(p) \subseteq L(q)^{\cap R} = L(q) \cap L(q)_R$. Since by the following claim, $L(p)_R$ crosses $L(q)$ on the right, $L(p)_R$ is a candidate for $L(q)_R$. We will show that $|O(L(p)_R \cap L(q))| < |O(L(q)_R \cap L(q))|$ which contradicts \cref{def:SLR}. Since $L(p)_R$ and $L(q)_R$ both cross $L(q)$ on the right, to prove the inequality it's enough to show that the leftmost outside atom of $L(q)_R$ is not in $L(p)_R$. However, this is immediate because $L(p)_R$ does not have the leftmost outside atom of $L(p)$, yet $L(p) \subseteq L(q)^{\cap R}$. 
\begin{claim}
$L(p)_R$ crosses $L(q)$ on the right.
\end{claim}
\begin{proof}
First we show that $L(p)_R$ crosses $L(q)$. Note $b \in L(p)_R \smallsetminus L(q)\not=\emptyset$, and $a \in L(p)_R \cap L(q) \not=\emptyset$. So, $L(p)_R$ crosses $L(q)$ unless $L(q) \subseteq L(p)_R$.  For contradiction, assume $L(q) \subseteq L(p)_R$. 

Now we claim that $L(p)$ crosses $L(q)$. By assumption, $a \in L(p) \cap L(q)$. $L(q) \subseteq L(p)_R$ implies $L(p) \not\subseteq L(q)$. Since $q$ is the rightmost point of the interval $O(L(p) \cup L(q))$, the rightmost atom of $L(q)$ is in $L(q) \smallsetminus L(p)$, giving $L(q) \not\subseteq L(p)$. Therefore, $L(q)$ crosses $L(p)$ on the right. 

Therefore, $L(q)$ is a candidate set for $L(p)_R$, but since $L(q) \subseteq L(p)_R$, we must have $L(q) = L(p)_R$. Yet $b \in L(p)_R \smallsetminus L(q)$ which contradicts this.

Now we establish that $L(p)_R$ crosses $L(q)$ \textit{on the right}. For contradiction, suppose it crosses on the left. Then, by \cref{lem:nointersectionLR}, we must have $(L(p)_R \smallsetminus L(q)) \cap (L(q)_R \smallsetminus L(q)) = \emptyset$, which contradicts the fact that $b$ lies in both sets. 
\end{proof} 
\end{proof}

\section{Putting everything together}\label{sec:hierarchy}
 

In this section we use the following theorem to demonstrate \cref{thm:main}. While the proof of \cref{thm:maintechnical} is non-trivial, using \cref{thm:cutsbothsideswithinside} it follows from statements in \cite{KKO21} and does not require any new ideas. For this reason, we sketch the proof in this section, leaving the formal proof to \cref{app:proofbeforetechnical}. 

It turns out not to be useful to prove \cref{thm:hierarchyKKO} directly, but is an immediate corollary of the following theorem (we stated \cref{thm:hierarchyKKO} in the overview to improve readability and highlight the importance of \cref{thm:cutsbothsideswithinside}).



\begin{restatable}[Combination of \cref{thm:hierarchyKKO} and \cref{thm:cutsbothsideswithinside}]{theorem}{maintechnical}\label{thm:maintechnical}
Let $x^0$ be a solution of LP \eqref{eq:tsplp} with support $E_0=E\cup \{e_0\}$, and $x$ be $x^0$ restricted to $E$.
Let $\eta\leq 10^{-12}, \decrease > 0$ and let $\mu$ be  the max-entropy distribution with marginals  $x$. 
Then there are two functions $s: E_0\rightarrow \R$ and $s^*: E \rightarrow \R _{\ge 0}$ (as functions of $T\sim\mu$), such that
\begin{enumerate}[i)]
\item For each edge $e \in E$, $s_e \ge -x_e \decrease$ (with probability 1). 
\item
For each $S \in \cN_{\eta}$, if $\delta(S)_T$ is odd, then
$  s(\delta(S)) + s^*(\delta(S)) \ge  0.$
\item For every edge $e$, $\E{s^*_{e}}\leq 125\eta \decrease x_e$ and $\E{s_e}\leq -\frac{1}{3}x_e \eps_P\decrease $, where $\eps_P$ is defined in \cref{thm:payment-main}.      
\end{enumerate}
\end{restatable}

In \cref{subsec:proofofmain} we show how this theorem implies \cref{thm:main}. Now we sketch the ideas underlying the proof of \cref{thm:maintechnical}. To make this section as accessible as possible, we  oversimplify and ignore the details of how parameters are set.

In the above theorem, the role of $s$ is to generate gain over $3/2$. Roughly speaking, we follow the lead of \cite{KKO21} and divide $\cN_\eta$ into three categories: cuts crossed on both sides, cuts crossed on one side, and the remainder, which form a laminar family $\cH$ defined below. We define an $s^*$ vector to provide significant positive slack on each odd cut that is crossed; in particular, we  start with the vector defined in \cref{thm:cutsbothsideswithinside} and augment it to handle cuts crossed on one side. We will ensure that the expected cost of $s^*$ is negligible\footnote{i.e., $\E{s^*_{e}} \leq 18\alpha\eta x_e$, which will ultimately be $O(\eta \decrease x_e)$. In the end, this increase is dwarfed by a \textit{decrease} in $s_e$ of $\Omega(\decrease x_e )$ since $\eta$ is a minuscule constant.}. Now in $\cH$, there are only a linear number of cuts and they have a simple structure (for example, most edges are only in a constant number of cuts of $\cH$), so it is manageable to design a vector $s$ which generates negative slack in expectation while still satisfying every cut in $\cH$. 



First, we explain how to augment $s^*$ from \cref{thm:cutsbothsideswithinside} to handle cuts crossed on one side. Observe that any polygon associated to a connected component in $\cN_{\eta,\le 1}$ contains no inside atoms. This follows  from the fact that the existence of an inside atom is predicated on the existence of a $k$-cycle, which by its very definition  contains cuts crossed on both sides. Thus, each connected component $\cC$ of $\cN_{\eta,\le 1}$ consists only of outside atoms, where $a_0$ is the root.

A key structure needed for the construction of the slack vector $s$ is a laminar family of cuts $\cH$ that we call a {\em hierarchy}. This hierarchy $\cH$ includes the following set of cuts:
\begin{itemize}
\item The set of cuts in $\cN_{\eta, \le 1}$ that are not crossed by any other cut in  $\cN_{\eta, \le 1}$;
\item The cut consisting of the union of the non-root atoms $\{a_1, \ldots, a_{m-1}\}$ of each 	connected component $\cC$ of $\cN_{\eta, \le 1}$, which (in this section) we call the {\em outer polygon cut} for  $\cC$, and
\item The atoms $a_i$, $1\le i \le m-1$ of each connected component $\cC$ of $\cN_{\eta, \le 1}$.
\end{itemize}

Notice that $\cH$ {\em excludes} some cuts in $\cN_{\eta, \le 1}$, namely all the near min cuts in any polygon $P$ of $\cN_{\eta, \le 1}$ that are {\em not} outer polygon cuts.
It also {\em includes} some cuts that are {\em not} in $\cN_{\eta, \le 1}$. For example, the outer polygon cut itself may not be an $\eta$ near min cut, and there may be atoms in some polygon that are not $\eta$ near min cuts.
However, one of the consequences of the following theorem is that these extra cuts  are $\eps_\eta$ near min cuts where $\eps_{\eta}= 7\eta$:

%
%
\begin{theorem}[Structure of Polygons of $\cN_{\eta, 1}$ (Theorem 4.9 from \cite{KKO21})]\label{thm:approxpoly}
For  $\eps_{\eta}\geq 7\eta$ and any polygon of $\eta$ near min cuts $\cC$ crossed on one side with atoms $a_0...a_{m-1}$ (where $a_0$ is the root) the following holds:
	\begin{itemize}
	\item For all adjacent atoms $a_i,a_{i+1}$ (also including $a_0,a_{m-1}$), we have $x(E(a_i,a_{i+1})) \ge 1-\eps_\eta$. 
	\item All atoms $a_i$ (including the root) have $x(\delta(a_i)) \le 2+\eps_\eta$. 
	\item $x(E(a_0, \{a_2,\dots,a_{m-2}\}))\leq \eps_{\eta}$.
	\end{itemize}
\end{theorem}

\cref{thm:approxpoly} shows that polygons of cuts crossed on one side  nearly look like cycles.
Now, if magically it was the case that  $x(E(a_i, a_{i+1 \pmod{m}})) =1$, and $x(\delta(a_i))=2$
for $1 \le i \le m-1$, then with probability 1 (see \cref{lem:treeoneedge}), we would have $E(a_i, a_{i+1})_T = 1$ and we would be able to claim that:  
\begin{itemize}
\item [(i)] any cut in $\cC$ which does not include either $a_1$ or $a_{m-1}$ (and is therefore not in $\cH$) is even in the tree with probability 1; 
\item[(ii)] The  cuts in $\cC$ that contain $a_1$ but not $a_{m-1}$, i.e., the so-called "leftmost cuts" (also not represented in $\cH$) are even precisely when $E(a_0, a_1)_T$ is odd and 
\item [(iii)] the cuts in $\cC$ that contain $a_{m-1}$  but not $a_1$ i.e., the "rightmost cuts", are even when $E(a_0, a_{m-1})_T$ is odd. 
\end{itemize}
\cref{thm:approxpoly} can be used to show that this approximation is 
 correct up to $O(\eta)$. In other words, we augment $s^*_e$ as needed on each edge between adjacent non-root atoms in each connected component $\cC$, at the cost of increasing $\E{s_e^*}$ by an additional (again negligible) $O(\eta\decrease x_e)$. This allows us to pretend our magical thinking is correct. Thus, all of the $\eta$ near min cuts in the polygon that are {\em not} represented in the hierarchy are satisfied so long as the outer polygon cut 
is {\em happy}, that is,  $E(a_0, a_1)_T = E(a_0, a_{m-1})_T = 1$ and $E(a_0, \{a_2, \ldots, a_{m-2}\})_T = 0$.

%

 

%

\subsection{Constructing the slack vector $s$} 


\begin{figure}[htb]
	\centering
	\begin{tikzpicture}
		\node [draw,circle] at (0,0) (a) {$a$};
		\node [draw,circle] at (0,2) (b) {$b$} edge (a);
		\node [draw,circle] at (2,0) (c) {$c$} edge [dashed] (a) edge [color=green,line width=6pt,opacity=0.2] (a) ;
		\node [draw,circle] at (2,2) (d) {$d$} edge (c) edge [dashed] (b) edge [line width=6pt,opacity=0.2,color=blue] (b);
		\draw [dashed] (a) -- +(-1,3)  (b) -- +(0,1.5) (d) -- +(0,1.5) (c) --  +(1,3);
		\draw [color=red,line width=6pt,opacity=0.2] (a) -- +(-1,3);
		\draw [color=orange,line width=6pt,opacity=0.2] (c) --  +(1,3);
		\draw [dotted,line width=1pt, color=red] (0,1) ellipse (0.75 and 1.75)
			(2,1) ellipse (0.75 and 1.75);
		\draw[line width=1pt,dotted,color=blue]	(1,1) ellipse (2.5 and 2.1);
		\node [color=red]at (-0.95,0.25) () {$u_1$};
		\node [color=red] at (2.95,0.25) () {$u_2$};
		\node [color=blue] at (4,1) () {$u_3$};
		
	\end{tikzpicture}
	\quad \quad \quad
	\begin{tikzpicture}
		\node [circle,draw] at (0,0) (a) {$a$};
		\node [circle,draw] at (2,0) (b) {$b$} edge (a);
		\node[circle,draw]  at (1,1) (u1) {$u_1$} edge (a) edge [color=green,opacity=0.2,line width=6pt,bend left=10] (a) edge [opacity=0.2,line width=6pt,color=red,bend right=10] (a) edge [line width=1pt](b) edge [color=blue,opacity=0.2,line width=6pt](b);
		\node[circle,draw]  at (4,0) (c) {$c$};
		\node [circle,draw] at (6,0) (d) {$d$} edge (c);
		\node [circle,draw] at (5,1) (u2) {$u_2$} edge (c) edge [color=orange,line width=6pt,opacity=0.2,bend left=10] (c) edge [color=green,opacity=0.2,line width=6pt,bend right=10] (c) edge (d) edge [color=blue,opacity=0.2,line width=6pt] (d) edge (u1) edge [color=blue,opacity=0.2,line width=6pt,bend right=6] (u1) edge [color=green,opacity=0.2,line width=6pt,bend left=6] (u1);
		\node[circle,draw]  at (3,2) (u3) {$u_3$} edge (u1) edge [color=red,opacity=0.2,line width=6pt] (u1) edge (u2) edge [color=orange,line width=6pt,opacity=0.2] (u2) ;
	\end{tikzpicture}
	\caption{\small An example of part of a hierarchy with three "triangles". The graph on the left shows part of a feasible LP solution where dashed (and sometimes colored) edges have fraction $1/2$ and solid edges have fraction 1. The dotted ellipses on the left show the min-cuts $u_1,u_2,u_3$ in the graph. (Each vertex is also a min-cut). On the right is a representation of the corresponding hierarchy. Triangle $u_1$ corresponds to the cut $\{a,b\}$, $u_2$ corresponds to $\{c,d\}$ and $u_3$ corresponds to $\{a,b,c,d\}$. Note that, for example, the edge $\{a,c\}$, represented in green,  is in $\delta(u_1)$,  $\delta(u_3)$, and inside $u_3$. In this example, $\p (\{a,b\}) = u_1$ and  $\p (\{a,c\}) = u_2$. Triangle $u_1$ is happy if $(\delta(a)\smallsetminus \{a,b\})_T = (\delta(b)\smallsetminus \{a,b\})_T = 1$.  }
	\label{fig:hierarchywithtriangle}
\end{figure}
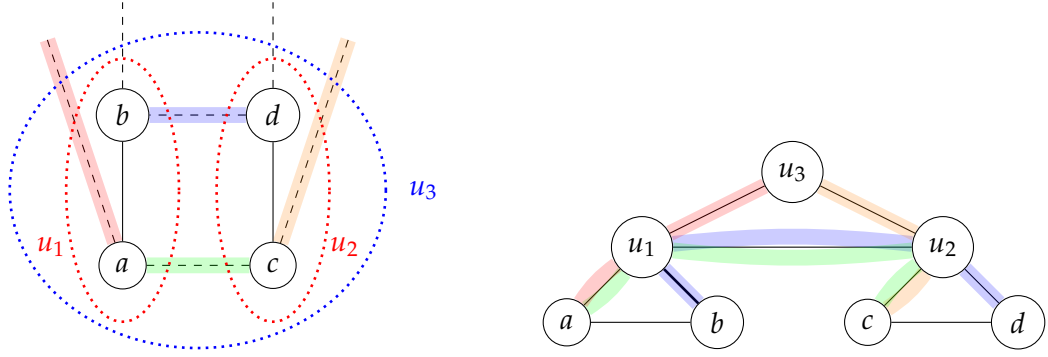

 Our main remaining task is to explain how to use the hierarchy $\cH$ to choose a slack vector $s$ that has {\em negative} expected value, specifically, has $\E{s_e} = -\Omega(\decrease x_e)$ for each edge, while ensuring that all $O$-Join constraints coming from $\cH$ are satisfied. 

 As mentioned in \cref{sec:overview}, the approach taken is to set $s_e$ to be negative (i.e. {\em reduce} it) when certain special cuts $e$ is on are even in the tree and therefore induce no $O$-join constraint. Roughly speaking, it works as follows: For each LP edge $f$, consider the lowest cut $S$ in the hierarchy that contains both endpoints of $f$. We call this cut $\p(f)$ (for "parent of $f$").
  Let $\bbe = \{u,v\}$ (where $u$ and $v$ are children of $S$ in $\cH$; recall $u,v$ are subsets of vertices) be the set of all edges $f = \{u', v'\}$ such that $u' \in u$ and $v' \in v$. 
  
 Cuts in $\cH$ are separated into three types. If $S \in \cH$ has at least three children and it is not an outer polygon cut, call it a \textit{degree cut}. If it has exactly two children, call it a \textit{triangle cut}. The remaining cuts, as defined above, are outer polygon cuts.

 If $\p(f)$ is a degree cut, then  set $s_f :=  - 0.57\decrease x_f$ for all $f \in \bbe$ whenever the event that $\delta(u)_T$ and $\delta(v)_T$ are both even in the tree occurs. Call $f$ ``good" if this event occurs with constant probability. Furthermore, it is shown that every cut $u$ with $\p(u)$ a degree cut contains a $\Omega(1)$ fraction of good edges.
 
 On the other hand, when $\p(f)$ is an outer polygon cut or a "triangle" cut (see \cref{fig:hierarchywithtriangle}), set $s_f = -  \decrease x_f$ for all $f \in \bbe$ whenever $p(f)$ is happy (as defined above).   Thus, when a polygon  is happy, {\em all} edges $e$ whose parent is that cut have their slack $s_e$ reduced simultaneously. 
 Moreover, the event that $p(f)$ is happy for a polygon cut (or triangle cut) occurs with constant probability.
    
However, regardless of the type of cut $p(f)$ is,   setting $s_f$ to a negative value can be problematic for the feasibility of other cuts lower down in the hierarchy that contain $f$. Therefore, when a cut $S'$ lower down in the hierarchy such that $f \in \delta(S')$ is  odd in the tree, the slack of {\em other} edges  in $\delta(S')$ are increased to compensate for the reduction in $s_f$, (i.e., to maintain feasibility of $y$ for the cut $S'$). 
 
 The challenge is to do all of this in a way that still guarantees that  overall $\E{s_e} < - \eps \decrease x_e$, while simultaneously ensuring that for any cut $S \in \cH$ if $\delta(S)_T$ is odd, $\sum_{e \in \delta(S)} s_e \ge 0$. Showing this is involved and requires careful probabilistic arguments that rely on the fact that the tree is sampled from a max-entropy distribution. We refer the reader to \cite{KKO21} for the details.
 

 \subsection{Proof of \cref{thm:main} using \cref{thm:maintechnical}}\label{subsec:proofofmain}

\maintechnical*
\begin{proof}[Proof of \cref{thm:main}]
Let $x^0$ be an extreme point solution of LP \eqref{eq:tsplp}, with support $E_0$ and let $x$ be $x^0$ restricted to $E$. By \cref{fact:sptreepolytope} $x$ is in the spanning tree polytope. 
Let $\mu=\mu_{\lambda^*}$ be the max entropy distribution with marginals $x$, and let $s,s^*$ be as defined in \cref{thm:maintechnical}.
We will define $y:E_0\to\R_{\geq 0}$ such that:
$$y_e=\begin{cases}
x_e/2+s_e+s^*_e & \text{if } e\in E\\
\infty & \text{if } e=e_0
\end{cases}
$$	
We will show that $y$ is a feasible solution
to \eqref{eq:tjoinlp}.
First, observe that for any $S$ where $e_0\in\delta(S)$, we have $y(\delta(S))\geq 1$. Otherwise, we assume $u_0,v_0\notin S$.
If $S$ is an $\eta$-near min cut and $\delta(S)_T$ is odd, then by property (ii) of \cref{thm:maintechnical}, we have
$$ y(\delta(S)) = \frac{x(\delta(S))}{2}+ s(\delta(S))+s^*(\delta(S))\geq 1.$$
On the other hand, if $S$ is not an $\eta$-near min cut, then
$$y(\delta(S)) \geq (\frac{1}{2}-\decrease)x(\delta(S)) \ge (\frac{1}{2}-\decrease) \cdot (2+\eta) = 1+\frac{\eta}{2}-2\decrease -\decrease \eta$$
where in the first inequality we used property (i) of \cref{thm:maintechnical} which gives $s_e\geq -x_e \decrease$ with probability 1 along with the fact that $s^*$ is non-negative.
Therefore, choosing $\beta = \frac{\eta}{4+2\eta}$ ensures that $y$ is a feasible $O$-join solution.

Finally, using $c(e_0)=0$ and part (iii) of \cref{thm:maintechnical},
\begin{align*}
\E{c(y)}&= c(x)/2 + \E{c(s)} +\E{c(s^*)}\\
&\leq c(x)/2 - \eps_P \decrease \cdot \frac{1}{3}c(x) +125\eta \cdot \decrease  \cdot c(x)\\
&\leq (1/2-\frac{1}{6}\eps_P\decrease) \cdot c(x)
\end{align*}
choosing $\eta$ such that 
\begin{equation}\label{eq:whatiseta}
 	125\eta= \frac{1}{6}\eps_P
 \end{equation}

Now, we are ready to bound the approximation factor of our algorithm.
First, since $x^0$ is an  extreme point solution of \eqref{eq:tsplp}, $\min_{e\in E_0} x^0_e\geq \frac{1}{n!}$. So, by \cref{thm:maxentropycomp}, in polynomial time\footnote{Since the claim that the integrality gap is bounded below $3/2$ does not depend on the running time, it may appear that this step is unnecessary. However, we need to discuss the running time here because we are giving a stronger result that the max entropy algorithm returns a solution of expected cost at most $(3/2-\eps)c(x)$ in polynomial time.} we can find $\lambda:E\to\R_{\geq 0}$ such that for any $e\in E$, $\PP{\mu_\lambda}{e}\leq x_e(1+\delta)$ for some $\delta$ that we fix later. It follows that
$$ \sum_{e\in E} |\PP{\mu}{e} - \PP{\mu_{\lambda}}{e}| \leq n\delta.$$
By stability of maximum entropy distributions (see  \cite[Thm 4]{SV19} and references therein), we have that $\norm{\mu-\mu_\lambda}_1\leq O(n^4\delta)=:q$. Therefore, {for some $\delta \ll n^{-4}$ we get} $\norm{\mu-\mu_{\lambda}}_1=q\leq \frac{\eps_P\eta}{100}$. That means that 
$$ \EE{T\sim\mu_{\lambda}}{\text{min cost matching}} \leq \EE{T\sim\mu}{c(y)}+  q (c(x)/2) \leq \left(\frac12-\frac{1}{6}\eps_P\decrease 
 + \frac{\eps_P \eta }{100}\right)c(x), $$
where we used that for any spanning tree the cost of the minimum cost matching on odd degree vertices is at most $c(x)/2$.
Finally, since $\EE{T\sim\mu_\lambda}{c(T)}\leq c(x)(1+\delta)$, $\eps_P=3.12\cdot 10^{-16}$, and $\eta=4.16 \cdot 10^{-19}$ (from \eqref{eq:whatiseta}) and $\decrease = \eta/(4+2\eta)$, we get a $3/2-10^{-36}$ approximation algorithm (compared to $c(x)$).
\end{proof}

\section{Conclusion}
In addition to the obvious question of improving the bound proved in this paper, we conclude with some additional  questions.
\paragraph{Open Directions Related to Structure of Near Min Cuts.}
Given a fractionally 2-edge connected graph $G$ let $\cN_\eta$ be the set of $\eta$-near min cuts of $G$. 
Let $\C$ be a connected component of $\cN_\eta$ with $|\C|>2$ with corresponding polygon $P$.

\begin{itemize}
\item Although \cref{thm:uncrossing} characterizes the edges of $G$ w.r.t. $P$ to some extent, we are still far from a perfect characterization. We find the following question illuminating in this direction: perhaps for sufficiently small $\eta$, is it true that for any cut $S\in \C$, $x(E(S, I(P))\leq O(\eta)$, where $I(P)\subseteq \cA(P)$ is the set of inside atoms of $P$?
\item Is there a compact representation of all $2/3$-near mincuts \cite{BG08,Goe21}? (This is a natural target as $2/3$ is the limit $\alpha$ where the number of $\alpha$-mincuts is 
always at most ${n\choose 2}$ in a graph with $n$ vertices; see \cite{GR95}.)
\end{itemize}

\paragraph{Open Directions Related to TSP.}

\begin{enumerate}
\item Since this analysis does not depend on the integral optimum Hamiltonian cycle (in contrast with \cite{KKO21}), the following now appears more approachable: is it possible to de-randomize the max entropy algorithm and obtain a deterministic $3/2-\eps$ approximation for metric TSP?
One path to solving this problem leads to the following question: is it possible to design an efficient algorithm that outputs  a convex combination of at most \textit{polynomially} many spanning trees such that their expected cost is at most OPT and the expected cost of the minimum matching on their odd degree vertices is at most $(1/2-\eps)OPT$?
\item As we discussed in  \cref{sec:overview}, in our analysis first we ``satisfy'' cuts crossed on both sides by choosing a random slack vector.
Then, we delete these cuts and look at the hierarchy of cuts crossed on at most one side, $\cN_{\eta,\leq 1}$, and design another slack vector for every connected component of crossed cuts in $\cN_{\eta,\leq 1}$. The remaining cuts form a hierarchy as we defined in  \cref{app:proofbeforetechnical}. Unfortunately, working with this structure is not identical to satisfying a laminar family of cuts. For one, we do not guarantee that $x$ has no near minimum cuts not given in the hierarchy: we simply show that if they exist, they can be ignored in the constraints of the $O$-Join. This limitation seems to prevent the use of more direct combinatorial approaches such as \cite{HN19,GLLM21}.
So, we leave it as an open problem as to whether it is possible to construct a random slack vector for all cuts in $\cN_{\eta,1},\cN_{\eta,2}$ simultaneously. Such a vector would preserve the original structure of near min cuts; thus when focusing on the hierarchy one can assume no other near min cuts exist. The answer to this question may lead to more significant improvements to the approximation ratio for metric TSP.	
\end{enumerate} 

\printbibliography

\appendix
\section{Cuts crossed on one side}\label{app:oneside}

\subsection{Cuts crossed on one side}\label{sec:oneside}
In \cref{thm:cutsbothsideswithinside}, we found a vector which satisfies all cuts crossed on both sides. Consequently, we can study the structure of cuts which remain after deleting all cuts crossed on both sides, i.e. connected components of cuts crossed on one side. The arguments in this section closely follow Section 4.3 of \cite{KKO21}. Even though in that work these families were handled using OPT edges, the extension to charging to LP edges is very natural in this setting and requires little modification. 

\begin{lemma}\label{obs:one-side-interval}
Suppose $\mathcal{C}$ is a connected component of cuts in $\cN_{\eta,1}$ with $|\cC| \ge 2$. If $\eta \le \frac{2}{5}$, the corresponding polygon $P$ of $\mathcal{C}$ has no inside atoms. 
\end{lemma}
\begin{proof}
By \cref{def:inside}, to show that the polygon of $\cC$ has no inside atoms it is sufficient to show that there are no $k$-cycles for any integer $k$. 
Since $\eta\leq 2/5$, by \cref{lem:no-kcycle} there are no $3$ or $4$-cycles.
By way of contradiction suppose there was a $k$-cycle $C_1,\dots,C_k \in \cC$ with $k\geq 5$. Then, perhaps after renaming, we can assume that $C_1,C_2,C_3$ do not contain the root, $C_1$ and $C_3$ each cross $C_2$, and $C_1 \cap C_3 = \emptyset$. Then by \cref{lem:characterize-N2} (below), $C_2 \in \cN_{\eta,2}$, which is a contradiction since we assumed $\cC \subseteq \cN_1$. Therefore, $P$ has no $k$-cycles for $k\geq 5$. So, $P$ has no inside atoms. 
\end{proof} 

Also note that by \cref{lem:characterize-N2}, if $\cC$ is a connected component of cuts in $\cN_{\eta,1}$, then every cut $C \in \cC$ is crossed on one side in the polygon of $\cC$. (In other words, deleting the cuts in $\cN_{\eta,2}$ does not allow a cut previously crossed on one side to be crossed on both sides in its new polygon.)

\subsection{Notation and results from \cite{KKO21}}

As before suppose $\cC$ is a connected component of cuts crossed on one side with corresponding polygon $P$. Now assume $P$ has outside atoms $a_0,\dots,a_{m-1}$, and WLOG assume $a_0$ is the root (recall there are no inside atoms). 

\begin{definition}[Leftmost and Rightmost cuts]
	  We call any cut $C\in {\cal C}$ with leftmost atom $a_1$  a {\em leftmost} cut of $P$, and any cut $C\in {\cal C}$ with rightmost atom $a_{m-1}$  a {\em rightmost} cut of $P$. 
		We also call $a_1$ the leftmost atom of $P$ (resp.  $a_{m-1}$ the rightmost atom).
\end{definition}

In \cite{KKO21}, it was shown that polygons of cuts crossed on one side have a simple structure. In particular, they look like a near-integral cycle: 

\begin{theorem}[Structure of Polygons of $\cN_{\eta,1}$ (Theorem 4.9 from \cite{KKO21})]\label{thm:poly-structure}
For  $\eps_{\eta}\geq 7\eta$ and any polygon of cuts crossed on one side with atoms $a_0...a_{m-1}$ (where $a_0$ is the root) the following holds:
	\begin{itemize}
	\item For all adjacent atoms $a_i,a_{i+1}$ (also including $a_0,a_{m-1}$), we have $x(E(a_i,a_{i+1})) \ge 1-\eps_\eta$. 
	\item All atoms $a_i$ (including the root) have $x(\delta(a_i)) \le 2+\eps_\eta$. 
	\item $x(E(a_0, \{a_2,\dots,a_{m-2}\}))\leq \eps_{\eta}$.
	\end{itemize}
\end{theorem}

\begin{definition}[$A,B,C$-Polygon Partition]\label{def:abcpolygonpartitioning}
The $A,B,C$-polygon partition of a polygon $P$ is a partition of edges of $\delta(a_0)$ into sets $A=E(a_1,a_0)$ and $B=E(a_{m-1},a_0)$, $C=\delta(a_0)\smallsetminus A\smallsetminus B$. 
\end{definition}

\begin{definition}[Happy Polygons]
For a spanning tree $T$, we say that a polygon  $P$ of cuts crossed on one side is {\em happy} if 
$$A_T \text{ and } B_T \text{ odd}, C_T=0.$$

We say that $P$ is {\em left-happy} (respectively {\em right-happy}) if 
$$A_T \text{ odd}, C_T=0,$$
(respectively $B_T \text{ odd}, C_T=0$).
\end{definition}


\begin{definition}[Happy Cut]We say a leftmost cut $L\in \cC$  is {\em happy} if
$$ E(L, \overline{L\cup a_0})_T=1. 
$$
Similarly, the leftmost atom $a_1$ is {\em happy} if $E(a_1,\overline{a_0\cup a_1})_T=1$. Define  rightmost cuts in $u$ or the rightmost atom in $u$ to be happy, similarly.
\end{definition}
Note that, by definition, if leftmost cut $L$ is happy and $P$ is left happy then $L$ is even, i.e., $\delta(L)_T=2$. Similarly, $a_1$ is even if it is happy and $P$ is left-happy.

\begin{definition}[Relevant Cuts]
Define the family of relevant cuts of a polygon $P$ representing a connected component $\cC \subseteq \cN_{\eta,1}$ as follows:
$$\cC_{+}={\cal C}\cup \{a_i:  1\leq i\leq m-1 \land x(\delta(a_i))\leq 2+\eta\}.$$
\end{definition}

\begin{lemma}[{\cite[Lemma 4.28]{KKO21}}]\label{lem:4cutmapping}
There is a {\em mapping} of cuts in $\cC_+$ to the collections of edges $E(a_1,a_2),\dots,E(a_{m-2},a_{m-1})$ such that each set $E(a_i,a_{i+1})$ has at most 4 cuts mapped to it, every cut $C \in \cC_+$ containing atoms $a_i$ through $a_j$ is mapped to either $E(a_{i-1},a_i)$ or $E(a_j,a_{j+1})$ (or both), and every atom of the polygon in ${\cal C}'$ gets mapped to two (not necessarily distinct) groups of edges $E(a_i,a_{i+1})$, $E(a_j,a_{j+1})$. 
\end{lemma}

Note in the following three statements, we gain a factor of two compared to \cite{KKO21} as we look at $\eta$-near min cuts instead of $2\eta$-near min cuts.

\begin{lemma}[{\cite[Lemma 4.26]{KKO21}}]\label{lem:Astrictpeven}
For every cut $A\in {\cal C}$ that is not a leftmost or a rightmost cut, $\P{\delta(A)_T=2} \geq 1-11\eta$.
\end{lemma}

\begin{lemma}[{\cite[Lemma 4.27]{KKO21}}]\label{lem:atoms-even-whp}
	For any atom $a_i\neq a_0$ that is not the leftmost or the rightmost atom we have
	$$\P{\delta(a_i)_T = 2} \ge 1-21\eta.$$
\end{lemma}

\begin{lemma}[{\cite[Lemma 4.30]{KKO21}}]\label{lem:even-cuts-cond-on-happy}
For every leftmost or rightmost cut $A$ in $P$ that is an $\eta$-near min cut, $\P{A\text{ happy}}\geq 1-5\eta$, and for the leftmost atom $a_1$ (resp. rightmost atom $a_{m-1}$), if it is an $\eta$-near min cut then  $\P{a_1\text{ happy}}\geq 1-12\eta$ (resp. $\P{a_{m-1}\text{ happy}}\geq 1-12\eta$).
\end{lemma}

\subsection{Main theorem for cuts crossed on one side}

The following is our extension of Theorem 4.24 in \cite{KKO21}. This is the key theorem used to deal with components cuts crossed on one side and the atoms in $\cN_\eta$ which compose them.

\begin{restatable}[Happy Polygons (Similar to Theorem 4.24 in \cite{KKO21})]{theorem}{thmcrossedoneside}\label{thm:crossed-one-side}
	Let $G=(V,E,x)$ for an LP solution $x$.
	Let $\mu$ be an arbitrary distribution of spanning trees with marginals $x$. For any $\alpha>0$, $\eta\leq 1/10$, and $\eps_\eta=7\eta$, there is a random vector $s^*:E\to\R_{\geq 0}$ (as a function of $T\sim\mu$) such that 
	\begin{itemize}
		\item For a connected component ${\cal C}$ of cuts crossed on one side with corresponding polygon $P$ and atoms $a_0,a_1...a_{m-1}$ and cycle partition $A,B,C$ the following holds:
		\begin{itemize}
		\item For any cut $S\in \cC_+$ which is not a leftmost/rightmost cut/atom if  $\delta(S)_T$ is odd then we have $s^*(\delta(S))\geq \alpha(1-\epsilon_\eta)$,
		\item If $P$ is left happy, then for any $S\in \cC_+$ that is a leftmost cut or the leftmost atom, if $\delta(S)_T$ is odd, then we have 
		$s^*(\delta(S))\geq \alpha(1-\epsilon_\eta)$. 
		\item Similarly, if $P$ is right happy then for any cut $S\in \cC_+$ that is  a rightmost cut or the rightmost atom, if $\delta(S)_T$ is odd, then $s^*(\delta(S))\geq \alpha(1-\epsilon_\eta)$.		
		\end{itemize}
		\item $\E{s^*_{e}}\leq 44\alpha\eta x_e$ for all $e \in E$.
	\end{itemize}
\end{restatable}

%

Before proving the theorem, we study a special case.
\begin{lemma}[\cref{thm:crossed-one-side} Holds for Triangles]\label{lem:trianglereduction}
	Let $S=X\cup Y$ where $X,Y,S$ are $\eps_{\eta}$-near min cuts which do not cross. Then, letting $X$ be $a_1$ and $Y$ be $a_2$ (and $a_0=\overline{X\cup Y}$) \cref{thm:crossed-one-side} holds. 
\end{lemma}
\begin{proof}
In this case this system has cycle partition $A=E(a_1,a_0), B=E(a_2,a_0), C=\emptyset$.
For the edges $E(a_1,a_2)$ we define an increase event $\cI$ when at least one of $T\cap E(X)$, $T\cap E(Y)$, $T\cap E(S)$ is not a tree.
Whenever this happens we define $s^*_{e}=\alpha x_e$ for all $e \in E(a_1,a_2)$.
If $S$ is left-happy we need to show when $\delta(X)_T$ is odd, then $s^*(\delta(X))\geq \alpha(1-\epsilon_\eta)$. This is because when $S$ is left-happy we have $A_T=1$ (and $C_T=0$), so either the increase event $\cI$ does not happen and we get $\delta(X)_T=2$ or it happens in which case $s^*(\delta(X)) = \alpha \cdot x(E(a_1,a_2)) \geq \alpha(1-\epsilon_\eta)$ by \cref{lem:sub-NMC-shared}.
Finally, observe that by \cref{lem:treeoneedge}, 
$\P{\cI}\le 3 \eps_\eta/2$, so
$\E{s^*_{e}}=1.5\eps_\eta \alpha x_e < 44\alpha\eta x_e$ for all $e \in E(a_1,a_2)$. 
\end{proof}

\begin{proof}[Proof of \cref{thm:crossed-one-side}]
Fix a connected component $\cC$ of ${\cN_1}$ with corresponding polygon $P$. Fix $1<i<m$. By \cref{thm:poly-structure} $x(E(a_{i-1},a_i)) \ge 1-\epsilon_\eta$. 

For the at most four cuts mapped to $E(a_{i-1},a_i)$ in \cref{lem:4cutmapping}, we define the following three events:
\begin{enumerate}[i)]
\item A leftmost cut mapped to $E(a_{i-1},a_i)$ is not happy.
\item A rightmost cut mapped to $E(a_{i-1},a_i)$ is not happy.
\item A cut which is not leftmost or rightmost mapped to $E(a_{i-1},a_i)$ is odd.
\end{enumerate}
Observe that the cuts in (i) and (ii) are assigned to $E(a_{i-1},a_i)$ in \cref{lem:4cutmapping}. We say an atom $a$ is singly-mapped to $E(a_{i-1},a_i)$ if in the matching $a$ is only mapped to $E(a_{i-1},a_i)$ once, otherwise we say it is doubly-mapped to $E(a_{i-1},a_i)$.

We say an event $\cI(E(a_{i-1},a_i))$ occurs if either (i), (ii), or (iii) occurs. If $\cI(E(a_{i-1},a_i))$ occurs then for all $e \in E(a_{i-1},a_i)$, we set:
\begin{align*}
	s^*_{e} = \begin{cases}\alpha x_e & \text{If (i),(ii), or (iii) occurred for at least one non-atom cut in $\cC'$, or for an atom}\\ &\text{which is doubly-mapped to $E(a_{i-1},a_i)$} \\
		\alpha x_e /2 & \text{Otherwise.}
	 \end{cases}
\end{align*}
If $\cI(E(a_{i-1},a_i))$ does not occur we set $s^*_{e} = 0$ for all $e \in E(a_{i-1},a_i)$. 


First, observe that for any non-atom cut $S \in \cC_+$ (i.e. any relevant cut) that is not a leftmost or a rightmost cut/atom, if $\delta(S)_T$ is odd, then if $E(a_{i-1},a_i)$ is the set of edges that $S$ is mapped to, it satisfies $s^*(\delta(S)) \ge \alpha \cdot x(E(a_{i-1},a_i)) \ge \alpha(1-\epsilon_\eta)$. So, these cuts satisfy the conditions of the theorem.

The same inequality holds for non-leftmost/rightmost atom cuts $a\in {\cal C}'$ which are doubly-mapped to $E(a_{i-1},a_i)$. For non-leftmost/rightmost atom cuts $a\in {\cal C}'$ which are singly-mapped to $E(a_{i-1},a_i)$, $a$ is mapped (possibly even twice) to another edge $E(a_{j-1},a_j)$ (note $j=i-1$ or $i+1$), and in this case $s^*(\delta(S)) \ge \alpha/2 \cdot 2(1-\epsilon_\eta) = \alpha(1-\epsilon_\eta)$, and again the above inequality holds.

Now, suppose $S \in \cC$ is a leftmost cut of $P$ and $\delta(S)_T$ is odd, and the rightmost atom of $S$ is $a_{i-1}$ (i.e. it is mapped to $E(a_{i-1},a_i)$). If $P$ is not left-happy then there is nothing to prove. If $P$ is left-happy, we may assume $S$ is not happy. Then $\cI(E(a_{i-1},a_i))$ happens, so as in the above inequality $s^*(\delta(S)) \ge \alpha(1-\epsilon_\eta)$. We obtain the same condition for rightmost cuts and leftmost/rightmost atoms that are assigned to $P$ (note leftmost/rightmost atoms are always doubly-mapped: $a_1$ to $E(a_{1},a_2)$ and $a_{m-1}$ to $E(a_{m-2},a_{m-1})$).



It remains to upper bound $\E{s^*_e}$ for any edge $e\in E(a_{i-1},a_i)$. By \cref{lem:4cutmapping}, at most four cuts are mapped to $E(a_{i-1},a_i)$.

First suppose exactly one atom is doubly-mapped to $E(a_{i-1},a_i)$. Then there are at most three cuts mapped to $E(a_{i-1},a_i)$, including that atom. The probability of an event of type (i) or (ii) occurring for the leftmost or rightmost atom is at most $1-12\eta$ by \cref{lem:even-cuts-cond-on-happy}. Atoms which are not leftmost or rightmost are even with probability at least $1-21\eta$ by \cref{lem:atoms-even-whp}. Therefore, in the worst case, the doubly-mapped atom is not leftmost or rightmost. For the remaining two cuts, leftmost and rightmost cuts are happy with probability at least $1-5\eta$ by \cref{lem:even-cuts-cond-on-happy}, and (non-atom) non leftmost/rightmost cuts are even with probability at least $1-11\eta$ by \cref{lem:Astrictpeven}. Therefore in the worst case the remaining two (non-atom) cuts mapped to $E(a_{i-1},a_i)$ are not leftmost/rightmost. Therefore, if an atom is doubly-mapped to $E(a_{i-1},a_i)$, for any $e \in E(a_{i-1},a_i)$ we have
$$\E{s^*(e)}\leq 21\eta \alpha x_e + 2\cdot 11\eta \alpha x_e < 44\eta \alpha x_e$$
Note if two atoms are doubly-mapped to $E(a_{i-1},a_i)$, there are no other mapped cuts and in the worst case the atoms are not leftmost/rightmost, so for any $e \in E(a_{i-1},a_i)$,
$$\E{s^*(e)}\leq 2\cdot 21\eta \alpha x_e < 44\eta \alpha x_e$$

Otherwise, any atoms mapped to $E(a_{i-1},a_i)$ are singly-mapped. In this case, if only an atom cut is odd/unhappy, we set $s^*(e) = x_e\alpha/2$. The probability of an event of type (i) or (ii) occurring for the leftmost or rightmost atom is at most $1-12\eta$ by \cref{lem:even-cuts-cond-on-happy}, so we can bound the contribution of this event to $\E{s^*(e)}$ by $12\eta \alpha x_e/2$. Atoms which are not leftmost or rightmost are even with probability at least $1-21\eta$ by \cref{lem:atoms-even-whp}, and so we can bound their contribution by $21\eta \alpha x_e/2$. Therefore, in the worst case four non-leftmost/rightmost \textit{non}-atom cuts are mapped to $E(a_{i-1},a_i)$, in which case, for any $e \in E(a_{i-1},a_i)$,
$$\E{s^*(e)}\leq 4\cdot 11\eta \alpha x_e = 44\eta \alpha x_e$$ as desired.
\end{proof}
%
\section{Proof of \cref{thm:maintechnical}}\label{app:proofbeforetechnical}

In this section, we use the previous section and \cref{thm:cutsbothsideswithinside} to prove  \cref{thm:maintechnical}. 

\begin{definition}[Hierarchy, \cite{KKO21}]\label{def:hierarchy}
\hypertarget{tar:hierarchy}{For an LP solution $x^0$ with support $E_0=E\cup \{e_0\}$ where $x$ is $x^0$ restricted to $E$, a hierarchy ${\cal H}\subseteq \cN_{\eps_\eta}$ is a {\em laminar} family  with root $V\smallsetminus \{u_0,v_0\}$, where every cut $S\in \cH$ is called either a ``near-cycle" cut or a degree cut. In the special case that $S$ has exactly two children we call it a triangle cut. Furthermore, every cut $S$ is the union of its children. 
For any (non-root) cut $S\in \cH$, define the parent of $S$, $\p(S)$, to be the smallest cut $S'\in\cH$ such that $S\subsetneq S'$.}

\hypertarget{tar:AS}{For a cut $S\in \cH$, 
let $\cA(S):=\{a\in \cH: \p(a)=S\}$. If $S$ is called a ``near-cycle" cut, then we can order cuts in $\cA(S)$, $a_1,\dots,a_{m-1}$ such that 
\begin{itemize}
\item $A=E(\overline{S},a_1), B=E(a_{m-1},\overline{S})$ satisfy $x(A),x(B)\geq 1-\eps_\eta$.
\item For any $1\leq i<m-1$, $x(E(a_i,a_{i+1}))\geq 1-\eps_\eta$.
\item $C=\cup_{i=2}^{m-2} E(a_i,\overline{S})$ satisfies $x(C)\leq \eps_\eta$.
\end{itemize}}

We call the sets $A,B,C$ the ``near-cycle" partition of edges in $\delta(S)$. We say $S$ is left-happy when $A_T$ is odd and  $C_T=0$ and right happy when $B_T$ is odd and $C_T=0$ and happy when $A_T,B_T$ are odd and $C_T=0$.

We abuse notation and for an  edge $e=(u,v)$ that is not a neighbor of $u_0,v_0$, we write $\p(e)$ to denote the smallest\footnote{in the sense of the number of vertices that it contains} cut $S'\in \cH$ such that $u,v\in S'$. We say edge $e$ is a \textbf{bottom edge} if $\p(e)$ is a polygon cut and we say it is a \textbf{top edge} if $\p(e)$ is a degree cut.

%
%
\end{definition}

The terminology of the above differs slightly from \cite{KKO21}, where we replace ``polygon" cut with ``near-cycle" cut and ``polygon" partition with ``near-cycle" partition.

By \cref{thm:poly-structure}, an example of a near-cycle cut is the union of non-root atoms  of a connected component of cuts crossed on one side (i.e. its outer polygon cut). Another example is the non-root atoms of a connected component of minimum cuts (i.e. a cycle of a cactus of length at least four).

In the following, we will define a hierarchy $\cH$ satisfying the above definition such that every cut $S \in \cN_{\eta,\le 1}$ is either in $\cH$ or there is a near-cycle cut $P \in \cH$ representing a connected component $\cC$ such that $S \in \cC$.

We will use the following ``main payment theorem" from ~\cite{KKO21}.

\begin{restatable}[Main Payment Theorem (4.33 in \cite{KKO21})]{theorem}{paymentmain}
\label{thm:payment-main}
For an LP solution $x^0$ where $x$ is $x^0$ restricted to $E$ and a hierarchy $\cH$ for some $\eps_\eta\leq 10^{-10}$ and any $\decrease > 0$,
the maximum entropy distribution $\mu$ with marginals $x$ satisfies the following:
\begin{enumerate}[i)]
\item There is a set of {\em good} edges $E_g\subseteq E\smallsetminus \delta(\{u_0,v_0\})$ such that any bottom edge $e$ is in $E_g$ and  for any (non-root) $S\in \cH$ such that $\p(S)$ is not a near-cycle cut, we have $x(E_g\cap \delta(S))\geq 3/4$. 
\item There is a random vector $s:E_g \to \R$  (as a function of $T\sim\mu$) such that for all $e$, $s_e\geq -x_e \decrease$ (with probability 1), and \label{payment:non-near-min-cuts}
\item If a near-cycle cut $S$ with cycle partition $A,B,C$ is not left happy, then for any set $F\subseteq E$ with $\p(e)=S$ for all $e\in F$ and $x(F)\geq 1-\eps_\eta/2$, we have
$$ s(A)+s(F)+s^-(C)\geq 0,$$
where $s^-(C)=\sum_{e\in C} \min\{s_e,0\}$.
A similar inequality holds if $S$ is not right happy.
\item 
For every cut $S\in \cH$ such that $\p(S)$ is not an near-cycle cut, if $\delta(S)_T$ is odd, then $s(\delta(S))\geq 0$. \label{payment:satisfy-non-poly-cuts}
\item 
For a good edge $e\in E_g$, $\E{s_e} \le  - \eps_P \decrease x_e$ (where $\eps_P \ge 3.12 \cdot 10^{-16}$) . 
\end{enumerate}
\end{restatable}

In \cref{app:proofbeforetechnical}, we show how the main payment theorem along with \cref{thm:crossed-one-side} and \cref{thm:cutsbothsideswithinside} implies the following:

\begin{restatable}{theorem}{beforetechnical}\label{lem:beforetechnicalthm}
Let $x^0$ be a feasible solution of LP \eqref{eq:tsplp} with support $E_0=E\cup\{e_0\}$ with $x$ the restriction of $x^0$  to $E$.
Let $\mu$ be the max entropy distribution with marginals $x$.
For $\eta\leq 10^{-12}$, $\decrease > 0$, there is a set $E_g\subset E\smallsetminus \delta(\{u_0,v_0\})$  of {\em good} edges
and two functions $s: E_0\rightarrow \R$ and $s^*: E \rightarrow \R _{\ge 0}$ (as functions of $T\sim\mu$) such that
	\begin{itemize}
\item[(i)] 	For each edge $e \in E_g$, $s_e \ge -x_e \decrease$ and for any $e\in E\smallsetminus E_g$, $s_e=0$.
\item[(ii)] For each $\eta$-near min cut $S$, including those for which $\{u_0,v_0\}\in\delta(S)$,  if $\delta(S)_T$ is odd, then $  s(\delta(S)) + s^*(\delta(S)) \ge  0.$
\item[(iii)] We have $\E{s_e} \le -\epsilon_P \decrease x_e$  for all edges $e \in E_g$  and $\E{s^*_{e}} \le 125 \eta \decrease x_e$ for all edges $e \in E$, where $\eps_P$ is defined in \cref{thm:payment-main}.
\item[(iv)] For every cut $S$ crossed on at most one side such that $S\neq \overline{\{u_0,v_0\}}$, $x(\delta(S)\cap E_g) \ge 3/4.$
\end{itemize}
\end{restatable}

Now we will use it to prove the appendix theorem, which we already showed implies \cref{thm:main}:
\maintechnical*
\begin{proof}[Proof of \cref{thm:maintechnical}]
Let $E_g$ be the good edges defined in \cref{lem:beforetechnicalthm} and let $E_b:=E\smallsetminus E_g$ be the set of bad edges; in particular, note all edges in $\delta(\{u_0,v_0\})$ are bad edges. We define a new vector $\tilde{s}:E\cup \{e_0\}\to\R$ as follows: 
\begin{equation}\label{eq:tildesdef}\tilde{s}(e)\gets \begin{cases}\infty & \text{if } e=e_0\\
-x_e(4\decrease/5)(1-2\eta) & \text{if } e\in E_b,\\
x_e(4\decrease/3) & \text{otherwise.}
\end{cases}
\end{equation}
Let $\tilde{s}^*$ be the vector $s^*$ from \cref{thm:cutsbothsideswithinside} called with $\alpha = 2\decrease$.
We claim that for any  $\eta$-near minimum cut $S$ such that $\delta(S)_T$ is odd, we have 
$$ \tilde{s}(\delta(S))+\tilde{s}^*(\delta(S))\geq 0.$$
To check this note by (iv) of \cref{lem:beforetechnicalthm} for every set $S\in \cN_{\eta,\leq 1}$ such that $S\neq V\smallsetminus \{u_0,v_0\}$, we have $x(E_g \cap \delta(S)) \ge \frac{3}{4}$, so we have
\begin{equation}\label{eq:tildeScutsok}\tilde{s}(\delta(S))+\tilde{s}^*(\delta(S))\geq \tilde{s}(\delta(S)) = \frac{4\decrease}{3} x(E_g\cap\delta(S)) - \frac{4\decrease}{5}(1-2\eta)x(E_b\cap\delta(S))\geq 0.	
\end{equation}

For $S=V\smallsetminus \{u_0,v_0\}$, we have $\delta(S)_T=\delta(u_0)_T + \delta(v_0)_T=2$ with probability 1, so condition ii) is satisfied for these cuts as well.
Finally, consider cuts $S\in\cN_{\eta,2}$. By \cref{thm:cutsbothsideswithinside}, if $\delta(S)_T$ is odd, then $\tilde{s}^*(\delta(S)) \ge \alpha(1-\eta) = 2\decrease(1-\eta)$. Therefore, in such a case we have:
\begin{equation}\label{eq:tildeScutstwo}\tilde{s}(\delta(S)) + \tilde{s}^*(\delta(S)) \geq  2\decrease(1-\eta) - \frac{4\decrease}{5}(1-2\eta) x(\delta(S))  \geq 0	
\end{equation}
where we use that $x(\delta(S)) \le 2+\eta$. 

Now, we are ready to define $s,s^*$. Let $\hat{s},\hat{s}^*$ be the $s,s^*$ of \cref{lem:beforetechnicalthm} respectively.
Define $s=\gamma \tilde{s} + (1-\gamma) \hat{s}$ and similarly define $s^*=\gamma\tilde{s}^*+(1-\gamma)\hat{s}^*$ for some $\gamma$ that we choose later. 
We prove all three conclusions of \cref{thm:maintechnical} for $s,s^*$. (i) follows by (i) of \cref{lem:beforetechnicalthm} and  \cref{eq:tildesdef}. (ii) follows by (ii) of \cref{lem:beforetechnicalthm} and \cref{eq:tildeScutsok,eq:tildeScutstwo} above. 
It remains to verify (iii). For edge $e \in E$, $\E{s^*_{e}}\leq 125\eta \decrease x_e$ by (iii) of \cref{lem:beforetechnicalthm} and the construction of $s^*$. On the other hand, by (iii) of \cref{lem:beforetechnicalthm} and \cref{eq:tildesdef},
\begin{align*} 
\E{s_e} \begin{cases}\leq  x_e  (\gamma\frac{4}{3}\decrease - (1-\gamma)\eps_P\decrease ) & \forall e\in E_g,\\	
= -x_e\gamma\cdot (\frac{4}{5}\decrease)(1-2\eta)& \forall e\in E_b.
\end{cases}
\end{align*}
Setting $\gamma=\frac{15}{32}\eps_P$ we get $\E{s_e}\leq -\frac{1}{3}\eps_P\decrease x_e$ for $e\in E_g$ and $\E{s_e}\leq -\frac{1}{3}x_e\decrease\epsilon_P$ for $e\in E_b$ as desired.
\end{proof}

%
\beforetechnical*
\begin{proof}
We start by explaining how to construct $\cH$.
Run the following procedure on $\cN_{\eta,\leq 1}$ (of $x$):
For every connected component ${\cal C}$ of $\cN_{\eta,\leq 1}$, if $|{\cal C}|=1$ then add the unique cut in ${\cal C}$ to the hierarchy. Otherwise, ${\cal C}$ corresponds to a polygon $P$ of cuts crossed on one side with atoms $a_0,\dots,a_{m-1}$ (for some $m>3$).
By \cref{obs:one-side-interval} all these atoms are outside atoms. 
 Add $a_1,\dots,a_{m-1}$ to $\cH$\footnote{Notice that an atom may already correspond to a connected component, in such a case we do not need to add it in this step.} and $\cup_{i=1}^{m-1} a_i$ to $\cH$.
Note that since $x(\delta(\{u_0,v_0\}))=2$, the root of the hierarchy is always $V\smallsetminus \{u_0,v_0\}$.

Now, we name every cut in the hierarchy.
For a cut $S$, if there is a connected component of at least two cuts with union equal to $S$, then call $S$ a near-cycle cut with $A,B,C$ partition as defined in \cref{def:hierarchy}.
If $S$ is a cut with exactly two children $X,Y$ in the hierarchy (i.e. a triangle), then let $A=E(X,\overline{X}\smallsetminus Y)$, $B=E(Y,\overline{Y}\smallsetminus X)$ and $C=\emptyset$.
Otherwise, call $S$ a degree cut.

\begin{fact}[{\cite[Fact 4.34]{KKO21}}] The above procedure produces a valid hierarchy.	
\end{fact}

The following observation simply follows from the fact that the new cuts that we introduce in the above hierarchy, i.e., atoms and union of non-root atoms of a polygon, are not crossed and are never part of a non-singleton connected component.
\begin{fact}\label{fact:nonsingletoncompcorr}
The set of non-singleton connected components $\cC_1,\cC_2,\dots$  that the above procedure produces are in one-to-one correspondence to the set of non-singleton connected components of $\cN_{\eta,\leq 1}$. 
\end{fact}

Let $E_g$ and $s$ be defined as in \cref{thm:payment-main} for the hierarchy defined above, and let {$s_{e_0}=\infty$}.
Also, let $s^*$ be the sum of the $s^*:E\to\R_{\geq 0}$ vectors from \cref{thm:cutsbothsideswithinside} and \cref{thm:crossed-one-side} called with $\alpha = \frac{2+\eta}{1-\eps_\eta}\decrease$.
(i) follows from (ii) of \cref{thm:payment-main}. Then,
$\E{s^*_{e^*}}\leq (18+44)\eta (\frac{2+\eta}{1-\eps_\eta}\decrease) \le  125\eta \beta$ follows from \cref{thm:cutsbothsideswithinside} and \cref{thm:crossed-one-side} and using that $\eta \le 10^{-12}$ and $\eps_\eta = 7\eta$. 
 Also, $\E{s_e}\leq -\eps_P \decrease x_e$ for edges $e\in E_g$ follows from (v) of \cref{thm:payment-main}.
 
 Now, we verify (iv): For any (non-root) cut $S\in \cH$ such that $\p(S)$ is not a near-cycle cut $x(\delta(S)\cap E_g)\geq 3/4$ by (i) of \cref{thm:payment-main}. The only remaining case is $\eta$-near minimum cuts  which are either atoms or near minimum cuts in a polygon. Fix such a set $S$ in a polygon $P$. 
 Let $S'$ be the union of the non-root atoms of $P$. Then by \cref{lem:shared-edges}, $x(\delta(S)\cap \delta(S'))\leq 1+\eps_\eta$. All edges in  $\delta(S)\smallsetminus \delta(S')$ are bottom edges, so by (i) of \cref{thm:payment-main} are in $E_g$. Therefore, $x(\delta(S)\cap E_g))\geq 1-\eps_\eta\geq 3/4$.
 
 It remains to verify (ii): We consider 5 groups of cuts:
 
{\bf Type 1}: Cuts $S$ such that $e_0\in \delta(S)$. Then, since $s_{e_0}=\infty$, $s(\delta(S))+s^*(\delta(S))\geq 0$.

 {\bf Type 2:} Cuts $S\in\cN_{\eta,2}$. By \cref{thm:cutsbothsideswithinside} and the fact that $\alpha=\frac{2+\eta}{1-\eps_\eta}\decrease$, if $\delta(S)_T$ is odd then 
 $$s^*(\delta(S)) \ge \frac{2+\eta}{1-\eps_\eta}\decrease (1-\eta) \ge (2+\eta)\decrease \ge -s(\delta(S))$$
where we use that $s_e\geq -\decrease x_e$ for all edges $e$ and $x(\delta(S)) \le 2+\eta$.
 
 {\bf Type 3}: Cuts $S\in {\cH}\cap \cN_{\eta}$ where $\p(S)$ is not a near-cycle cut. By (iv) of \cref{thm:payment-main} and that $s^*\geq 0$ the inequality follows.

  {\bf Type 4:} Cuts $S$  such that either $S\in\cN_{\eta,\leq 1}\smallsetminus\cH$ or  $S\in \cH\cap\cN_\eta$ and $\p(S)$ is a (non-triangle) near-cycle cut.
In this case either $S$ is an atom or an $\eta$-near minimum cut of a non-singleton connected component $\cC$ of $\cN_{\eta,\leq 1}$ with corresponding polygon $P$ of cuts crossed on one side and the cycle partition $A,B,C$.  If $S$ is not a leftmost cut/atom or a rightmost cut/atom, then  by \cref{thm:crossed-one-side}, whenever $\delta(S)_T$ is odd, we have (similar to Type 2):
\begin{equation}s^*(\delta(S)) \geq \frac{2+\eta}{1-\eps_\eta}\decrease (1-\eps_\eta) = (2+\eta)\beta \ge -s(\delta(S))\end{equation} 

Otherwise, 
suppose $S$ is  a leftmost cut. If $P$ is left-happy then by \cref{thm:crossed-one-side}, similar to above, $s^*(\delta(S))+s(\delta(S))\geq 0$ if $\delta(S)_T$ is odd. Otherwise, let $S'$ be the union of the non-root atoms of $P$ and $F=\delta(S)\smallsetminus \delta(S')$. By \cref{lem:shared-edges}, we have $x(F)\geq 1-\eps_\eta/2$. Therefore, by (iii) of \cref{thm:payment-main} we have
$$ s(\delta(S))+s^*(\delta(S))\geq  s(A)+s(F)+s^-(C) \geq 0$$
as desired. 
Note that since $S$ is a leftmost cut, we always have $A\subseteq \delta(S)$. But $C$ may have an unpredictable intersection with $\delta(S)$; in particular, in the worst case only edges of $C$ with negative slack belong to $\delta(S)$. This is why we need to use $s^-(C)$ instead of $s(C)$. 
A similar argument holds when $S$ is the leftmost atom or a rightmost cut/atom.

	{\bf Type 5:} Cuts $S\in \cH\cap\cN_\eta$ where $\p(S)$ is a triangle $P$. This is similar to the previous case except we use \cref{lem:trianglereduction} to argue that the inequality is satisfied when $P$ is left/right happy.	
\end{proof}

\end{document}